\documentclass[twoside]{article}

\usepackage{amsmath, amssymb, graphicx, url, algorithm2e, amsthm}
\usepackage{mystyle}
\numberwithin{equation}{section}

\usepackage{chngcntr}
\usepackage{apptools}
\AtAppendix{\counterwithin{lemma}{section}}
\AtAppendix{\counterwithin{theorem}{section}}
\AtAppendix{\counterwithin{proposition}{section}}
\AtAppendix{\counterwithin{corollary}{section}}
\AtAppendix{\counterwithin{rem}{section}}

\usepackage[accepted]{aistats2019}

\usepackage[utf8]{inputenc} 
\usepackage[T1]{fontenc}    
\usepackage{hyperref}       
\usepackage{url}            
\usepackage{booktabs}       
\usepackage{amsfonts}       
\usepackage{nicefrac}       
\usepackage{microtype}      
\usepackage{algorithm2e}
\usepackage{subcaption}
\usepackage{enumitem}
\usepackage{wrapfig}
\usepackage{stmaryrd}
\usepackage{appendix}
\usepackage{setspace}
\usepackage{mystyle}
\usepackage{comment}
\usepackage[round]{natbib}

\begin{document}

%
\runningtitle{Support Localization and the Fisher Metric for off-the-grid Sparse Regularization}

%
\runningauthor{C. Poon, N. Keriven, G. Peyr\'e}

\twocolumn[

\aistatstitle{Support Localization and the Fisher Metric \\ for off-the-grid Sparse Regularization}

\aistatsauthor{ Clarice Poon$^1$ \And Nicolas Keriven$^2$ \And  Gabriel Peyr\'e$^2$ }

\aistatsaddress{ $^1$Centre for Mathematical Sciences\\
University of Cambridge\\
Wilberforce Rd, Cambridge, United Kingdom
 \And  $^2$D\'epartement de Math\'ematiques et Applications \\
\'Ecole Normale Sup\'erieure \\
45 rue d'Ulm, Paris, France } ]


\begin{abstract}
Sparse regularization is a central technique for both machine learning (to achieve supervised features selection or unsupervised mixture learning) and imaging sciences (to achieve super-resolution). Existing performance guaranties  assume a separation of the spikes based on an ad-hoc (usually Euclidean) minimum distance condition, which ignores the geometry of the problem. In this article, we study the BLASSO (i.e. the off-the-grid version of $\ell^1$ LASSO regularization) and show that the Fisher-Rao distance is the natural way to ensure and quantify support recovery, since it preserves the invariance of the problem under reparameterization. 
We prove that under mild regularity and curvature conditions, stable support identification is achieved even in the presence of randomized sub-sampled observations (which is the case in compressed sensing or learning scenario). On deconvolution problems, which are translation invariant, this generalizes to the multi-dimensional setting existing results of the literature. For more complex translation-varying problems, such as Laplace transform inversion, this gives the first geometry-aware guarantees for sparse recovery. 
\end{abstract}

\section{Introduction}

\subsection{Sparse Regularization}

In this work, we consider the general problem of estimating an unknown Radon measure $\mu_0 \in \Mm(\Xx)$ defined over some metric space $\Xx$ (for instance $\Xx=\RR^d$ for a possibly large $d$) from a few number $m$ of randomized linear observations $y \in \CC^m$, Let 
$
\Phi: \Mm(\Xx) \mapsto \CC^m$ be defined by 
\begin{equation}\label{eq-measurements}
\Phi \mu \eqdef 
\frac{1}{\sqrt{m}}\pa{\int_\Xx \phi_{\om_k}(x) \dd \mu(x)}_{k=1}^m,  
\end{equation}
where  $(\om_1,\ldots,\om_m)$ are identically and independently distributed according to some probability distribution $\featdist(\om)$ on $\om \in \Om$, and for $\om \in \Om$, $\phi_\om : \Xx \rightarrow \CC$ is a continuous function, denoted $\phi_\om \in \Cder{}(\Xx)$.  We further assume that $\phi_\om(x)$ is normalized, that is
\begin{equation}\label{eq:norm_constr}
 \EE_\om[\abs{\varphi_\om(x)}^2]   = 1, \qquad \forall x\in \Xx.
\end{equation} 
The observations are $y = \Phi \mu_0 + w$, where $w\in \CC^m$ accounts for noise or modelling errors.  
Some representative examples of this setting include:
\begin{rs}
	\item \textit{Off-the-grid compressed sensing:} off-the-grid compressed sensing, initially introduced in the special case of 1-D Fourier measurements on $\Xx=\TT=\RR/\ZZ$ by~\citep{tang2013compressed}, corresponds exactly to measurements of the form~\eqref{eq-measurements}. This is a ``continuous'' analogous of the celebrated compressed sensing line of works~\citep{candes2006robust,donoho2006compressed}.

%
	\item \textit{Regression using a continuous dictionary:} given a set of $m$ training samples $(\om_k,y_k)_{k=1}^m$, one wants to predicts the values $y_k \in \RR$ from the features $\om_k \in \Om$  using a continuous dictionary of functions $\om \mapsto \phi_\om(x)$ (here $x \in \Xx$ parameterizes the dictionary), as  $y_k \approx \int_\Xx \phi_{\om_k}(x) \dd\mu(x)$. 
	A typical example, studied for instance by~\cite{bach2017breaking} is the case of neural networks with a single hidden layer made of an infinite number of neurons, where $\Om=\Xx=\RR^p$ and one uses ridge functions of the form $\phi_\om(x) = \psi(\dotp{x}{\om})$, for instance using the ReLu non-linearity $\psi(u)=\max(u,0)$.

	%
	\item \textit{Sketching mixtures:} the goal is estimate a (hopefully sparse) mixture of density probability distributions on some domain $\Tt$ of the form
		$\xi(t) = \sum_i a_i \xi_{x_i}(t)$ where the $(\xi_x)_{x \in \Xx}$ is a family of template densities, and $a_i \geq 0$, $\sum_i a_i=1$.
		Introducing the measure $\mu_0 = \sum_i a_i \de_{x_i}$, this mixture model is conveniently re-written as
		$\xi(t) = \int_\Xx \xi_{x}(t) \dd\mu_0(x)$.
		The most studied example is the mixture of Gaussians, using (in 1-D for simplicity, $\Tt=\RR$) as 
		$\xi_x(t) \propto \si^{-1} e^{-\frac{(t-\tau)^2}{2\si^2}}$ where the parameter space is the mean and standard deviation $x=(\tau,\si) \in \Xx = \RR \times \RR^+$. 
		In a typical machine learning scenario, one does not have direct access to $\xi$ but rather to $n$ i.i.d. samples $(t_1,\ldots,t_n) \in \Tt^n$ drawn from $\xi$.
		Instead of recording this (possibly huge, specially when $\Tt$ is high dimensional) set of data, following~\cite{gribonval2017compressive}, one computes ``online'' a small set $y \in \CC^m$ of $m$ sketches against sketching functions $\theta_\om(t)$, that is, for $k=1,\ldots,m$,
		\begin{align*}
				 y_k \eqdef \frac{1}{n} \sum_{j=1}^n \theta_{\om_k}(t_j) 
				&\approx 
				\int_\Tt \theta_{\om_k}(t) \xi(t) \dd t.
				%
				\end{align*}
		These sketches exactly have the form~\eqref{eq-measurements} when defining the functions  
		$\phi_\om(x) \eqdef \int_\Tt  \theta_{\om}(t) \xi_{x}(t) \dd t$.
		A popular set of sketching functions, over $\Tt=\RR^d$ are Fourier atoms $\theta_\om(t) \eqdef e^{\imath\dotp{\om}{t}}$, for which  $\phi_{\cdot}(x)$ is the characteristic functions of $\xi_x$, which can generally be computed in closed form. 
\end{rs}

\paragraph{BLASSO.}
In all these applications, and many more, one is actually interested in recovering a discrete and $s$-sparse measure $\mu_0$ of the form 
$
	\mu_0=\sum_{i=1}^s a_i \de_{x_i}
	$ where $ (x_i,a_i) \in \Xx \times \CC.
$ 
%
An increasingly popular  method to estimate such a sparse measure corresponds to solving a infinite-dimensional analogous of the Lasso regression problem
\eql{\tag{$\Pp_\la(y)$}\label{eq-blasso}
	\umin{\mu \in \Mm(\Xx)} \frac{1}{2} \norm{\Phi \mu - y}_2^2 + \la |\mu|(\Xx).
}
Following~\cite{deCastro-exact2012}, we call this method the BLASSO (for Beurling-Lasso).
Here $|\mu|(\Xx)$ is the so-called total variation of the measure $\mu$, and is defined as
\eq{
	|\mu|(\Xx) \eqdef \sup \enscond{ \mathrm{Re}\dotp{f}{\mu} }{f \in \Cder{}(\Xx), \normi{f} \leq 1}.
}
Note that on unbounded $\Xx$, one needs to impose that $f$ vanishes at infinity.
If $\Xx=\{x_i\}_i$ is a finite space, then this corresponds to the classical finite-dimensional Lasso problem~\citep{tibshirani1996regression}, because $|\mu|(\Xx)=\norm{a}_1 \eqdef \sum_i |a_i|$ where $a_i=\mu(\{x_i\})$. Similarly, if $\Xx$ is possibly infinite but $\mu=\sum_i a_i \de_{x_i}$, one also has that $|\mu|(\Xx)=\norm{a}_1$. 

\paragraph{Previous Works.}

The BLASSO problem~\eqref{eq-blasso} was initially proposed by~\cite{deCastro-exact2012}, see also~\cite{bredies-inverse2013}.
The first sharp analysis of the solution of this problem is provided by~\cite{candes-towards2013} in the case of Fourier measurement on $\TT^d$. They show that if the spikes are separated enough, then $\mu_0$ is the unique solution of~\eqref{eq-blasso} when $w=0$ and $\la \rightarrow 0$. 
Robustness to noise under this separation condition is addressed in~\citep{candes-superresolution2013,fernandez-support2013,azais-spike2014}.
A refined stability results is detailed by~\cite{duval2015exact} which shows that conditions based on minimum separation imply \textit{support stability}, which means that  when $\norm{w}$ and $\norm{w}/\la$ are small enough, then the solution of~\eqref{eq-blasso} has the same number of Diracs as $\mu_0$, and that both the amplitudes and positions of the spikes converges smoothly as $w \rightarrow 0$. 
These initial works have been extended by~\cite{tang2013compressed} to the case of randomized compressive measurements of the form~\eqref{eq-measurements}, when using Fourier sketching functions $\phi_\om$.
In all these results, the separation condition are given for the Euclidean cases, which is an ad-hoc choice which does not take into account the geometry of the problem, and gives vastly sub-optimal theories for spatially varying operators (such as data-dependent kernels in supervised learning, Gaussian mixture estimation and Laplace transform in imaging, see Section~\ref{sec-fisher-metric}).



While this is not the topic of the present paper, note that for positive spikes, the separation condition is in some cases not needed, see for instance~\citep{schiebinger2015superresolution,2017-denoyelle-jafa}.
It is important to note that efficient algorithms have been developed to solve~\eqref{eq-blasso}, among which SDP relaxations for Fourier measurements~\citep{candes-superresolution2013} and Frank-Wolfe (also known as conditional gradient) schemes~\citep{bredies-inverse2013,boyd2017alternating}. 
Note also that while we focus here on variational convex approaches, alternative methods exist, in particular greedy algorithms~\citep{gribonval2017compressive} and (for Fourier measurements) Prony-type approaches~\citep{schmidt-multiple1986,roy1989esprit}. To the best of our knowledge, their theoretical analysis in the presence of noise is more involved, see however~\citep{liao-music2014} for an analysis of robustness to noise when a minimum separation holds.

\subsection{The Fisher information metric}
\label{sec-fisher-metric}
The empirial covariance operator is defined as $\Cov(x,x') \eqdef \frac{1}{m}\sum_i \overline{ \varphi_{\om_i}(x)}\varphi_{\om_i}(x')$ and the deterministic limit as $m \rightarrow +\infty$ is denoted $\fullCov$ with
\eql{\label{eq:kernel_rf}
	\fullCov(x,x') \eqdef \int_\Om {\overline{ \phi_\om(x)} \phi_\om(x')} \dd\featdist(\om).
}
Note that many covariance kernels can be written under the form \eqref{eq:kernel_rf}. By Bochner's theorem, this includes all translation-invariant kernels, for which possible features are $\phi_\om(x) = e^{\imath \om^\top x}$. The associated metric tensor is
\begin{equation}\label{eq:metric}
	\met_x \eqdef { \nabla_x\nabla_{x'}\fullCov(x,x)} \in \CC^{d\times d}.
\end{equation} 
Throughout, we assume that $\met_x$ is positive definite for all $x\in \Xx$. 
Then, $\met$ naturally induces a distance between points in our parameter space $\Xx$.
Given a piecewise smooth curve $\gamma:[0,1] \to \Xx$, the length $\ell_\met[\gamma]$ of $\gamma$ is defined by 
$\ell_\met[\gamma] \eqdef \int_0^1 \sqrt{ \dotp{\met_{\gamma(t)} \gamma'(t)}{\gamma'(t)} } \mathrm{d}t$.
Given two points $x,x'\in \Xx$, the distance from $x$ to $x'$, induced by $\met$ is $d_\met(x,x') \eqdef\inf_{\gamma\in \Ff} \ell_\met[\gamma] $ where $\Ff$ is the set of all piecewise smooth paths $\gamma:[0,1]\to \Xx$ with $\gamma(0) = x$ and $\gamma(1) = x'$. 
 
The metric $\met$ is closely linked to the Fisher information matrix~\citep{fisher1925theory} associated with $\Phi$: since \eqref{eq:norm_constr} holds,  $f(x,\omega) \eqdef \abs{\varphi_\om(x)}^2$ can be interpreted as a probability density function for the random variable $\omega$ conditional on parameter $x$, and  the metric $\met_x$ is equal (up to rescaling) to its Fisher information matrix, since
\begin{align*}
&\int \nabla\pa{\log f(x,\omega)} \nabla\pa{\log f(x,\omega) }^\top f(x,\omega) \mathrm{d}\Lambda(\om)\\
&= 4 \; \EE_\om [ \rep{\overline{\nabla \varphi_\om(x) } \nabla \varphi_{\omega}(x)^\top }] = 4 \met_x.
\end{align*} 
The distance $d_\met$ is called the ``Fisher-Rao'' geodesic distance~\citep{rao1945information} and is used extensively in information geometry for estimation and learning problems on parametric families of distributions~\citep{amari2007methods}. 
The Fisher-Rao is the unique Riemannian metric on a statistical manifold~\citep{cencov2000statistical} and it is invariant to reparameterization, which matches the invariance of the BLASSO problem~\eqref{eq-blasso} to reparameterization of the space $\Xx$. 
Although $d_\met$ has been used in conjunction with kernel methods (see for instance~\cite{burges1999geometry}), to the best of our knowledge, it is the first time this metric is put forward to analyze the performance of off-the-grid sparse recovery problems. 

%
%
%
\subsubsection{Examples} \label{sec:examples}

We  detail some popular  learning and imaging examples. 
\paragraph{The Fej\'er kernel} One of the first seminal result of super-resolution with sparse regularization was given by \cite{candes-towards2013} for this kernel, which corresponds to discrete Fourier measurements on the torus. We give a multi-dimensional generalization of this result here.
Let $f_c\in \NN$, $\Xx\in \TT^d$, $\Omega = \enscond{\om\in \ZZ^d}{\norm{\om}_\infty\leq f_c}$.  Let $\phi_\om(x) \eqdef  e^{\imath 2\pi \om^\top x}$ and $\Lambda(\om) \propto\prod_{j=1}^d g(\omega_j)$ where $g(j) = \frac{1}{f_c} \sum_{k=\max(j-f_c,-f_c)}^{\min(j+f_c,f_c)}(1-\abs{k/f_c})(1-\abs{(j-k)/f_c})$. Note that this corresponds to sampling  \textit{discrete} Fourier frequencies. Then,
the associated kernel is the Fej\'er kernel
$\fullCov(x,x') = \prod_{i=1}^d \kappa(x_i-x_i')$, 
where $\kappa(x)\eqdef \text{sinc}_{f_c/2+1}^4(x)$ where $\text{sinc}_{s}(x) \eqdef s^{-1} \sin(\pi s x)/\sin(\pi x)$, which has
a constant metric tensor $\met_x = C_{f_c} \Id$ and $d_\met(x,x') = \sqrt{C_{f_c}} \norm{x-x'}_2$ is a scaled Euclidean metric (quotiented by the action of translation modulo 1 on $\TT^d$), where $C_{f_c} = -\kappa''(0) = \tfrac{\pi^2 f_c(f_c+4)}{3}$. 

\paragraph{The Gaussian kernel} 
Let $\Sigma \in \RR^{d\times d}$ be a positive semidefinite matrix, $\Xx\subseteq \RR^d$ and $\Omega = \RR^d$. Let $\phi_\om(x) = e^{\mathrm{i}\omega^\top x}$ and $\Lambda(\om) = \Nn(0,\Sigma^{-1})$, the centered Gaussian distribution with covariance $\Sigma^{-1}$. This can be interpreted as sampling \textit{continuous} Fourier frequencies. Then, the associated kernel is 
$
\fullCov(x,x')  = e^{-\frac12 \normmah{\Sigma^{-1}}{x-x'}^2}
$
where $\normmah{\Sigma}{x} = \sqrt{x^\top \Sigma x}$, with constant metric $\met_x = \Sigma^{-1}$, and $d_\met(x,x') = \normmah{\Sigma^{-1}}{x-x'}$. In Section~\ref{sec:main}, we also detail how to exploit this kernel for Gaussian Mixture Model (GMM) estimation with the BLASSO. 

\paragraph{The Laplace transform}
Let $\bar \alpha = (\alpha_j)\in \RR_+^d$, $\Xx \subseteq (0,+\infty)^d$ and $\Omega = \RR_+^d$. A (sampled) Laplace transform is defined by setting $\phi_\om(x) = \prod_{i=1}^d\sqrt{\frac{2(x_i+\alpha_i)}{\alpha_i}}e^{- \dotp{ x}{\omega}}$ and $\Lambda(\om) = \prod_{j=1}^d (2\alpha_j) e^{-\dotp{2 \bar \alpha }{\omega}}$. Then, 
  $K(x,x') = \prod_{i=1}^d \kappa(x_i+\al_i, x_i'+\alpha_i)$ where $\kappa(a,b) = \tfrac{2\sqrt{ab}}{a+b}$,
 with metric $\met_x$ as the diagonal matrix with diagonal $\pa{(2(x_i+\alpha_i))^{-2}}_{i=1}^d$ and distance $d_\met(x,x') = \sqrt{\sum_i \abs{\log\pa{\tfrac{x_i+\alpha_i}{x_i'+\alpha_i}}}^2}$.
  We remark that this kernel, associated to the Laplace transform (which should not be confused with the translation-invariant Laplace kernel $\exp(-\norm{x-x'})$) appears in some microscopy imaging technique, see for instance~\cite{boulanger2014fast}. Unlike the previous examples, it is not translation-invariant, and therefore the metric $\met_x$ is not constant. Our results show that the corresponding Fisher metric is the natural way to impose the separation condition in super-resolution.
  
\subsection{Contributions.}
Our main contribution is Theorem~\ref{thm:main}, which states that if the sought after spikes positions $X_0$ are sufficiently separated with respect to the Fisher distance $d_\met$, then the solution to~\eqref{eq-blasso} is \textit{support stable} (that is, the solution of the BLASSO is formed of exactly $s$ Diracs) provided that the number of random noisy measurements $m$ is, up to log factors and under the assumption of random signs of the amplitudes $a_0$, linear in $s$, and the noise level $\norm{w}$ is less than $1/s$.
 In the case of translation invariant kernels, this generalizes existing results to a large class of multi-dimensional kernels, and also provides for the first time a quantitative bounds on the impact of the noise and sub-sampling on the spikes positions and amplitudes errors. For non-translation kernels, this provides for the first time a meaningful support recovery guarantee, a typical example being the Laplace kernel (see Section~\ref{sec-fisher-metric}). 



\section{Key concepts}

\paragraph{Notation for derivatives.}


Given $f\in\Cder{\infty}(\Xx)$, by interpreting the
$r^{th}$ derivative as a multilinear map: $\nabla^r f : (\CC^d)^r \to \CC$, so given $Q\eqdef \{q_\ell\}_{\ell=1}^r \in (\CC^d)^r$,
$$
\nabla^r f[Q] = \sum_{i_1,\cdots, i_r} \partial_{i_1}\cdots \partial_{i_r} f(x) q_{1,i_1} \cdots q_{r, i_r}.
$$
and we define the $r^{th}$ normalized derivative of $f$ as 
\begin{align*}
\diff{r}{f}(x)[Q] \eqdef \nabla^r f(x)[\{\met_x^{-\frac12} q_i\}_{i=1}^r]
\end{align*}
with norm
$
\norm{\diff{r}{f}(x)} \eqdef  \sup_{\forall \ell, \norm{q_\ell}\leq 1} \abs{\diff{r}{f}(x)[Q] }
$.
For $i,j\in \{0,1,2\}$, let  $\fullCov^{(ij)}(x,x')$ be a ``bi''-multilinear map, defined  for $Q\in (\CC^d)^i$ and $V\in (\CC^d)^j$ as \[
[Q]\fullCov^{(ij)}(x,x')[V] \eqdef \EE[ \overline{\diff{i}{\phi_\om}(x)[Q]} {\diff{j}{\phi_\om}(x')[V]}]\] 
and
$
\norm{\fullCov^{(ij)}(x,x') } \eqdef \sup_{Q,V} \norm{[Q]\fullCov^{(ij)}(x,x')[V]}
$
where the supremum is defined over all $Q\eqdef \{q_\ell\}_{\ell=1}^i$, $V\eqdef \{v_\ell\}_{\ell=1}^j$ with $ \norm{q_\ell}\leq 1$, $\norm{v_\ell}\leq 1$. Note that $\diff{2}{f}(x)$ and $\fullCov^{(02)}(x,x')$ can also be interpreted as a matrix in $\CC^{d\times d}$, and we have the normalization $\fullCov^{(02)}(x,x) = -\Id$ for all $x$. 


\subsection{Admissible kernel and separation}
In previous studies on the recovery properties of \eqref{eq-blasso}~\citep{candes-towards2013,bhaskar2013atomic,bendory2016robust,
duval2015exact,fernandez2016super},  
 recovery bounds are attained in the context of $\fullCov$ being \textit{admissible}  and a separation condition on the underlying positions $\{x_j\}_j$.  Namely, given $X = \{x_j\}_j$, that 
$\min_{i\neq j} d_\met(x_i,x_j)$ is sufficiently large with respect to the decay properties of $\fullCov$. 
For example, in the case where $\Phi$ corresponds to Fourier sampling on a grid, up to frequency $f_c$, this separation condition is $\min_{j\neq \ell} \norm{x_j - x_\ell}_2 \gtrsim 1/f_c$. 
In fact, if $\sign(a_j)$ can take arbitrary values in $\{+1,-1\}$, this  separation condition is a necessary to ensure exact recovery for the BLASSO~\citep{tang2015resolution}.

Following the aforementioned works, we introduce the notion of an admissible kernel.
\begin{definition}\label{def:admiss}
A kernel $\fullCov$ will be said \emph{admissible} with respect to $\paramKernel \eqdef \ens{\rnear,\Delta,\constker_i,B_{ij},s_{\max}}$, where $0< \rnear < \Delta/4$ is a neighborhood size, $\constker_0\in (0,1),~\constker_2\in(0,\rnear^{-2})$ are respectively a distance to $1$ and a curvature, $\Delta>0$ is a minimal separation, $B_{ij}>0$ for $i,j = 0,\ldots,2$ are some constants and $s_{\max}\in\NN^*$ is a maximal sparsity level, if 
\begin{enumerate}[leftmargin=*,itemsep=0pt]
\item \textbf{Uniform bounds}:  For $(i,j)\in \{(0,0), (1,0)\}$, $\sup_{x,x'\in \Xx}\snorm{\fullCov^{(ij)}(x,x')} \leq B_{ij}\, 
$; for $(i,j) \in \{ (0,2), (1,1), (1,2)\}$ and all $x,x'$ such that $\dsep(x,x')\leq \rnear$ or $\dsep(x,x')>\Delta/4$, $\snorm{\fullCov^{(ij)}(x,x')} \leq B_{ij}$; and finally,
$\sup_{x\in \Xx} \norm{\fullCov^{(22)}(x,x)}\leq B_{22}$.

\item \textbf{Neighborhood of each point}: For all $x\in \Xx$, $\fullCov(x,x) = 1$ 
and for all $x,x'\in \Xx$ with $\dsep(x,x')\leq \rnear$, $\rep{\fullCov^{(02)}(x,x')}\preccurlyeq -\constker_2 \Id$ and  $\norm{\imp{\fullCov^{(02)}(x,x')}} \leq c\constker_2$,
where $c\eqdef \frac{1}{2} \sqrt{\tfrac{2-\constker_2 \rnear^2}{\constker_2 \rnear^2}}$
and for $\dsep(x,x')\geq \rnear$, $\abs{\fullCov(x,x')}\leq 1-\constker_0$.
\item \textbf{Separation}: For $\dsep(x,x')\geq \Delta/4$, for all $i,j \in \{0,\ldots,2\}$ with $i+j\leq 3$, $\snorm{\fullCov^{(ij)}(x,x')} \leq \frac{h}{s_{\max}}$,
where $h \eqdef \min_{i\in \ens{0,2}}\pa{ {\tfrac{\constker_i}{32B_{1i}+32}}, \; \tfrac{5\constker_2 }{16 B_{12} + 24}  }$.
\end{enumerate}
Additionally, there exists $\Cmetrictensor\geq 0$ such that for $\dsep(x,x_0) \leq r_\textup{near}$:
$\norm{\Id - \met_{x_{0}}^{-\frac12}\met_{x}^{\frac12}} \leq \Cmetrictensor d_\met(x,x_0)$.
 We also denote $d_\met(X,X_0) = \sqrt{\sum_i d_\met(x_i,x_{0,i})^2}$ and $B\eqdef \sum_{i+j\leq 3} B_{ij}$ and $\constker \eqdef \min\{\constker_0,\constker_2\}$.
\end{definition}

Intuitively, these three conditions express the following facts: 1) the kernel and its derivatives are uniformly bounded, 2) near $x=x'$, the kernel has negative curvature, and otherwise it is strictly less than $1$, and 3) for $x$ and $x'$ sufficiently separated, the kernel and all its derivatives have a small value.

\subsection{Almost bounded random features}

Ideally, we would like our features and its derivatives to be uniformly bounded for all $\om$. However this may not be the case: think of $e^{i\om^\top x}$ where the support of the distribution $\Lambda$ is not bounded. Hence our results will be dependent on the probability that the derivatives are greater than some value $T$ decays sufficiently quickly as $T$ increases.
In the following, for $r\in \{0,1,2,3\}$,
$
L_r(\om) \eqdef \sup_{x\in \Xx} \norm{\diff{r}{\varphi_\om}(x)},
$ and let $F_r$ be such that $
\PP_\om \pa{L_r(\om) > t} \leq F_r(t).
$

\subsection{Key assumptions}\label{sec:assumption}
Our main result will be valid under the following assumptions.

\paragraph{I. On the domain and limit kernel} Let $\Xx$ be a compact domain with radius {$R_\Xx\eqdef \sup_{x,x'\in \Xx} \dsep(x,x')$}. Assume the kernel is admissible wrt $\paramKernel \eqdef \ens{\rnear,\Delta,\constker_i,B_{ij},s_{\max}}$.

\paragraph{II. Assumption on the underlying signal}
For $s\leq s_{\max}$, let $a_0 \in \CC^s$ and let $X_0 \eqdef (x_{0,j})_{j=0}^s$ be such that  $\dsep(x_{0,i},x_{0,j})\geq \Delta$ for $i \neq j$. The underlying measure is assumed to be $\mu_0 = \sum_{j=1}^s a_{0,j} \delta_{x_{0,j}}$. 


\paragraph{III. Assumption on the sampling complexity}
  For $\rho>0$,
suppose that $m\in \NN$ and $\{\Lu_i\}_{i=0}^3 \in \RR_+^4$ are chosen such that 
\begin{equation}\label{eq:stoc_lip_bd}
\begin{split}
\sum_{j=0}^3 F_j(\Lu_j) \leq \frac{\rho}{m}, \qquad\qandq\\
\max_{j=0}^3 \{  \Lu_j^2 \sum_{i=0}^3 F_i({\Lu_i}) + 2\int_{\Lu_j}^\infty t F_j({t}) \mathrm{d}t \} \leq \frac{\constker}{m}, 
\end{split}
\end{equation}

and  either one of the following hold:
\begin{equation}\label{eq:samp_ran}
m \gtrsim C \cdot   s  \cdot    \log\pa{N^d/\rho} \log\pa{s/\rho},  
\end{equation}
\begin{equation}\label{eq:samp_arb}
\text{or}\quad m \gtrsim  C \cdot   s^{3/2}  \cdot    \log\pa{N^d/\rho},
\end{equation}
where $C \eqdef  \constker^{-2} (\Lu_2^2 B_{11} + \Lu_1^2 B_{22} + (B_0+B_2)\Lu_{01}^2)$,
$  N \eqdef \LL_3 d R_\Xx(\rnear\constker)^{-1}$ and $\LL_r = \max_{i=1}^r \Lu_i$.

\begin{rem}
Our main theorem presents support stability guarantees under the sampling complexity rate \eqref{eq:samp_ran} if $\sign(a_0) = (a_{0,i}/\abs{a_{0,i}})_{i=1}^s$ forms a Steinhaus sequence, that is, iid uniformly distributed on the complex unit circle.  
This assumption has been used before in compressed sensing \citep{Candes2006c,tang2013compressed} to achieve this optimal complexity (see also \cite{Foucart2013}, Chap.~14). 
As noted in previous works, this random signs assumption is likely to be a proof artefact, however achieving optimal complexity without it may require more involved arguments \citep{candes2011probabilistic}.
When the signs are arbitrary, we prove our results under \eqref{eq:samp_arb}.  Although this $s^{3/2}$ scaling is still sub-optimal in $s$, we remark it improves upon the previous theoretical rate  of $s^2$ (up to log factors) \citep{li2017stable}. 
\end{rem}

\begin{rem}
The assumption on the choice of $\Lu_r$ ensures that with high probability, $\diff{r}{\phi_\om}(x)$ is uniformly bounded up to $r=3$.
Note also that, generally, the $\{\Lu_r\}$ depend on $m$, through \eqref{eq:stoc_lip_bd}. However, in all our examples: either a) $\sup_{x\in \Xx}\norm{\diff{r}{\phi_\om}(x)}$ are already uniformly bounded, in which case $\Lu_i$ can be chosen independently of $\rho$ and $m$ (for instance this is the case of the Fej\'er kernel); or b) the $F_r(t)$ are exponentially decaying, in which case we can show that $\Lu_r = \Oo(\log(m/\rho)^p)$ for some $p>0$, which only incurs additional logarithmic terms on the bounds \eqref{eq:samp_ran} and \eqref{eq:samp_arb}. This is the case of the Gaussian or Laplace transform kernel. 
\end{rem}


\section{Main result}\label{sec:main}

Our main theorem below states quantitative exact support recovery bounds under a minimum separation condition according to $d_\met$. 



\begin{theorem}\label{thm:main}
Let $\rho>0$, suppose that $\fullCov$ is admissible, and that $a_0$, $X_0$, $m$ and $\Lu_i$ satisfy the assumptions of Section~\ref{sec:assumption}.
Let $\Dd_{\la_0,c_0} \eqdef \enscond{(\la,w) \in \RR_+\times \CC^m}{ \la \leq \la_0, \; \norm{w} \leq c_0 \la}$ where $c_0 \sim  \min\pa{\tfrac{\constker_0}{\Lu_0},~\tfrac{\constker_2}{\Lu_2}}$ and $\la_0  \sim D/s$ with
\begin{equation}\label{eq:lip_eta_bound2}
D \eqdef \underline a \min\pa{{\rnear}{\sqrt{s}}, \; \tfrac{\constker \sqrt{s}}{\LL_2^2 \norm{a} }, \; \tfrac{\constker}{C_\met (B+ \LL_2^2)}  }
\end{equation}
where $\underline{a} = \min\{\abs{a_{0,i}}, \abs{a_{0,i}}^{-1}\}$.
Suppose that either $\sign(a_0)$ is a Steinhaus sequence and $m$ satisfies \eqref{eq:samp_ran} or $\sign(a_0)$ is an arbitrary sign sequence and $m$ satisfies \eqref{eq:samp_arb}. Then, with probability at least $1-\rho$, 
\begin{itemize}[leftmargin=*]
\item[(i)] for all $v\eqdef (\la,w) \in \Dd_{\la_0,c_0}$, 
\eqref{eq-blasso} has a unique solution which consists of exactly $s$ spikes. Moreover, up to a permutation of indices, the solution can be written as $\sum_{i=1}^s a^{v}_i \delta_{x^{v}_i}$, and $\sign(a^{v}_i) = \sign(a_{0,i})$ for all $i=1,\ldots, s$
\item[(ii)] The mapping
$
v  \in \Dd_{\la_0,c_0} \mapsto (a^v, X^v)
$
is $\Cder{1}$ and we have the error bound
\begin{equation}\label{eq:error-supp}
\norm{a^v-a_0} + d_\met(X^v,X_0) \leq \tfrac{\sqrt{s} (\lambda + \norm{w})}{\min_i \abs{a_{0,i}}} 
\end{equation}

\end{itemize}


\end{theorem}



We detail below the values relating to the sampling complexity corresponding to each of the examples detailed in Section \ref{sec:examples}. The corresponding proofs can be found in Section \ref{sec:app-examples} of the appendix.
\paragraph{Discrete Fourier sampling} The Fejer kernel of order $f_c\geq 128$ is admissible with $\Delta = \Oo(\sqrt{d} \sqrt[4]{s_{\max}})$, $\rnear =1/(8\sqrt{2})$, $\constker_0=0.00097$, $\constker_2=0.941$, $B_{01}=\Oo(d)$, $B_{11}=B_{02}=B_{12} = \Oo(1)$ and $B_{22}= \Oo(d)$. Moreover, $\Lu_r = \Oo(d^{r/2})$. 
Hence, up to logarithmic terms, Thm.~\ref{thm:main} is applicable with 
$
m = \Oo(s d^3)$
when the random signs assumption holds,
 and
$
m = \Oo(s^\frac{3}{2} d^3)
$
in the general case, with guaranteed support stability when $~\lambda = \Oo(s^{-1} d^{-2}),~\norm{w} = \Oo(s^{-1} d^{-3} )$. 
 Note that our choice of $\Delta$ imposes that $\norm{x_i - x_j}_2 \gtrsim \sqrt{d} s_{\max}^{1/4}/f_c$ whereas the previous result of \cite{candes-towards2013} requires  $\norm{x_i - x_j}_\infty \gtrsim C_d/f_c$ with no dependency in $s_{\max}$, however, their proof would imply that the constant $C_d$ grows exponentially in $d$. Since we are interested in having a general theory in arbitrary dimension, we have opted to present a polynomial dependency on $s_{\max}$. 

\paragraph{Continuous Gaussian Fourier sampling} 


%

 In the appendix we prove that the kernel is admissible  with
 $\Delta = \order{\sqrt{\log s_{\max}}}$,
  $\rnear = 1/\sqrt{2}$, $\constker_0 = 1-e^{-\frac{1}{4}}$, $\constker_2 = e^{-\frac{1}{4}}/2$, $B_{ij} = \Oo(1)$ for $i+j\leq 3$, $B_{22}= \Oo(d)$ and $\Lu_r = \pa{d + \log\pa{\frac{dm}{\rho}}^2}^\frac{r}{2}$ (as mentioned before, the dependence in $m$ only incurs additional logarithmic factors in \eqref{eq:samp_ran} and \eqref{eq:samp_arb}). Hence, up to log factors, the sample complexity and noise level for the application of Thm.~\ref{thm:main} is the same as for the Fej\'er kernel.


\paragraph{Laplace sampling}
The associated kernel is admissible with $\Delta = \order{d + \log(d s_{\max})}$, $\rnear =0.2$, $\constker_0=0.005$, $\constker_2=1.52$,  $B_{ij}= \Oo(1)$ for $i+j\leq 3$ and $B_{22}=\Oo(d)$. Define $\bar{R}_\Xx = \pa{1+ \tfrac{R_\Xx}{\min_i \alpha_i}}^d$ (where we recall that $R_\Xx$ is the radius of $\Xx$). Assuming for simplicity that all $\alpha_j$ are distinct, we can set
$\Lu_r =  \bar R_\Xx {(R_\Xx + \norm{\al}_\infty)^r \pa{\sqrt{d} +\max_i \frac{1}{\al_i} \log\pa{\tfrac{d \beta_i m \bar R_\Xx}{\rho \al_i}}  }^r}$
%
  Hence, choosing $\alpha_i \sim d$, we have that $\bar R_\Xx = \O(1)$ 
  and up to log factors, \eqref{eq:samp_ran} is $\Oo(sd^7)$ and \eqref{eq:samp_arb} is $\Oo(s^{3/2}d^7)$, and support stability is guaranteed when $\la = \Oo(s^{-1} d^{-3})$ and $\norm{w} = \Oo(s^{-1} d^{-5})$. Note that despite the stronger dependency on $d$, for practical applications (microscopy), one is typically only interested in the low dimensional setting of $d=2,3$.


\paragraph{Gaussian mixture learning} Consider $n$ datapoints $z_1,\ldots,z_n \in \RR^d$ drawn $iid$ from a mixture of Gaussians $\sum_i a_{0,i} \mathcal{N}(x_{0,i},\Sigma)$ with means $x_{0,i} \in \Xx \subset \RR^d$ and known covariance $\Sigma$, where $\Xx$ is bounded. Consider the following procedure:
\begin{rs}
\item draw $\om_j$ $iid$ from $\mathcal{N}(0,\Sigma^{-1}/d)$ (the $1/d$ normalization is necessary to avoid an exponential dependency in $d$ later on)
\item compute the generalized moments $y = \frac{1}{\sqrt{m}} \sum_{i=1}^n (e^{i\ps{\om_j}{x_i}})_{j=1}^m$
\item solve the BLASSO with features $\phi_\om(x) = e^{i\ps{\om}{x}} e^{-\frac12 \normmah{\Sigma}{\om}^2}$, to obtain a distribution $\tilde \mu$
\end{rs}
Then, as described in the introduction, we can interpret $y$ as noisy Fourier measurements of $\mu_0 = \sum_i a_{0,i} \delta_{x_{0,i}}$ \emph{in the space of means $\Xx$}, where the "noise" $w$ corresponds to using the empirical average over the $z_i$ instead of a true integration. It is easily bounded with probability $1-\rho$ by $\norm{w}\leq \order{\sqrt{\frac{\log(1/\rho)}{n}}}$, by a simple application of Hoeffding's inequality \citep{gribonval2017compressive}.

{
The associated kernel is the Gaussian kernel with covariance $(2+d)\Sigma$ and}  hence, our result states that, if $\normmah{\Sigma^{-1}}{x_i - x_j} \geq \sqrt{d \log s}$, and the number of measurements and sample complexity satisfy, up to logarithmic terms, $m = \order{s^\frac{3}{2} d^3}$, $n = \order{s^2 d^6/\min_i\abs{a_{0,i}}^2}$
and $\lambda_0 = \order{\tfrac{\min_i \abs{a_{0,i}}}{\sqrt{s}d^2 \norm{a_0}_2}}$ ,
 then, with probability $1-\rho$ on both samples $z_j$ and frequencies $\om_j$, the distribution $\tilde \mu$ is formed of exactly $s$ Diracs, and their positions and weights converge to the means and weights of the GMM. Let us give a few remarks on this result.

\emph{On model selection.} Besides convexity (with respect to the distribution of means) of the BLASSO, which is not the case of classical likelihood- or moments-based methods for learning GMM, the most striking feature of our approach is probably the support stability: with a sample complexity that is polynomial in $s$ and $d$, the BLASSO yields \emph{exactly} the right number of components for the GMM. Despite the huge literature on model selection for GMM, to our knowledge, this is one of the only result which is \emph{non-asymptotic} in sample complexity, as opposed to many approaches \citep{Roeder1997, Huang2013} which guarantee that the selected number of components approaches the correct one when the number of samples grows to infinity. 

\emph{On separation condition.} Our separation condition of $\sqrt{d\log s}$ is, up to the logarithmic term, similar to the $\sqrt{d}$ found in the seminal work by \cite{Dasgupta1999}. This was later improved by different methods \citep{Dasgupta2000, Vempala2004}, until the most recent results on the topic \citep{Moitra2010} show that it is possible to learn a GMM with \emph{no} separation condition, provided the sample complexity is exponential in $s$, which is a necessary condition \citep{Moitra2010}. As mentioned in the introduction, similar results exist for the BLASSO: \cite{2017-denoyelle-jafa} showed that in one dimension, one can identify $s$ positive spikes with no separation, provided the noise level is exponentially small with $s$. Hence learning GMM with the BLASSO and no separation condition may be feasible, which we leave for future work, however we note that the multi-dimensional case is still largely an open problem \citep{Poon2017}.

\emph{On known covariance.} An important path for future work is to handle arbitrary covariance. When the components all share the same mean and have diagonal covariance, the Fisher metric is related, up to a change of variables, to the Laplace transform kernel case treated earlier. When both means and covariance vary, in one dimension, the Fisher metric is related to the Poincar\'e half-plane metric \citep{Costa2015}. In the general case, it does not have a closed-form expression. We leave the treatment of these cases for future work.

\section{Sketch of proof}


\subsection{Background on dual certificates}\label{sec:dual}
Our approach to establishing that the solutions to \eqref{eq-blasso} are support stable is via the study of the associated dual solutions in accordance to the framework introduced in \cite{duval2015exact}. We first recall some  of their key ideas. In order to study the support stability properties of \eqref{eq-blasso} in the small noise regime, we consider the limit problem as $\la \to 0$ and $\norm{w}\to 0$, that is
\eql{\tag{$\Pp_0(y)$}\label{eq-blasso0}
	\umin{\mu \in \Mm(\Xx)} |\mu|(\Xx) \text{ subject to }  \Phi \mu = y.
}
The dual of \eqref{eq-blasso} and \eqref{eq-blasso0} are
\begin{align*}
\min_{p} \enscond{ \norm{y/\la-p}_2^2}{\norm{\Phi^* p}_\infty \leq 1} \tag{$\Dd_\la(y)$} \label{eq:dual}\\
\max_{p}\enscond{ \dotp{y}{p}}{\norm{\Phi^* p}_\infty \leq 1}. \tag{$\Dd_0(y)$} \label{eq:dual0}
\end{align*}
Any solution $\mu_\la$ of \eqref{eq-blasso} to related to the (unique) solution $p_\la$ of  \eqref{eq:dual} by $-p_\la = \frac{1}{\la}(\Phi \mu_\la - y)$ and writing $\eta_\la \eqdef \Phi^* p_\la$,
$
\dotp{\eta_\la}{\mu_\la} = \abs{\mu_\la}(\Xx)$.  
Note that $\mathrm{Supp}(\mu_\la)\subseteq \enscond{x\in \Xx}{\abs{\Phi^* p_\la(x)} = 1}$, so $\eta_\la$  ``certifies'' the support of $\mu_\la$ and is often referred to as a \textit{dual certificate}. Furthermore, by defining the minimal norm certificate $\eta_0$ as $\eta_0 \eqdef \Phi^* p_0$ where
\begin{equation}\label{eq:mnc}
 p_0 = \argmin\enscond{\norm{p}_2}{\text{$p$ is a solution to \eqref{eq:dual0}}}
\end{equation}
one can show that $p_\la$ converges as $\lambda \to 0$ to $p_0$ and hence $\eta_\la$ converges to $\eta_0\eqdef \Phi^* p_0$ in $L^\infty$. When $\la$ and $\norm{w}$ are sufficiently small, solutions  to \eqref{eq-blasso} are support stable provided that $\eta_0$ (called the minimal norm certificate) is \textit{nondegenerate}, that is $\eta_0(x_i) = \sign(a_i)$ for $i=1,\ldots, s$ and $ \nabla^2 \abs{\eta_0}^2(x_i)$ is negative definite. This is proven to be an almost sharp condition for support stability, since \cite{2017-Duval-IP-lasso} provided explicit examples where $\abs{\eta_0(x)} = 1$ for some $x\not\in \{x_i\}_i$ implies that \eqref{eq-blasso} recovers more than $s$ spikes under arbitrarily small noise. 

\paragraph{Pre-certificates}
In practice, the minimal norm certificate is hard to compute and analyse due to the nonlinear $\ell^\infty$ constraint in \eqref{eq:mnc}. So, one often introduces a proxy which can be computed in closed form by solving an linear system associated to the following least squares problem: $\eta_X \eqdef \Phi^* p$ where
\begin{equation}\label{eq:etaV}
\begin{split}
  p_X \eqdef \argmin \{\norm{p}_2 \; ; \; &(\Phi^* p)(x_i) = \sign(a_i), \\
  &\qquad \nabla (\Phi^* p)(x_i) = 0 \}.
  \end{split}
\end{equation}
Note that if $\eta_X$ satisfies $\norm{\eta_X}_\infty\leq 1$, then $\eta_X = \eta_0$.

\paragraph{Computation of $\eta_X$}
For  $x\in \Xx$, let $\phi(x) \eqdef \frac{1}{\sqrt{m}} \pa{\phi_{\om_k}(x)}_{k=1}^m$.
For $X = \{x_i\}_{i=1}^s$ we define $\Gamma_X:\CC^{s(d+1)}\to \CC^{m}$ as
$
\Gamma_X ([\alpha,\beta]) \eqdef \sum_{i=1}^s \alpha_i \varphi(x_i) + \nabla\varphi(x_i)^\top \beta_i
$
where $\nabla \phi \in \CC^{m\times d}$. 
Then, the minimizer of \eqref{eq:etaV} is $ p_X = \Gamma_X^{*,\dagger}  \binom{ \sign(a)}{ \mathbf{0}_{sd}} $.
Furthermore,  when $\Gamma_X$ is full rank, we can write $\subeta_X(x) \eqdef \sum_{i} \hat \alpha_i \Cov(x_i,x) + \dotp{\hat \beta_i}{\nabla_1 \Cov(x_i,x)},$ where $\hat \alpha_i\in \CC$, $\hat \beta_i\in \CC^d$ are such that $\binom{\hat \alpha}{\hat \beta} = (\Gamma_X^* \Gamma_X)^{-1}  \binom{\sign(a)}{0_{sd}}$, and the hat notation refers to the fact that we are using sub-sampled measurements. The limit precertificate is defined as
$\fulleta_X(x) \eqdef \sum_{i} \alpha_i \fullCov(x_i,x) + \dotp{\beta_i}{\nabla_1 \fullCov(x_i,x)},$ 
where $\binom{ \alpha}{ \beta} = (\EE[\Gamma_X^* \Gamma_X])^{-1}  \binom{\sign(a)}{0_{sd}}$.

The key to establishing our recovery results is to show that $\subeta_{X}$ is nondegenerate. In this paper, we will actually prove  a stronger notion of nondegeneracy: 
\begin{definition}
Let $a\in \CC^s$, $X = \{x_i\}_{i=1}^s\in \Xx^s$ for some $s\in \NN$, and $\epsilon_0,\epsilon_2,r>0$.
We say that $\eta\in \Cder{1}(\Xx)$ is $(\epsilon_0, \epsilon_{2})$-nondegenerate with respect to $a$, $X$ and $r$ if for all $i$, $\eta(x_i) = \sign(a_i),~\nabla \fulleta(x_i)=0 $ and
\begin{align*}
&\foralls x \in \Sf, \; |\eta(x)| \leq 1-\epsilon_0 \\
&\foralls x \in \Sn_j, \abs{\eta(x)} \leq 1- \epsilon_2 d_\met(x, x_j)^2
\end{align*}
where $\Sn_j \eqdef \enscond{x\in \Xx}{\dsep(x_i, x)\leq r }$ and $\Sf\eqdef \Xx\setminus \bigcup_{j=1}^s \Sn_j$.

\end{definition}

Our proof proceeds in three steps:
\begin{enumerate}[leftmargin=*, nolistsep]
\item Show that under admissibility of the kernel and sufficient separation, the limit precertificate $\fulleta_{X_0}$ is non-degenerate (see Theorem \ref{thm:NDetanew}).
\item Show that this non-degeneracy transfers to $\subeta_{X}$ when $m$ is large enough and $X$ is close to $X_0$. This is the purpose of Section \ref{sec:rand}.
\item As discussed, nondegeneracy of $\subeta_{X_0}$ automatically  guarantees support stability when $(\la,w)\in \Dd_{\la_0,c_0}$ for  $\la_0$ and $c_0$ \emph{sufficiently small}. To conclude we simply need to quantify $\la_0$ and $c_0$. This is the purpose of Section \ref{sec:quan}. In particular, given $(\lambda,w)$, we construct a candidate solution by means of (a quantitative version of) the Implicit Function Theorem, and show that it is indeed a true solution using the previous results.
\end{enumerate}

%
%

\subsection{Non-degeneracy of the limit certificate}

Our first result  shows that the ``limit precertificate" $\fulleta_{X_0} $ is nondegenerate:
\begin{theorem}\label{thm:NDetanew}
Assume the kernel is admissible wrt $\paramKernel$ (see Definition \ref{def:admiss}).
Then, for $s\leq s_{\max}$, for all $a=(a_j)_{j=1}^s \in \CC^s$ and $X = \{x_j\}_{j=1}^s \in \Xx^s$ such that $\dsep(x_i,x_j)\geq \Delta$, the function $\fulleta_{X_0}$ is $(\tfrac{\constker_0}{2}, \tfrac{\constker_2}{2})$-nondegenerate with respect to $a$, $X$ and $\rnear$.
\end{theorem}
The proof of this result can be found in Appendix \ref{sec:admiss}  and is a generalization of  the arguments of \cite{candes-towards2013} (see also \cite{bendory2016robust}).  We remark that unlike previous works which focus on translation invariant kernels, the Fisher metric provides a natural way to understand the required separation between the points in $X$ and thus open up the possibility of analysing more complex problems such as Laplace transform inversion.

%

\subsection{The randomized setting}\label{sec:rand}


For the remainder of this paper, we  consider solutions of \eqref{eq-blasso} given $y=\Phi\mu_{a_0,X_0}+w$ for some fixed $a_0\in \CC^s$ and $X_0\in \Xx^s$.
The following result shows that $\subeta_X$ is nondegenerate for all $X$ close to $X_0$:
\begin{theorem}\label{thm:main_lip_eta}
Let $\rho>0$. Under the assumptions of Section \ref{sec:assumption}, and assuming that either $m$ satisfies \eqref{eq:samp_ran} and $\sign(a_0)$ is a Steinhaus sequence, or $m$ satisfies \eqref{eq:samp_arb} and $\sign(a_0)$ is an arbitrary sign sequence, with probability at least $1-\rho$:
for all $X\in \Xx^s$ such that
\begin{equation}\label{eq:lip_eta_rad}
d_\met(X,X_0) \lesssim \min\pa{\rnear, \tfrac{\constker_r}{ C_\met \sqrt{s} \max\pa{ B,\Lu_{12} \Lu_{r}} }  },
\end{equation}
$\Gamma_X$ is full rank and
$\subeta_X$ is $(\constker_0/8, \constker_2/8)$-nondegenerate with respect to $a_0$, $X$ and $\rnear$.
\end{theorem}

The proof of this result is given in Appendix \ref{sec:nondegen_mnc}. We simply make a  remark on the proof here:
 We first prove that $\subeta_{X_0}$ is nondegenerate by bounding variations between $\fulleta_{X_0}$ and $\subeta_{X_0}$. The proof of this fact is a generalization of the arguments in \cite{tang2013compressed} to the multidimensional and general operator case.
We then exploit the fact the $\phi$ is smooth and hence, $\Gamma_{X}^* \Gamma_X$ satisfies certain Lipschitz properties with respect to $X$, to bound the local  variation between $\subeta_X$ and $\subeta_{X_0}$.  

\subsection{Quantitative support recovery}\label{sec:quan}
%

This final section concludes the proof of Theorem \ref{thm:main} by quantifying the regions for $\lambda$ and  $\norm{w}$ for which support stability is guaranteed.


\paragraph{Solution of the noisy BLASSO.} Let $\Phi_X: \CC^s\to \CC^m$ be defined by $\Phi_X a = \sum_{i=1}^s a_i \varphi(x_i)$.
Recall that $\mu_{a,X} = \sum_{i} a_i \delta_{x_i}$ is a solution to the BLASSO with $y = \Phi \mu_{a_0,X_0} + w$ if and only if
$
\subeta_\la = \Phi^* p_\la,$  with $ p_\la = \frac{1}{\la}(y - \Phi_X a)
$,
satisfies $\norm{\subeta_\la}_\infty\leq 1$ and $\subeta(x_j) = \sign(a_j)$. In that case, $p_\lambda$ is the \emph{unique} solution to the dual of the BLASSO. Moreover, if $\abs{\subeta_\lambda(x)} < 1$ for $x\neq x_i$ and $\Phi_X$ is full rank (which follows by Theorem \ref{thm:lip_eta}), then $\mu_{a,X}$ is also the unique solution of the primal.

\paragraph{Construction of a solution}

Following \cite{2017-denoyelle-jafa}, we define
 the function $f:\CC^s\times \Xx^s \times \RR_+ \times \CC^m$ by
$$
f(u,v) \eqdef \Gamma_X^* (\Phi_X a - \Phi_{X_0} a_0 - w) + \lambda \binom{\sign(a_0)}{0_{sd}}
$$
where $u = (a,X)$ and $v = (\la,w)$.  Observe that having $f(u,v) = 0$ ensures the existence of $\subeta_\la$ defined as above that satisfies $\subeta_\la(x_i) = \sign(a_{0,i})$ and $\nabla \subeta_ \la(x_i) = 0$. We will use it to construct a non-degenerate solution to \ref{eq:dual} for small $\lambda$ and $\norm{w}$.
Now, $f$ is continuously differentiable, with explicit forms of $\partial_v f(u,v)$ and $\partial_u f(u,v)$ given in \eqref{eq:ift_partialv} and \eqref{eq:ift_partialu} in the appendix, and in particular, 
letting $u_0 = (a_0, X_0)$, $\partial_u f(u_0,0) = \Gamma_{X_0}^* \Gamma_{X_0} J_a$, where $J_a$ is the diagonal matrix with $\binom{1}{a}\otimes 1_d\in \CC^{s(d+1)}$ along its diagonal and $\Gamma_{X_0}$ is full rank (with probability at least $1-\rho$) by Theorem \ref{thm:lip_eta}. So, $\partial_u f(u_0,0)$  is invertible and $f(u_0,0) = 0$. 
Hence, by the Implicit Function Theorem, there exists a neighbourhood $V$ of $0$ in $\CC \times \CC^m$, a neighbourhood  $U$ of $u_0$ in $\CC^s \times \Xx^s$ and a Fr\'echet differentiable function $g: V\to U$ such that for all $(u,v) \in U\times V$, $f(u,v) = 0$ if and only if $u = g(v)$. So, to establish support stability for \eqref{eq-blasso}, we simply need to estimate the size of the neighbourhood $V$ on which $g$ is well defined, and given $(\la,w) \in V$, for $(a,Z) = g((\la,w))$,  to check that the associated certificate $\subeta_{\la,w}\eqdef \Phi^*  p_{\la,w}$ with 
$
p_{\la,w} \eqdef \frac{1}{\la}\pa{\Phi_X a - \Phi_{X_0} a_0 - w}
$
is nondegenerate.

Indeed, one can  prove (see Theorem \ref{thm:appli_IFT}) that with probability at least $1-\rho$,  $V$ contains the ball $B_r(0)$ with radius
$
r \sim { \frac{1}{\sqrt{s}}\min\pa{\tfrac{\min\{\rnear,(C_\met B)^{-1} \}}{ \min_i \abs{a_{0,i}}},~\tfrac{1}{\Lu_{01} \Lu_{12}(1+\norm{a_0})} }}
$
 and given any $v\in B_r(0)$, $(a,X) = g(v)$ indeed satisfy the error bound \eqref{eq:error-supp}.

\paragraph{Checking that the candidate solution is a true solution}
It remains to check that $g(\la,w)$ defines a valid certificate and is non-degenerate (and hence, $\sum_i a_i \delta_{x_i}$ is the unique solution to \eqref{eq-blasso}) provided that $\lambda, w$ satisfy \eqref{eq:lip_eta_bound2}.
Given  $(\lambda,w) \in V$, let $(a,X) = g((\lambda,w))$. Define  $\subeta_{\lambda,w} \eqdef \frac{1}{\la} \Phi^* (\Phi_X a - \Phi_{X_0} a_0 - w)$ and following \cite{2017-denoyelle-jafa}, one can show that
\[
\subeta_{\lambda,w} = \subeta_X + \phi(\cdot)^\top\Pi_X \frac{w}{\lambda} + \frac{1}{\lambda}\phi(\cdot)^\top\Pi_X \Phi_{X_0} a_0
\]
where $\Pi_X$ is the orthogonal projection onto $\Im(\Gamma_X)^\perp$.

Note that since we have the error bound \eqref{eq:error-supp},  our choice of $\lambda$ and $\norm{w}$ ensures that \eqref{eq:lip_eta_rad} holds and hence, 
 Theorem \ref{thm:lip_eta} implies that
 $\subeta_X$ is nondegenerate with probablity at least $1-\rho$. To conclude, it is sufficient to show that the two remaining terms are sufficiently small, so 
 that $\subeta_{\lambda,w}$ remains non-degenerate. 
 Under $\Eve$, $\norm{\diff{r}{\phi_\om}(\cdot)} 
\leq \Lu_r$, and for any $z\in \CC^m$,
 $\norm{\diff{r}{\phi^\top z}{\cdot}} \leq \Lu_r \norm{z}$. Therefore, since $\Pi_X$ is a projection, we have $\norm{\diff{r}{\phi(\cdot)^\top\Pi_X \frac{w}{\lambda}}} \lesssim \constker_r$  when $\norm{w}/\la \lesssim \constker_r/\Lu_r$. 
Finally, since $\Phi_{X_0} a_0 = \sum_{j=1}^s \phi(x_{0,j})$, by Taylor expansion of $\phi(x_{0,j})$ around $x_j$ and applying $\Pi_X$ (see Lemma \ref{lem:bound_Pi} for this computation), we have
\begin{equation*}
\norm{\frac{1}{\lambda}\Pi_X\Gamma_{X_0} \binom{a_0}{0_{sd}}} \leq \frac{\Lu_2}{\lambda} 
\norm{a_{0}}_\infty d_H(X,X_0)^2.
\end{equation*}  
Since $g$ satisfies \eqref{eq:error-supp}  our choice of $\la_0 = \Oo(s^{-1})$ ensures that we can upper bound this by
$ \Lu_2\norm{a_{0}}_\infty \frac{s \pa{\la + \norm{w}^2/\la}}{\min \abs{a_{0,i}}^2} \lesssim \constker$ 
and consequently, $\frac{1}{\lambda} \norm{\diff{r}{\phi(\cdot)^\top\Pi_X \Phi_{X_0} a_0}} \lesssim \constker_r$.



\section*{Acknowledgements}
We would like to thank Ben Adcock for a helpful conversation regarding the stochastic gradient bounds. This work was partly funded by the CFM-ENS chair ``Mod\`eles et Sciences des donn\'ees'' and the European Research Council, NORIA project.

\small

\bibliographystyle{myabbrvnat}
\bibliography{biblio,biblio-nico}

\normalsize
\onecolumn
\appendix
\renewcommand{\d}{\ins{d}}



\section{Notations.}

In this section, we recall and introduce some notation which will be used throughout the appendix.

\paragraph{Block norms.} By default, $\norm{\cdot}$ is the Euclidean norm for vector and spectral norm for matrices. For a vector $x =[x_1,\ldots,x_s] \in \CC^{sd}$ formed of $s$ blocks $x_i \in \CC^d$, $1\leq i\leq s$, we define the block norm
\[
\normblock{x} \eqdef \sup_{1\leq i\leq s} \norm{x_i}_2
\]
For a vector $q = [q_1,\ldots,q_s,Q_1,\ldots,Q_s]\in \CC^{s(d+1)}$ decomposed such that $q_i\in \CC$ and $Q_i \in \CC^d$, we define
%
$$
\ns{q} \eqdef \max_{i=1}^s \ens{\abs{q_i}, \; \norm{Q_i}}.
$$

%

\paragraph{Kernel}
The empirical kernel is defined as
\[
\Cov(x,x') =  \frac{1}{m} \sum_{k=1}^m {\overline{\phi_{\om_k}(x)}  \phi_{\om_k} (x')}
\]
and the limit kernel is $\fullCov(x,x)\eqdef \EE_\om[\overline{\phi_{\om}(x)}  \phi_{\om} (x')]$.
The metric tensor associated to this kernel is
\[
\met_x \eqdef \EE_\om [{\overline{\nabla\phi_\om (x)} \nabla\phi_\om (x)^\top} ] \qquad 
\]
Given an event $E$, we write $\fullCov_E(x,x') \eqdef \EE_\om [\Cov(x,x') |E]$ to denote the conditional expectation on $E$.

\paragraph{Derivatives}

Given $f\in\Cder{\infty}(\Xx)$, by interpreting the
$r^{th}$ derivative as a multilinear map: $\nabla^r f : (\CC^d)^r \to \CC$, so given $Q\eqdef \{q_\ell\}_{\ell=1}^r \in (\CC^d)^r$,
$$
\nabla^r f[Q] = \sum_{i_1,\cdots, i_r} \partial_{i_1}\cdots \partial_{i_r} f(x) q_{1,i_1} \cdots q_{r, i_r}.
$$
and we define the $r^{th}$ normalized derivative of $f$ as 
\begin{align*}
\diff{r}{f}(x)[Q] \eqdef \nabla^r f(x)[\{\met_x^{-\frac12} q_i\}_{i=1}^r]
\end{align*}
with norm
$
\norm{\diff{r}{f}(x)} \eqdef  \sup_{\forall \ell, \norm{q_\ell}\leq 1} \abs{\diff{r}{f}(x)[Q] }
$.
We will sometimes make use the the multiarray interpretation:
$
\diff{0}{f} = f$, $ \diff{1}{f}(x) = \met_{x}^{-\frac12} \nabla f(x)\in \CC^d$, $ \diff{2}{f}(x) = \met_{x}^{-\frac12} \nabla^2 f(x) \met_{x}^{-\frac12}\in \CC^{d\times d}$.

For a bivariate function $\fullCov:\dom\times \dom \to \CC$, $\partial_{1,i}$ (resp. $\partial_{2,i}$) designates the derivative with respect to the $i^\textup{th}$ coordinate of the first variable (resp. second variable), and similarly  $\nabla_i$ and $\nabla_i^2$ denote the gradient and Hessian on the $i^\textup{th}$ coordinate respectively.

For $i,j\in \{0,1,2\}$, let  $\fullCov^{(ij)}(x,x')$ be a ``bi''-multilinear map, defined  for $Q\in (\CC^d)^i$ and $V\in (\CC^d)^j$ as \[
[Q]\fullCov^{(ij)}(x,x')[V] \eqdef \EE[ \overline{\diff{i}{\phi_\om}(x)[Q]} {\diff{j}{\phi_\om}(x')[V]}]\] 
and
$
\norm{\fullCov^{(ij)}(x,x') } \eqdef \sup_{Q,V} \norm{[Q]\fullCov^{(ij)}(x,x')[V]}
$
where the supremum is defined over all $Q\eqdef \{q_\ell\}_{\ell=1}^i$, $V\eqdef \{v_\ell\}_{\ell=1}^j$ with $ \norm{q_\ell}\leq 1$, $\norm{v_\ell}\leq 1$.

When $i + j \leq 2$, an equivalent definition is $\fullCov^{(ij)}(x,x') = \EE[ \overline{\diff{i}{\phi_\om}(x)} {\diff{j}{\phi_\om}(x')}^\top]$, and we note that  $\fullCov^{(00)} = \fullCov$, and we have normalized so that $\rep{\fullCov^{(11)}(x,x)} = -\rep{\fullCov^{(02)}(x,x)}$. Finally, we will make use of the still equivalent definition: $[q]\fullCov^{(12)}(x,x') = \EE[ \overline{q^\top \diff{1}{\phi_\om}(x)} {\diff{2}{\phi_\om}(x')}^\top] \in \CC^{d\times d}$.

\paragraph{Kernel constants}
For for $i,j\in \{(0,0), (0,1)\}$, define $B_{ij} \eqdef  \sup_{x,x'\in \Xx} \abs{\fullCov^{(ij)}(x,x')}$ , for $(i,j) \in \{(0,2), (1,2)\}$, 
$$
B_{ij} \eqdef \sup\enscond{ \norm{\fullCov^{(ij)}(x,x')}}{\dsep(x,x')\leq \rnear \text{ or } \dsep(x,x')>\Delta/2}.
$$
and define for $i=1,2$
$$
B_{ii} \eqdef \sup_{x\in \Xx} \norm{\fullCov^{(ii)}(x,x)}.
$$
%
%
%
For convenience, we define
\begin{equation}
\label{eq:Bi_def}
B_i \eqdef B_{0i} + B_{1i} + 1,\quad B \eqdef \sum_{\substack{i,j\in\{0,1,2\}\\
i+j\leq 3}} B_{ij} + 1.
\end{equation}

\paragraph{Matrices and vectors}
We will make use of the following vectors and matrices throughout:
Given $X \eqdef \{x_j\}_{j=1}^s \in \Xx^s$ and $a\in \CC^s$ which are always clear from context, define the vector $\RFVec_X(\om)\in \CC^{s(d+1)}$ as
\begin{equation}\label{eq:gamma_vec}
\RFVec_X(\om) \eqdef \pa{ \pa{\overline{\phi_\om(x_i)}}_{i=1}^s,\pa{\overline{\diff{1}{\phi_\om}(x_i)}^\top}_{i=1}^s}^\top,
\end{equation}
and
\begin{align*}
\etaMat_X &\eqdef \EE_\om [{\RFVec(\om)\RFVec(\om)^\adj}] \in \CC^{s(d+1) \times s(d+1)}\\
\etaFunc_X(x) &\eqdef \EE_\om [{\RFVec(\om) \phi_\om(x)}] \in \CC^{s(d+1)}\\
\etaCoeff &\eqdef \etaMat_X^{-1} \SignVecPad_s, \qquad \SignVecPad_s = \binom{\sign(a)}{0_{sd}}.
\end{align*}
Note that the diagonal of $\etaMat$ has only $1$'s.
For $\om_1,\ldots,\om_m$, we denote their empirical versions as:
\begin{align*}
&\subetaMat_X \eqdef \frac{1}{m} \sum_{k=1}^m {\RFVec(\om_k)\RFVec(\om_k)^\adj},\\
& \subetaFunc_X(x) \eqdef \frac{1}{m} \sum_{k=1}^m {\RFVec(\om_k) \phi_{\om_k}(x)},\quad \subetaCoeff \eqdef \subetaMat_X^{-1} \SignVecPad_s.
\end{align*}
which will serve us to construct our certificate, using the properties of their respective limit version. 

We remark that $\metg_X^{-1/2} \Gamma_X^*\Gamma_X \metg_X^{-1/2} = \subetaMat_X$, where $\Gamma_X$ is defined in the main paper and
\[
\metg_X =\left(\begin{matrix}
\Id_s & & & 0\\
& \met_{x_1} & & \\
& & \ddots & \\
0& & & \met_{x_s}
\end{matrix}\right)
\]
The vanishing derivative pre-certificate $\subeta_X$ is $\subetaCoeff^\top \subetaFunc_X(\cdot)$ and the limit pre-certificate is $\fulleta_X \eqdef  \etaCoeff^\top \etaFunc_X(\cdot)$.
When the set of points $X$ is clear from context, we will drop the subscript $X$ and write instead $\RFVec$, $\etaMat$, $\etaFunc$, $\fulleta$, and so on.

\paragraph{Metric induced distances}
Given $X= (x_j)_{j=1}^s \in \Xx^s$ and $X' = (x'_j)_{j=1}^s \in \Xx^s$, denote $d_\met(X,X') \eqdef \sqrt{\sum_j d_\met(x_j,x_j')^2}$.
Observe also that $\metg_X$ is positive definite for all $X$ and induces a metric on $\RR^{s}\times \Xx^{s}$ so that given $a,a'\in \RR^s$ and $X,X'\in \Xx^s$,
$$
d_{G}((a,X), (a',X')) = \sqrt{ \norm{a-a'}_2^2 + d_\met(X,X')^2}.
$$

\paragraph{Stochastic gradient bounds}
For $r\in \NN$,
$$
L_r(\om) = \sup_{x\in \Xx} \norm{\diff{r}{\varphi_\om}(x)},
$$
and $L_{ij}(\om) \eqdef\sqrt{ L_i(\om)^2 + L_j(\om)^2}$. For $i=0,1,2,3$, let $F_i$ be such that
$$
\PP_\om \pa{L_j(\om) > t} \leq F_i(t),
$$

Throughout, for $(\Lu_j)_{j=0}^3\in \RR_+^4$, the event $\Eve$ is defined as
\begin{equation}\label{eq:stoc_grad}
\Eve \eqdef \bigcap_{k=1}^m E_{\om_k} \qwhereq E_{\om} \eqdef \ens{ L_j(\om) \leq \Lu_j, \; \forall j=0,1,2,3}.
\end{equation}


\section{Proof of Theorem \ref{thm:NDetanew}}\label{sec:admiss}

In this section, we consider the (limit) vanishing derivative pre-certificate 
\[
\fulleta(x) = \SignVecPad^\top \etaMat_X^{-1}\etaFunc_X(x).
\]
Note that
\[
\diff{2}{\fulleta}(x) = \sum_{i=1}^s \etaCoeff_{1,i}\fullCov^{(02)}(x_i,x) + [\etaCoeff_{2,i}]\fullCov^{(12)}(x_i,x)
\]
where we have decomposed $\etaCoeff = [\etaCoeff_{1,1},\ldots,\etaCoeff_{1,s}, \etaCoeff_{2,1}, \ldots,\etaCoeff_{2,s}] \in \CC^{s(d+1)}$ where $\etaCoeff_{2,i}\in \CC^d$.

We aim to prove that $\eta$ is nondegenerate if $\fullCov$ is an admissible kernel.
Our first lemma shows that nondegeneracy of $\eta$ within each small neighbourhood of $x_i$ can be established by controlling the real and imaginary parts of $\diff{2}{\eta}$ in each small region:
\begin{lemma}\label{lem:lem_nondegen_cond}
Let $\epsilon>0$.
Let $a_0 \neq 0$, $x_0 \in \Xx$ and let $\sigma\in \CC$ be such that $\abs{\sigma} =1$. Suppose that $\eta \in \Cder{2}(\Xx; \CC)$ is such that
$\eta(x_0) = \sigma$, $\nabla \eta(x_0) = 0$ and 
$\rep{\overline{\sigma}\diff{2}{ \eta}(x_0)} \prec -\epsilon \Id $. Then, $\nabla^2 \abs{\eta}^2(x_0) \prec -2 \epsilon \Id$.
If in addition, we have $c, r>0$ with $\epsilon r<1$ and $c^2 \leq  (1-\epsilon r^2)/(\epsilon r^2)$ such that for all $x$ such that $d_\met(x,x_0) \leq r$,
\[
\rep{\overline{\sigma}\diff{2}{ \eta}(x)} \prec -\epsilon \Id \qandq\norm{\imp{\overline{\sigma}\diff{2}{ \eta}(x)}} \leq c \epsilon,
\]
then, $\abs{\eta(x)}^2 \leq 1- \epsilon^2 d_\met(x,x_0)^2$ for all $x$ such that
$  d_\met(x,x_0) \leq r$.

\end{lemma}
\begin{proof}
The first claim follows immediately from the computation: by writing $\eta = \eta_r(x) + \mathrm{i} \eta_i(x)$ where $\eta_i$ and $\eta_r$ are real valued functions,
$$
\frac{1}{2} \diff{2}{\abs{\eta}^2}  =  \rep{ \overline{\diff{1}{\eta}} \diff{1}{ \eta}^\top  + \diff{2}{\eta} \overline{\eta}},
$$
and evaluation at $x_0$ gives the required result.

Let $\gamma:[0,1]\to \Xx$ be a piecewise smooth path such that $\gamma(0) = x_0$, $\gamma(1) = x$.
\begin{align*}
 \eta(x) &= \eta(x_0) +  \int_0^1 (1-t) \dotp{\nabla^2 \eta(\gamma(t)) \gamma'(t) }{\gamma'(t)} \d t 
\\
&= \eta(x_0) +  \int_0^1 (1-t) \dotp{\diff{2}{ \eta}(\gamma(t)) \met^{\frac12}_{\gamma(t)}\gamma'(t) }{ \met^{\frac12}_{\gamma(t)} \gamma'(t)}  \d t
.
\end{align*}
So,
$$
\rep{\overline{\sign(a_0)} \eta(x)} = 1+ \inf_\gamma \rep{ \overline{\sign(a_0)} \int_0^1 (1-t) \dotp{\diff{2}{ \eta}(\gamma(t)) \met^{\frac12}_{\gamma(t)}\gamma'(t) }{ \met^{\frac12}_{\gamma(t)} \gamma'(t)} \d t } \leq 1- \epsilon d_\met(x,x')^2
$$ 
if we minimise over all paths from $x$ to $x_0$.
Similarly,
$$
\norm{\imp{\overline{\sign(a_0)} \eta(x)}} \leq c\epsilon d_\met(x,x_0)^2
$$
Therefore,
\begin{align*}
\abs{\eta(x)}^2 &\leq \abs{1-  \epsilon d_\met(x,x_0)^2}^2 + \abs{c\epsilon d_\met(x,x_0)^2}^2 \\
&\leq 1- 2\epsilon  d_\met(x,x_0)^2 + \epsilon^2  d_\met(x,x_0)^4 +  c^2 \epsilon^2  d_\met(x,x_0)^4 \\
&= 1-\epsilon d_\met(x,x_0)^2 - \epsilon d_\met(x,x_0)^2 \pa{1- \epsilon  d_\met(x,x_0)^2 \pa{1 +  c^2  }} \leq 1- \epsilon  d_\met(x,x_0)^2.
\end{align*}
\end{proof}

\begin{proof}[Proof of Theorem \ref{thm:NDetanew}] 
%
In order to show that $\fulleta$ is $(\constker_0/2, \constker_2/2)$-nondegenerate, it is enough to show that
\begin{align}
&\forall x\in \Sf, \quad \abs{\fulleta(x)} \leq 1- \constker_0/2 \label{eta:snear} 
\\
&\forall x\in \Sn, \quad \rep{\overline{\sign(a_j)}\diff{2}{ \eta}(x)} \prec -\frac{\constker_2}{2} \Id \qandq\norm{\imp{\overline{\sign(a_j)}\diff{2}{ \eta}(x)}} \leq  \frac{p}{4}\constker_2 \label{eta:sfar}
\end{align}
where $p = \sqrt{\frac{1-\constker_2 \rnear^2/2	}{\constker_2 \rnear^2/2}}$.

We first prove that 
the matrix $\etaMat$ is invertible. To this end, we write
\begin{equation}
\label{eq:Rdivide}
\etaMat = \pa{\begin{matrix}
\etaMat_0 & \etaMat_1^\top \\
\etaMat_1 & \etaMat_2
\end{matrix}}
\end{equation}
where $\etaMat_0 \eqdef (\fullCov(x_i,x_j))_{i,j=1}^s \in \CC^{s\times s}$, $\etaMat_1 \eqdef(\fullCov^{(10)}(x_i,x_j))_{i,j=1}^s \in \CC^{sd\times s}$,  and $\etaMat_2 \eqdef (\fullCov^{(11)}(x_i,x_j))_{i,j=1}^s \in \CC^{sd \times sd}$. By definition of $\fullCov^{(ij)}$, $\etaMat$ (and also $\etaMat_0$ and $\etaMat_2$) has only $1$'s on its diagonal. 

To prove the invertibility of $\etaMat$, we use the Schur complement of $\etaMat$, and in particular it suffices to prove that $\etaMat_2$ and the Schur complement  $\etaMat_S \eqdef \etaMat_0 - \etaMat_1 \etaMat_2^{-1} \etaMat_1^\top$ are both invertible. To show that $\etaMat_2$ is invertible, we define $A_{ij} = \fullCov^{(11)}(x_i,x_j)$. So $\etaMat_2$ has the form:
\[
\etaMat_2 = \pa{\begin{matrix}
\Id & A_{12} & \ldots &  A_{1s}  \\
A_{21} & \Id & \ddots & \vdots \\
\vdots & \ddots & \ddots & \vdots \\
A_{s1}& \ldots & \ldots & \Id
\end{matrix}}
\]
and by Lemma \ref{lem:block_norm}, we have
\begin{align*}
\normblock{\Id - \etaMat_2} \leq&~ \max_{i} \sum_j \norm{A_{ij}} \leq 1/4.
\end{align*}
Since $\normblock{\Id - \etaMat_2} <1$, $\etaMat_2$ is invertible, and we have $\normblock{\etaMat_2^{-1}} \leq \frac{1}{1 - \normblock{I - \etaMat_2}} \leq \frac{4}{3}$. Next, again with Lemma \ref{lem:block_norm}, we can bound
\begin{align*}
\norm{I - \etaMat_0}_\infty =&~ \max_{i} \sum_{j\neq i}\abs{\fullCov(x_i,x_j)} \leq\frac{\constker_0}{16} \\
\norm{\etaMat_1}_{\infty \to \textup{block}} \leq&~ \max_{i}\sum_j \norm{\fullCov^{(10)}(x_i,x_j)} \leq h \quad \text{since $\fullCov^{(10)}(x,x)=0$} \\
\norm{\etaMat_1^\top}_{\textup{block} \to \infty} \leq&~ \max_{i}\sum_j \norm{\fullCov^{(10)}(x_j,x_i)} \leq h
\end{align*}
Hence, we have
\begin{equation}
\norm{I - \etaMat_S}_\infty \leq \norm{I-\etaMat_0}_\infty + \norm{\etaMat_1^\top}_{\textup{block}\to \infty}\normblock{\etaMat_2^{-1}} \norm{\etaMat_1}_{\infty \to \textup{block}} \leq \frac{\constker_0}{16} + \frac{4}{3}h^2 \leq \frac{\constker_0}{8}
\end{equation}
since $h \leq \frac{\constker_0}{32}$. 
Therefore the Schur complement of $\etaMat$ is invertible and so is $\etaMat$.

\paragraph{Expression of $\fulleta$.} By definition, $\fulleta=$ satisfies $\fulleta(x_i) = \sign(a_i)$ and $\nabla \fulleta(x_i) = 0$.

We divide:
\[
\etaCoeff = \etaMat^{-1} \SignVecPad_s = \pa{\begin{matrix} \etaCoeff_1 \\ \etaCoeff_2 \end{matrix}}
\]
where $\etaCoeff_1\in \CC^s$ and $\etaCoeff_2\in \CC^{sd}$, and we denote $\etaCoeff_{2,i} \in \CC^{d}$ blocks such that $\etaCoeff_2 = [\etaCoeff_{2,1},\ldots,\etaCoeff_{2,s}]$. 

The Schur's complement of $\etaMat$ allows us to express $\etaCoeff_1$ and $\etaCoeff_2$ as
\begin{equation}
\etaCoeff_1 = \etaMat_S^{-1} \sign(a),\qquad \etaCoeff_2 = -\etaMat_2^{-1} \etaMat_1 \etaMat_S^{-1} \sign(a)
\end{equation}
and therefore we can bound
\begin{align}
\norm{\etaCoeff_1}_\infty \leq&~ \frac{1}{1-\constker_0/8} \label{eq:boundalpha1} \\
\normblock{\etaCoeff_2} \leq&~ \frac{8}{3}h\leq 4h \label{eq:boundalpha2}
\end{align}
Moreover, we have
\begin{equation}\label{eq:alphaCloseToSign}
\norm{\etaCoeff_1 - \sign(a)}_\infty \leq \norm{I - \etaMat_S^{-1}}_\infty \leq \norm{\etaMat_S^{-1}}_\infty \norm{I-\etaMat_S}_\infty\leq \frac{1}{4}
\end{equation}

\paragraph{Non-degeneracy.} We can now prove that $\fulleta$ is non-degenerate. 

Let $x$ be such that $\dsep(x_i,x)\leq \rnear$.
We need to prove that for all $x$ such that $d_\met(x,x_i) \leq r$,
$$\rep{\overline{\sign(a_i)} \diff{2}{ \eta}(x)} \prec -\frac{\constker_2}{2}\Id \qandq \norm{\imp{\overline{\sign(a_i)}\diff{2}{ \eta}(x)} } \leq   \frac{\constker_2}{2}\sqrt{\frac{2-\constker \rnear^2}{\constker_2 \rnear^2}}.
$$

 Then, since $\rnear\leq \Delta/2$ and the $x_i$'s are $\Delta$-separated, for all $j \neq i$ we have $\dsep(x,x_j)\geq \Delta/2$. Then, we have
\begin{align*}
\overline{\sign(a_i)} \diff{2}{\fulleta}(x) =&~ \overline{\sign(a_i)} \Bigg[\etaCoeff_{1,i}\fullCov^{(02)}(x_i,x) + \sum_{j \neq i} \etaCoeff_{1,j} \fullCov^{(02)}(x_j,x)\\
&\quad + [\etaCoeff_{2,i}] \fullCov^{(12)}(x_i,x) + \sum_{j \neq i} [\etaCoeff_{2,j}]\fullCov^{(12)}(x_j,x)\Bigg] 
\end{align*}

\begin{align*}
\rep{\overline{\sign(a_i)} \diff{2}{\fulleta}(x)} 
\preccurlyeq&~(1-\norm{\etaCoeff_1 - \sign(a)}_\infty)\rep{\fullCov^{(02)}(x_i,x)} + \norm{\etaCoeff_1}_\infty\sum_{j \neq i} \norm{\fullCov^{(02)}(x_j,x)}\Id\\ 
&\quad + \pa{\norm{\fullCov^{(12)}(x_i,x)} + \sum_{j \neq i} \norm{\fullCov^{(12)}(x_j,x)}}\normblock{\etaCoeff_{2}}\Id \\
\preccurlyeq&~ \Bigg(-\frac{3}{4}\constker_2 + \frac{1}{1-\constker_0/8}\frac{\constker_2}{16} + 4h(B_{12} + 1) \Bigg)\Id
\preccurlyeq \constker_2\pa{-\frac{3}{4} + \frac{1}{4}} \Id \preccurlyeq -\frac{\constker_2}{2} \Id\, .
\end{align*}
Taking the imaginary part, we have
\begin{align*}
&\norm{\imp{\overline{\sign(a_i)} \diff{2}{\fulleta}(x)}  }
\leq ~ (1+\norm{\etaCoeff_1 - \sign(a)})\norm{\imp{\fullCov^{(02)}(x_i,x)}} + \norm{\etaCoeff_1}_\infty\sum_{j \neq i} \norm{\fullCov^{(02)}(x_j,x)}\\ 
&\quad + \pa{\norm{\fullCov^{(12)}(x_i,x)} + \sum_{j \neq i} \norm{\fullCov^{(12)}(x_j,x)}}\normblock{\etaCoeff_{2}} \\
\leq &~ \Bigg(\frac{5c \constker_2 }{4}  + \frac{1}{(1-\constker_0/8)} h + 4h(B_{12} + 1) \Bigg) \leq\frac{5c \constker_2 }{4} + h\pa{4B_{12} +6} \leq \frac{\constker_2}{2}\sqrt{\frac{2-\constker \rnear^2}{\constker_2 \rnear^2}}.
\end{align*}
So, by Lemma \ref{lem:lem_nondegen_cond},  for each $i=1,\ldots, s$, $\abs{\eta(x)} \leq 1- \constker_2/2 \dsep(x,x_i)$  for all $x\in \Xx$ such that  $\dsep(x,x_i)\leq \rnear$.

Next, for any $x$ such that $\dsep(x, x_i)\geq \rnear$ for all $x_i$'s, we can say that there exists (at most) one index $i$ such that $\dsep(x,x_i)\geq \rnear$ and for all $j \neq i$ we have $\dsep(x,x_j)\geq \Delta/2$. We have
\begin{align*}
\abs{\fulleta(x)} =&~ \Bigg|\etaCoeff_{1,i} \fullCov(x_i,x) + \sum_{j \neq i} \etaCoeff_{1,j} \fullCov(x_j,x) \\
&\quad + \fullCov^{(10)}(x_i,x)^\top\etaCoeff_{2,i} + \sum_{j \neq i}\fullCov^{(10)}(x_j,x)^\top\etaCoeff_{2,j}\Bigg| \\
\leq&~ \norm{\etaCoeff_{1}}_\infty \pa{\abs{\fullCov(x_i,x)} + \sum_{j \neq i} \abs{\fullCov(x_j,x)}} \\
&\quad + \normblock{\etaCoeff_2}\pa{\norm{\fullCov^{(10)}(x_i,x)} + \sum_{j \neq i}\norm{\fullCov^{(10)}(x_j,x)}} \\
\leq&~ \frac{1-\constker_0 + \constker_0/16}{1-\constker_0/8} + 4h(B_{10} + 1) \leq 1-\frac{\constker_0}{2}\, .
\end{align*}
\end{proof}

\begin{rem}

Assuming that the derivatives of the kernel decay like a function $f(\norm{x-x'})$ when, there is always a separation $\Delta \propto f^{-1}(1/(C s_{\max})))$ such that the kernel is admissible. Ex: when $f = x^{-p}$, we have $\Delta \propto s_{\max}^{1/p}$ (eg Cauchy). When $f = e^{-x^p}$, we have $\Delta \propto \log^{1/p}(s_{\max})$ (eg Gaussian).
\end{rem}


\section{Preliminaries}
In this section, we present some preliminary results which will be used for proving our main results. We assume that $\fullCov$ is admissible, and given a set of points $X\in \Xx^s$, let $\Sn_j \eqdef \enscond{x\in \Xx}{\dsep(x,x_j)\leq \rnear}$, $\Sn \eqdef \bigcup_{j=1}^s \Sn_j$ and $\Sf \eqdef \Xx\setminus \Sn$.

\subsection{On the determistic kernel}
For an admissible kernel, we have the following additional bounds that will be handy.
\begin{lemma}\label{lem:additional_admissible}
Assume $\fullCov$ is an admissible kernel, let $X\in \Xx^s$ be  $\Delta$-separated points.
Then we have the following:
\begin{itemize}
\item[(i)] We have seen that $\etaMat$ is invertible. Additionally it satisfies
\begin{equation}
\norm{\Id-\etaMat}\leq \frac12 \qandq \ns{\Id-\etaMat}\leq \frac12. \label{eq:bound_etaMat_specnorm}
\end{equation}
\item[(ii)]  For any vector $q \in \CC^{s(d+1)}$ and  any $x\in \Xx^\textup{far}$, we have
\begin{align}
\norm{\etaFunc(x)} \leq B_0 \qandq  \abs{q^\top \etaFunc(x)} \leq B_0  \ns{q} \label{eq:Bf0}
\end{align}
\item[(iii)] For any vector $q \in \CC^{s(d+1)}$ and any $x \in \Xx^\textup{near}$ we have the bound:
\begin{align}
\norm{\diff{2}{q^\top \etaFunc(.)}(x)} \leq\norm{q}B_2\qandq \norm{\diff{2}{q^\top \etaFunc(.)}(x)} \leq \ns{q}B_2 \label{eq:Bf2}
\end{align}
\end{itemize}
\end{lemma}

\begin{proof}
We bound the spectral norm of $\Id - \etaMat$. Define $y\in \CC^{s(d+1)}$ decomposed as $y = [y_1,\ldots,y_s,Y_1,\ldots,Y_s]$ where $Y_i \in \RR^d$, such that $\norm{y}\leq 1$. We have
\begin{align*}
\norm{(\Id - \etaMat)y}^2 =&~ \sum_{i=1}^s \abs{\sum_{j\neq i} \fullCov(x_i,x_j)y_j + \sum_{j=1}^s \fullCov^{(10)}(x_i,x_j)^\top Y_j}^2 \\
&\qquad + \norm{\sum_{j} y_j\fullCov^{(10)}(x_i,x_j) + \sum_{j\neq i} \fullCov^{(11)}(x_i,x_j) Y_j}^2 \\
\leq&~ \sum_{i=1}^s \pa{\sum_{j\neq i} \abs{\fullCov(x_i,x_j)}\abs{y_j} + \sum_{j=1}^s \norm{\fullCov^{(10)}(x_i,x_j)}\norm{Y_j}}^2 \\
&\qquad + \pa{\sum_{j} \abs{y_j}\norm{\fullCov^{(10)}(x_i,x_j)} + \sum_{j\neq i} \norm{\fullCov^{(11)}(x_i,x_j)}\norm{ Y_j}}^2 \\
\leq&~\max_{\dsep(x,x')\geq \Delta}\pa{\abs{\fullCov(x,x')},\norm{\fullCov^{(10)}(x,x')},\norm{\fullCov^{(11)}(x,x')}}^2 \sum_i 2\pa{\sum_j \abs{y_j} + \norm{Y_j}}^2 \\
\leq&~4s^2\max_{\dsep(x,x')\geq \Delta}\pa{\abs{\fullCov(x,x')},\norm{\fullCov^{(10)}(x,x')},\norm{\fullCov^{(11)}(x,x')}}^2
\end{align*}
by Cauchy-Schwartz inequality and since $\fullCov^{(10)}(x,x)  = 0$ for all $x\in \Xx$. Since by hypothesis we have 
\[\max_{\dsep(x,x')\geq \Delta}\pa{\abs{\fullCov(x,x')},\norm{\fullCov^{(10)}(x,x')},\norm{\fullCov^{(11)}(x,x')}} \leq \frac{1}{4s_{\max}}\, ,
\]
we obtain
\begin{equation}
\norm{\Id - \etaMat}\leq \frac12
\end{equation}
and we deduce $(i)$. A near identical argument also yields $\ns{\etaMat - \Id} \leq \frac{1}{4}$.

For (ii), let  $x\in \Xx^\textup{far}$, then we have
\begin{align*}
\norm{\etaFunc(x)} &\leq \pa{\sum_{i=1}^s \abs{\fullCov(x_i,x)}^2 + \norm{\fullCov^{(10)}(x_i,x)}^2}^\frac12 \notag \\
&\leq \pa{B_{00}^2 + \frac{(s-1)\constker_0^2}{(16s_{\max})^2} +B_{10}^2 + \frac{(s-1)}{s_{\max}^2}}^\frac12\leq B_0
\end{align*}
for which, similar to the proof above, we have used the fact that $x$ is $\Delta/2$-separated from at least $s-1$ points $x_i$. 
Similarly, for any vector $q = [q_1,\ldots,q_s,Q_1,\ldots,Q_s]\in \CC^{s(d+1)}$  and any $x\in \Xx^\textup{far}$, we have
\begin{align*}
\norm{q^\top \etaFunc(x)} &\leq \sum_{i=1}^s \abs{q_i} \abs{\fullCov(x_i,x)} + \norm{Q_i}\norm{\fullCov^{(10)}(x_i,x)} \notag \\
&\leq \ns{q}\pa{ B_{00} + \frac{(s-1)\constker_0}{32 s_{\max})} +B_{10} + \frac{(s-1)\constker_0}{32 s_{\max}} }\leq B_0 \ns{q}.
\end{align*}

For any $x \in \Xx^\textup{near}$ we have the bound:
\begin{align*}
\norm{\diff{2}{q^\top \etaFunc}(x)} &= \norm{\sum_{i=1}^s q_i \fullCov^{(02)}(x_i,x) + [Q_i]\fullCov^{(12)}(x_i,x)} \notag\\
&\leq \norm{q}\pa{\sum_{i=1}^s \norm{\fullCov^{(02)}(x_i,x)}^2 + \norm{\fullCov^{(12)}(x_i,x)}^2}^\frac12 \notag \\
&\leq\norm{q}B_2
\end{align*}
and
\begin{align*}
\norm{\diff{2}{q^\top \etaFunc}(x)} &= \norm{\sum_{i=1}^s q_i \fullCov^{(02)}(x_i,x) + [Q_i]\fullCov^{(12)}(x_i,x)} \notag\\
&\leq \ns{q}\pa{\sum_{i=1}^s \norm{\fullCov^{(02)}(x_i,x)} + \norm{\fullCov^{(12)}(x_i,x)}} \notag \\
&\leq \ns{q}B_2
\end{align*}
\end{proof}

\subsection{Lipschitz bounds}

\begin{lemma}[Local Lipschitz constant of $\phi_\om$ and higher order derivatives]\label{lem:features_lip}
Suppose that $\norm{\diff{j}{\phi_\om}(x)} \leq \Lu_j$ for all $x\in \Xx$. For all $x,x'$ with $\dsep(x,x')\leq \rnear$, we have
\begin{itemize}
\item[(i)] $ \abs{\phi_\om(x) -  \phi_\om(x')} \leq \Ll_0 d_\met(x,x')$,
\item[(ii)] $\norm{\diff{1}{\phi_\om}(x) - \diff{1}{\phi_\om}(x')} \leq\Ll_1 d_\met(x,x')$,
\item[(iii)] $\norm{\diff{2}{\phi_\om}(x) - \diff{2}{\phi_\om}(x')} \leq \Ll_2 d_\met(x,x'),$
\end{itemize}
where $\Ll_0\eqdef \Lu_1$, $\Ll_1\eqdef \Lu_1 C_\met + \Lu_2 (1+C_\met \rnear) $ and $\Ll_2 \eqdef  \Lu_2 \pa{C_\met+C_\met^2 \rnear + 1}   + \Lu_3 (1+C_\met \rnear)^2$. As a consequence,
 for all $X= (x_j)$ and $X'=(x_j')$ such that $d_\met(x_j,x_j') \leq \rnear$,  we have
$$
\sup_{\norm{q}=1} \norm{\diff{r}{q^\top (\subetaFunc_{X} - \subetaFunc_{X'})}(y)}
\leq 
\Lu_r \sqrt{\Ll_0^2 + \Ll_1^2 }  d_\met(X,X').
$$
\end{lemma}
\begin{proof}
Let  $x,x'\in \Xx$ with $\dsep(x,x')\leq \rnear$.
Recall that $\norm{  \met_{x'}^{\frac12} \met_x^{-\frac12} -\Id} \leq C_\met d_\met(x,x')$, and so, $\norm{\met_{x'}^{\frac12} \met_x^{-\frac12}} \leq 1+ C_\met \rnear$.

 Let $p:[0,1]\to \Xx$ be a piecewise smooth path such that $p(0)=x'$, $p(1) = x$. Then, by Taylor's theorem,
\begin{equation}\label{eq:lip1}
\phi_\om(x) -  \phi_\om(x') = {\int_{t=0}^1 \dotp{\met_{p(t)}^{-\frac12} \nabla \phi_\om(p(t)) }{\met_{p(t)}^{\frac12}p'(t)}\d t}  \leq \Lu_1 \int_0^1 \norm{\met_{p(t)}^{\frac12}p'(t)} \mathrm{d}t
  \end{equation}
so taking the minimum over all paths $p$ yields 
$ \abs{\phi_\om(x) -  \phi_\om(x')} \leq \Lu_1 d_\met(x,x')$.

Given $q\in \RR^d$, by Taylor's theorem,
\begin{equation}\label{eq:lip2}
\begin{split}
&\diff{1}{\phi_\om}(x)[q] = \nabla\phi(x)[\met_x^{-\frac12} q]
= \nabla\phi(x')[\met_x^{-\frac12} q] + \int \nabla^2 \phi_\om(p(t))[\met_x^{-\frac12} q, p'(t)] \mathrm{d}t\\
&= \diff{1}{\phi_\om}(x')[q]+ \diff{1}{\phi_\om}(x')[ ( \met_{x'}^{\frac12} \met_x^{-\frac12} -\Id)q] + \int \diff{2}{\phi_\om}(p(t))[ \met_{p(t)}^{\frac12} \met_x^{-\frac12} q,  \met_{p(t)}^{\frac12}  p'(t)] \mathrm{d}t
\end{split}
\end{equation}
Therefore,
\[
\norm{\diff{1}{\phi_\om}(x) - \diff{1}{\phi_\om}(x')} \leq \Lu_1 C_\met d_\met(x,x') + \Lu_2 (1+C_\met \rnear) d_\met(x,x'). 
\]

Finally, for all $q_1,q_2\in \RR^d$, by Taylor's theorem
\begin{equation}\label{eq:lip3}
\begin{split}
&\diff{2}{\phi_\om}(x)[q_1,q_2] - \diff{2}{\phi_\om}(x')[q_1,q_2]\\
&=\nabla^2{\phi_\om}(x)[\met_x^{-\frac12} q_1,\met_x^{-\frac12} q_2] -  \nabla^2{\phi_\om}(x')[\met_{x'}^{-\frac12} q_1,\met_{x'}^{-\frac12} q_2]\\
&= \diff{2}{\phi_\om}(x')[\met_{x'}^{\frac12}\met_x^{-\frac12} q_1,(\met_{x'}^{\frac12} \met_x^{-\frac12}-\Id) q_2] + \diff{2}{\phi_\om}(x')[(\met_{x'}^{\frac12} \met_{x}^{-\frac12} -\Id) q_1,q_2] \\
&\qquad + \int 
\diff{3}{\phi_\om}(p(t))[\met_{p(t)}^{\frac12} \met_x^{-\frac12} q_1, \met_{p(t)}^{\frac12} \met_x^{-\frac12} q_2, \met_{p(t)}^{\frac12} p'(t)] \mathrm{d}t.
\end{split}
\end{equation}
Therefore,
$$
\norm{\diff{2}{\phi_\om}(x) - \diff{2}{\phi_\om}(x')} \leq \pa{ \Lu_2 \pa{(1+C_\met \rnear) C_\met + 1}   + \Lu_3 (1+C_\met \rnear)^2 } d_\met(x,x'). 
$$

By applying these Lipschitz bounds, we obtain
\begin{align*}
\sup_{\norm{q}=1} &\norm{\diff{r}{q^\top (\subetaFunc_{X} - \subetaFunc_{X'})}(y)} ^2\\
&
\leq  \sum_{j=1}^s \norm{ \Cov^{(0r)}(x_j,y) -  \Cov^{(0r)}(x_j',y)}^2 + \sum_{j=1}^s \norm{ \Cov^{(1r)}(x_j,y)-\Cov^{(1r)}(x_j',y)}^2
\\
&
\leq   \sum_{j=1}^s {\Ll_0^2 \Lu_r^2 d_\met(x_j,x_{j}')^2} + \sum_{j=1}^s \Ll_1^2 \Lu_r^2 d_\met(x_j,x_{j}')^2\\
&
= \pa{\Ll_0^2 + \Ll_1^2 } \Lu_r^2 d_\met(X,X')^2
\end{align*}

\end{proof}

\begin{lemma}[Local Lipschitz constant of $\Cov^{(ij)}$]\label{lem:lip_kernel}
Let $x_1,x_0\in \Xx$. Let $i,j\in \{0,1,2\}$ with $i+j\leq 3$. Define
\[
A_{ij} = \sup_{x} \norm{\Cov^{(ij)}(x,x_0)}
\]
where $x$ ranges over $\dsep(x,x_1) \leq r_\textup{near}$.
Then, for all $x$ such that $d_\met(x,x_1) \leq \rnear$,
\begin{align*}
\norm{\Cov^{(0j)}(x,x_0) - \Cov^{(0j)}(x_1,x_0)} &\leq A_{1j} d_\met(x,x_1) \\
\norm{\Cov^{(1j)}(x,x_0) - \Cov^{(1j)}(x_1,x_0)} &\leq \pa{\Cmetrictensor A_{1j} + (1+\Cmetrictensor \rnear) A_{2j}} d_\met(x,x_1)
\end{align*}
The same results hold if we replace $\Cov$ by $\fullCov$.
\end{lemma}
\begin{proof}
The Lipschitz bounds on $\Cov^{ij}$ follow
by combining 
\begin{align*}
&[q_1,\ldots,q_i](\Cov^{(ij)}(x,x_0) - \Cov^{(ij)}(x_1,x_0))[v_1,\ldots,v_j] \\
&= \hat{\EE}\rep{\overline{(\diff{i}{\phi_\om}(x) - \diff{i}{\phi_\om}(x_1))[q_1,\ldots,q_i]}\diff{j}{\phi_j}(x_0)[v_1,\ldots,v_j]}
\end{align*}
where $\hat{\EE}$ indicates either empirical expectation or true expectation with \eqref{eq:lip1}, \eqref{eq:lip2} and \eqref{eq:lip3}.

%
\end{proof}


\subsection{Probability bounds}

In the proof of our main results, we will often assume that event $\Eve$ (see \eqref{eq:stoc_grad}) holds since our assumptions in Section \ref{sec:assumption} imply that $\PP(\Eve^c) \leq \rho/m$. The following lemma shows that our assumptions also imply that $\EE_\om[L_i(\om)^2 1_{E_\om^c}] \leq \frac{\constker}{m}.$ and this is a condition which our proofs will often rely upon.

\begin{lemma}\label{lem:cdf_bd}
The following holds.
$\PP(E_{\om}^c)  \leq \sum_i F_i(\Lu_i)$
and
$$\EE_\om[ L_j(\om)^2 1_{E_\om^c}] \leq  2 \int_{\Lu_j}^\infty t F_j({t}) \mathrm{d}t + \Lu_j^2 \sum_{i} F_i({\Lu_i})
$$
\end{lemma}
\begin{proof}
Let $E_{\om,j}$ be the event that  $
L_r(\om) \leq \Lu_r$, so $E_\om = \cap_{j=0}^3 E_{\om,j}$.
By the union bound,
$\PP(E_{\om}^c) \leq \sum_j \PP(E_{\om,j}^c) \leq \sum_i F_i(\Lu_i)$.

For the second claim, observe that $E_\om^c = \cup_i E_{\om,i}^c$ so that $\EE[ L_j(\om)^2  1_{E_\om^c}] \leq \sum_i \EE[ L_j(\om)^2  1_{E_{\om,i}^c}]$ and we have
\begin{align*}
\EE[ L_j(\om)^2  1_{E_{\om,i}^c}]  &= \int_0^\infty \PP(L_j(\om)^2  1_{E_{\om,i}^c} \geq t) \d t \\
&= \int_0^\infty \PP\pa{(L_j(\om)^2 \geq t) \cap (L_i(\om) \geq \Lu_i)} \d t \\
&\leq \Lu_j^2 F_i(\Lu_i) + \int_{\Lu_j^2}^\infty F_j(\sqrt{t}) \d t = \Lu_j^2 F_i(\Lu_i) + 2\int_{\Lu_j}^\infty tF_j(t) \d t
\end{align*}
where we have bounded $\PP\pa{(L_j(\om)^2 \geq t) \cap (L_i(\om) \geq \Lu_i)}$ by respectively $\PP(L_i(\om) \geq \Lu_i) \leq F_i(\Lu_i)$ in the first term and by $\PP(L_j(\om)^2 \geq t) \leq F_j(\sqrt{t})$ in the second term. 
\end{proof}

%
%
%

\subsubsection{Concentration inequalities}

The following result is an adaption of the Matrix Bernstein inequality for dealing with conditional probabilities.

\begin{lemma}[Adapted unbounded Matrix Bernstein]\label{lem:adapt_bernstein}
Let $A_j \in \RR^{d_1 \times d_2}$ be a family of iid matrices for $j=1,\ldots,m$.
Let $Z = \frac{1}{m} \sum_{j=1}^m A_j$ and let $\bar Z = \EE[Z]$.
Let $t\in (0, 4\norm{\EE[A_1]}]$. Let events $E_j$ be independent events such that $E_j \subseteq \ens{\norm{A_j} \leq L}$ and let $E = \cap_j E_j$. Suppose that we have
\[
\PP(E_j^c ) \leq \frac{t}{t+4\norm{\EE[A_1]}} \qandq \EE[\norm{A_j} \bun_{E_j^c}] \leq \frac{t}{4}
\]
Then a first consequence is that we have $\EE_E[Z] = \EE_{E_j}[A_j]$ for all $j$ and $\norm{\EE[Z] - \EE_E[Z]} \leq \frac{t}{2}$.

Finally, assuming that 
\[
\sigma^2 \eqdef \max_j\{\norm{\EE_{E_j} [A_j A_j^\adj]},\norm{\EE_{E_j} [A_j^\adj A_j]}\} < \infty
\]
we have
\[
\PP_{E}\pa{\norm{Z- \EE[Z]}\geq t} \leq (d_1 + d_2)\exp\pa{-\frac{m t^2/4}{\sigma^2 + Lt/3}}.
\]
\end{lemma}

\begin{proof}
We first bound $\norm{ \EE[Z] - \EE_E[Z]}$.
First observe that $\EE[Z] = \EE_{E_1}[A_1]$ and  $\EE_{E} Z = \EE_{E_1} [A_1]$ since $A_j$ are iid. Moreover,
$$
\EE[A_1] = \EE[A_1 \bun_{E_1}] + \EE[A_1 \bun_{E_1^c}] = \EE[A_1| E_1] \PP(E_1) + \EE[ A_1 \bun_{E_1^c}].
$$
Hence,
\begin{align*}
&\norm{\EE[A_1] - \EE_{E_1}[A_1]} = \norm{ (P(E_1) -1) \EE_{E_1}[A_1] + \EE[ A_1 \bun_{E_1^c}]}\\
&
\leq \PP(E_1^c) \norm{\EE[A_1]} +  P(E_1^c)  \norm{\EE[A_1] - \EE_{E_1}[A_1]}   +  \EE[\norm{ A_1} \bun_{E_1^c}].
\end{align*}
Therefore,
$$
\norm{\EE[A_1] - \EE_{E_1}[A_1]}  \leq  \frac{P(E_1^c) \norm{\EE[A_1]}    +  \EE[\norm{ A_1} \bun_{E_1^c}]}{1-\PP(E_1^c)} \leq \frac{t}{2}
$$

For the second statement, 
\begin{align*}
\PP_{E}(\norm{Z - \EE[Z]}\geq t) 
&\leq \PP_{E}( \norm{Z - \EE_E[Z]}\geq t  - \norm{\EE[Z] - \EE_E[Z]})  \\
&\leq \PP_{E}( \norm{Z - \EE_E[Z]}\geq t/2).
\end{align*}

To conclude, we apply Bernstein's inequality (Lemma \ref{lem:bernstein_matrix}) to $Y_j = A_j - \EE[A_j | E] = Y_j = A_j - \EE[A_j | E_j]$ conditional to $E$. 
Observe that
$$
0\preceq \EE_E[Y_j Y_j^\top ] \preceq \EE_E[ A_j A_j^\top ] - \EE_E[ A_j]\EE_E[ A_j]^\top] \preceq \EE_E[ A_j A_j^\top ],
$$
which yields $\norm{\EE_E[Y_j Y_j^\top ]} \leq  \norm{ \EE[ A_j A_j^\top ]}$ and similarly, $\norm{\EE_E[Y_j^\top  Y_j ]} \leq  \norm{ \EE_E[ A_j^\top A_j ]}$. So
 by Bernstein's inequality
\[
\PP_{E}(\norm{Z-\EE_E[Z]}\geq t/2) \leq 2(d_1 + d_2)\exp\pa{-\frac{m t^2/4}{\sigma^2 + Lt/3}}.
\]
\end{proof}

\begin{corollary}
Let $x,x'\in \Xx$.
If
\[
\PP(E_{\om}^c ) \leq \frac{t}{t+4 \norm{\fullCov^{(ij)}(x,x')}} \qandq \EE[L_{ij}(\om) \bun_{E_\om^c}] \leq \frac{t}{4}
\]
then $\norm{\fullCov^{(ij)}_\Eve(x,x') - \fullCov^{(ij)}(x,x')} \leq t/2$.

\end{corollary}

\begin{proposition}\label{prop:bernstein_etaMat}
Let $t>0$ and assume that 
$$
\PP(E_\om^c) \leq \frac{t}{t+6} \qandq  \EE[ L_{01}(\om)^2 \bun_{E_\om^c}] \leq \frac{t}{4s}
$$
then $\norm{\etaMat - \etaMat_\Eve} \leq t/2$ and
\[
\PP_\Eve(\norm{\etaMat - \subetaMat}\geq t) \leq 4(d+1)s \exp\pa{-\frac{mt^2/4}{s\Lu_{01}^2(3 + t/3)}}
\]
Consequently,
\[
\PP_\Eve( \norm{\etaMat^{-1} - \subetaMat^{-1}}  \geq t) \leq 4(d+1)s \exp\pa{-\frac{mt^2}{16 s\Lu_{01}^2(3 + 2\tilde t)}}.
\]

\end{proposition}

\begin{proof}
We apply Lemma \ref{lem:adapt_bernstein} to $A_j = {\gamma(\om_j)\gamma(\om_j)^\adj}$ with the following observations:
\begin{itemize}
\item for each $\om$,
$$
\norm{\gamma(\om)\gamma(\om)^\adj} \leq \norm{\gamma(\om)}^2 \leq s \max_{x\in \Xx} \{ \norm{\diff{1}{\varphi_\om}(x)}^2 + \abs{\varphi_\om(x)}^2 \},
$$
so under event $\Eve$, $\norm{A_j}  \leq s\Lu_{01}^2$.
\item By Lemma \ref{lem:additional_admissible},  $\norm{\EE[A_j]} = \norm{\etaMat} \leq 3/2$,
\item We may set  $\sigma^2 = \Lu_{01} (3/2+t/2)$ since
\begin{align*}
0\preceq \EE_\Eve[A_1 A_1^\adj]=\EE_\Eve[A_1^\adj A_1] = \EE_\Eve[\norm{\gamma(\om_j)}^2{\gamma(\om_j)\gamma(\om_j)^\adj}] \preceq \Lu_{01} (\norm{\EE[A_j]}+t/2)\Id.
\end{align*}

\end{itemize}
The last claim is because
$\norm{\etaMat - \subetaMat} \leq t$ implies that $\norm{\etaMat} \leq 3/2+t$, $\norm{\etaMat^{-1}} \leq \frac{\norm{\etaMat}}{1-\norm{\etaMat-\subetaMat} \norm{\etaMat^{-1}}} \leq \frac{3}{2-4t}$ and $\norm{\etaMat^{-1} - \subetaMat^{-1}} \leq \norm{\etaMat^{-1}} \norm{\etaMat-\subetaMat} \norm{\subetaMat^{-1}} \leq  \frac{3t}{1-2t}$ and writing $\tilde t = \frac{3t}{1-2t}$ is equivalent to $t = \tilde t/(3+2\tilde t)$.
\end{proof}

\paragraph{Bounds on $\subetaFunc_X$ applied to  a fixed vector}

\begin{proposition}\label{prop:etaFunc_vec}
Let $t\in (0,1)$, $r\in \{0,2\}$,  $q\in \CC^{s(d+1)}$ and $y\in \Xx_r$, where $\Xx_0\eqdef \Xx$ and $\Xx_2\eqdef \Sn$. 
 If \[
\PP(E_\om^c ) \leq \frac{t}{t+4 B_r } \qandq \EE[ L_{01}(\om) L_r(\om)\bun_{E_\om^c}] \leq \frac{t}{4 \sqrt{s}}
\]
then \[\PP_\Eve\pa{\norm{ \diff{r}{(\subetaFunc_{X_0} - \etaFunc_{X_0})^\top q} (y)} \geq t \norm{q}} \leq 2\tilde d \exp\pa{\frac{ - m t^2/4}{2 \Lu_r^2  + \Lu_r \Lu_{01}t/(3\sqrt{s})}}\]
where $\tilde d = 1$ if $r=0$ and $\tilde d = d$ if $r=2$.

As a consequence, since $\sqrt{2s} \ns{q} \geq \norm{q}_2$, we have
$$
\PP_E\pa{\norm{\diff{r}{(\etaFunc_{X_0} - \subetaFunc_{X_0})^\top q}(y)} \geq t \ns{q} } \leq  2 \tilde d\exp\pa{\frac{ - m t^2}{16 s (\Lu_r^2  + 8 \Lu_r \Lu_{01}t/(3\sqrt{2} ))}}
$$
provided that
\[
\PP(E_\om^c ) \leq \frac{t}{t+4\sqrt{2s} B_r } \qandq \EE[ L_{01}(\om) L_r(\om)\bun_{E_\om^c}] \leq \frac{t}{4\sqrt{2} s}.
\]

\end{proposition}
\begin{proof}
Without loss of generality, assume that $\norm{q}=1$. First note that $$\diff{r}{(\subetaFunc_{X_0} - \etaFunc_{X_0})^\top q}(y) =  \frac{1}{m}\sum_{k=1}^m q^\top \RFVec(\om_k) \diff{r}{\phi_{\om_k}}(y) - \EE[q^\top \RFVec(\om_k) \diff{r}{ \phi_{\om_k}}(y)].
$$

We first consider the case of $r=0$.
We apply Lemma \ref{lem:adapt_bernstein} to $A_k \eqdef q^\top \RFVec(\om_k) \phi_{\om_k}(y) \in\CC$: Note that $\abs{A_k} \leq \sqrt{s}L_{01}(\om_k) L_0(\om_k)$ and $\abs{\EE[A_k]} \leq B_0$.
\begin{itemize}
\item Under event $E_{\om_k}$, $\norm{A_k} \leq \Lu_2 \Lu_{01} \sqrt{s}\eqdef L$.
\item $\EE_\Eve \abs{A_k}^2 = \EE_\Eve[  \dotp{\RFVec(\om_k)\RFVec(\om_k)^* q}{q} \abs{\phi_{\om_k}(y)}^2] \leq \Lu_0^2 \norm{\etaMat_\Eve} \leq \pa{3/2+t/2} \Lu_0^2 \leq 2 \Lu_0^2 \eqdef \sigma^2$.
\end{itemize}

For the case $r=2$,  we apply Lemma \ref{lem:adapt_bernstein} with  $A_k \eqdef q^\top \gamma(\om_k) \diff{2}{\phi_{\om_k}}(y) \in \CC^{d\times d}$. Then, $\norm{A_k} \leq \sqrt{s} L_{01}(\om_k) L_2(\om_k)$, $\norm{\EE[A_k]} \leq B_2$, under event $E_{\om_k}$, $\norm{A_k} \leq \Lu_2 \Lu_{01}\sqrt{s} \eqdef L$ and $$
\norm{\EE_\Eve[A_k A_k^*]} = \norm{\EE_\Eve[A_k^* A_k]} = \norm{\EE_\Eve[ \diff{2}{\phi_{\om_k}}(y)^* \diff{2}{\phi_{\om_k}}(y) \abs{q^\top \gamma(\om_k)}^2]} \leq \Lu_2^2 \EE_\Eve[\abs{q^\top \gamma(\om_k)}^2] \leq 2 \Lu_2^2\eqdef \sigma^2.
$$

\end{proof}

\begin{lemma}\label{lem:appl-vec-bern}
Assume that
$$
\PP(E_\om^c) \leq \frac{ t}{t+6 \sqrt{2 s}} \qandq \EE[L_{01}(\om)^2 \bun_{\Eve^c}] \leq \frac{t}{4 \sqrt{2}  s^{3/2}}
$$
Let   $q\in \CC^{s(d+1)}$. Then, for all 
$
t\geq {  \frac{2\sqrt{2} s \Lu_{01}\Lu_1}{m}  +  \sqrt{\frac{8 s^2 \Lu_{01}^2\Lu_1^2}{m^2}  +  \frac{144 s \Lu_1^2}{m} } }
$,
 we have for each $x_i\in X$,
$$
\PP_E\pa{\norm{ \diff{1}{q^\top(\etaFunc_X - \subetaFunc_X)}(x_i) }_2 > 2t\ns{q}} \leq 28 \exp \pa{ - \frac{m t^2/(4s)}{ 2\Lu_1^2 + \sqrt{2}t  \Lu_1 \Lu_{01} /3}} .
$$
\end{lemma}
\begin{proof}

For each $x_i\in X$,
$$
\norm{\diff{1}{(\EE_\Eve[q^\top\subetaFunc_X ]- q^\top \etaFunc_X)}(x_i)} \leq \norm{\etaMat - \etaMat_\Eve} \norm{q}  \leq \frac{t}{\sqrt{2s}} \norm{q},
$$
by Proposition \ref{prop:bernstein_etaMat}.
For convenience, we drop the subscript $X$ from $\etaFunc_X$.
Fix $i\in \{1,\ldots, s\}$. Observe that
\begin{align*}\PP_E&\pa{\norm{\diff{1}{q^\top(\etaFunc - \subetaFunc)}(x_i) }_2 > 2t\ns{q}} \leq
\PP_E\pa{\norm{\diff{1}{q^\top(\etaFunc - \subetaFunc)}(x_i)  }_2 > \frac{2t}{\sqrt{2s}}\norm{q}_2} \\
&
\leq \PP_E\pa{\norm{\diff{1}{q^\top(\EE_\Eve[\subetaFunc] - \subetaFunc)}(x_i) }_2 > \frac{t}{\sqrt{2s}}\norm{q}_2  }\\
\end{align*}
The claim of this lemma follows by applying Lemma \ref{lem:vec-bernstein}:
Let $$Y_k = {\diff{1}{ \varphi_{\om_k}}(x_i) \RFVec(\om_k)^\top} q -  \EE_\Eve {\diff{1}{ \varphi_{\om_k}}(x_i) \RFVec(\om)^\top }q \in \CC^d,
 $$
 and observe that
 $\diff{1}{q^\top( \subetaFunc- \EE_\Eve[\subetaFunc])}(x_i)  = \frac{1}{m}\sum_k Y_k$.
Without loss of generality, assume that $\norm{q}_2 = 1$.
We apply Lemma \ref{lem:vec-bernstein}. Observe that conditional on event $E$,
\begin{itemize}
\item $\norm{Y_k}_2 \leq 2\norm{q}_2 \norm{\RFVec(\om_k)}_2 \norm{\diff{1}{ \varphi_{\om_k}}(x_i)}_2 \leq 2\sqrt{s } \Lu_{01} \Lu_1.$
\item $\EE_E\norm{Y_k}^2 \leq \EE_E [\abs{\RFVec(\om_k)^\top q}^2 \diff{1}{ \varphi_{\om_k}}(x_i)\diff{1}{  \varphi_{\om_k}}(x_i)^\top ]  \leq \Lu_1^2 \norm{\etaMat_E}$. So, $\sigma^2 \leq m\Lu_1^2 \norm{\etaMat_E} \leq m\Lu_1^2 (t+\norm{\etaMat}) \leq m \Lu_1^2 (t/2+3/2) \leq 2m \Lu_1^2$ (here we are talking about the $\sigma^2$ in Lemma \ref{lem:vec-bernstein}).
\end{itemize}
Therefore, for all
\[
t\geq {  \frac{2\sqrt{2} s \Lu_{01}\Lu_1}{m}  +  \sqrt{\frac{8 s^2 \Lu_{01}^2\Lu_1^2}{m^2}  +  \frac{144 s \Lu_1^2}{m} } }
\]
$$
\PP\pa{\norm{\frac{1}{m} \sum_{k=1}^m Y_k}_2 \geq \frac{t}{\sqrt{2s}}} \leq 28 \exp \pa{ - \frac{m t^2/(4s)}{ 2\Lu_1^2 + \sqrt{2}t  \Lu_1 \Lu_{01} /3}}
$$

\end{proof}

\begin{proposition}[Block norm bound on $\subetaMat$ applied to a fixed vector]\label{prop:bound_R_vec}
Suppose that
$$
\PP(E_\om^c) \leq \frac{ t}{t+6 \sqrt{s}(B_0+1)} \qandq \EE[L_{01}(\om)^2 \bun_{\Eve^c}] \leq \frac{t}{4s^{3/2}(1+4B_0)}
$$
Then, 
for all
\[
t\geq \pa{  \frac{4\sqrt{2} s \Lu_{01}\Lu_1}{m}  +  \sqrt{\frac{32 s^2 \Lu_{01}^2\Lu_1^2}{m^2}  +  \frac{576 s \Lu_1^2}{m} } }
\]
 we have
\eql{
\begin{split}
\PP_E&\pa{\ns{(\etaMat - \subetaMat)q} \geq t \ns{q}} \leq  32s \exp \pa{ - \frac{m t^2}{  s\pa{ 32 \Lu_1^2 + 34 t  \Lu_1 \Lu_{01} }}}.
\end{split}
}

\end{proposition}
\begin{proof}
Let $S_0 \eqdef \{1,\ldots, s\}$ and $S_j \eqdef \{ s+(j-1)d+1, \ldots, s+jd\}$ for $j=1,\ldots, s$. 
Observe that by the union bound
\begin{align}
\begin{split}\label{eq:tobound1}
\PP_E&\pa{\ns{(\etaMat - \subetaMat)q} \geq t \ns{q}}
\\
&
\leq  \PP_E\pa{\norm{((\etaMat - \subetaMat) q)_{S_0}}_\infty \geq t \ns{q}} + \sum_{j=1}^s \PP_E\pa{\norm{((\etaMat - \subetaMat) q)_{S_j}}_2 \geq t \ns{q}} \\
&
\leq \sum_{j=1}^s \PP_E\pa{\abs{((\etaMat - \subetaMat) q)_{j}} \geq t \ns{q}} + \sum_{j=1}^s \PP_E\pa{\norm{((\etaMat - \subetaMat) q)_{S_j}}_2 \geq t \ns{q}}.
\end{split}
\end{align}
To bound the first sum, observe that
$((\etaMat - \subetaMat) q)_{j} = (\etaFunc(x_j) - \subetaFunc(x_j))^\top q$ and $((\etaMat - \subetaMat) q)_{S_j} = \diff{1}{q^\top(\etaFunc - \subetaFunc)}(x_j)$. So, the first sum can be bounded by applying Proposition \ref{prop:etaFunc_vec}. The second sum can be bounded by applying Lemma \ref{lem:appl-vec-bern}.

\end{proof}


\paragraph{Norm bounds for $\subetaFunc$}
We will repeatedly make use of the following result on $\subetaFunc_X$. This result is due to concentration bounds on the kernel $\Cov$ which are derived subsequently.
\begin{proposition}[Bound on $\subetaFunc_X$]\label{prop:bound_etafunc}
Let $X\in \Xx^s$. Let $\rho>0$. 
Assume that
for all $(i,j) \in \{(0,0),(1,0), (0,2),(1,2)\}$,
\[
\PP(E_\om^c)\leq \frac{t}{t+4\sqrt{s}\max\{B_0,B_2\}} ,\quad \EE[L_i(\om)L_j(\om)\bun_{E_\om^c}] \leq \frac{t}{4\sqrt{s}}
\]
Then, given any $y\in \Xx$, 
\begin{equation}\label{bound_etafunc_0}
\begin{split}
\PP_{\Eve}\pa{\norm{\subetaFunc_X(y)-\etaFunc_X(y)} \geq t} \leq 4s d \exp\pa{-\frac{m t^2/8}{ 3 s  \Lu_{01}^2 }} .
\end{split}
\end{equation}
and given any $y\in \Sn$, writing $\subetaFunc_X = (\hat f_j)_{j=1}^p$ and $\etaFunc_X = ( f_j)_{j=1}^p$ with $p=s(d+1)$, we have
\begin{equation}\label{bound_etafunc_2}
\begin{split}
\PP_{\Eve}\pa{\sup_{\norm{q} = 1} \sqrt{\sum_{j=1}^p \norm{\diff{2}{\hat f_j-  f_j}(y) q }^2} >  t} \leq s(3d+d^2) \exp\pa{-\frac{m t^2/8}{ s(\Lu_2^2 B_{11} + \Lu_1^2 B_{22} + \Lu_{01}\Lu_2)}} .
\end{split}
\end{equation}

\end{proposition}
\begin{proof}
 Let $i,j\in \NN_0$ with $i+j\leq 2$. 
Let $[s]\eqdef \{1,\ldots, s\}$ and $I \eqdef \{(0,0),(1,0)\}$,  
By Lemma \ref{lem:conc_kernel} and the union bound,
\begin{equation}
\begin{split}
\PP_{\Eve}\pa{\exists (i,j) \in I, \exists \ell\in [s], \norm{\Cov^{(ij)}(x_\ell,y) - \fullCov^{(ij)}(x_\ell,y)}\geq \frac{t}{\sqrt{s}}} \leq 4s d \exp\pa{-\frac{m t^2/4}{ 3 s \Lu_{01}^2 }}.
\end{split}
\end{equation}
So, \eqref{bound_etafunc_0} follows because
\[
\norm{\subetaFunc_X(y)-\etaFunc_X(y)} \leq  \sqrt{\sum_{i=1}^s \abs{\Cov(x_i,y) - \fullCov(x_i,y)}^2 + \norm{\Cov^{(10)}(x_i,y) - \fullCov^{(10)}(x_i,y)}^2} \leq  \sqrt{2} t.
\]

By Lemma  \ref{lem:conc_kernel}, Lemma \ref{lem:conc_kernel_3} and the union bound, letting $I_2\eqdef \{(0,2), (1,2)\}$, we have
\begin{equation}
\begin{split}
\PP_{\Eve}&\pa{\exists (i,j) \in I_2, \exists \ell\in [s], \norm{\Cov^{(ij)}(x_\ell,y) - \fullCov^{(ij)}(x_\ell,y)}\geq \frac{t}{\sqrt{s}}} \leq 2s d \exp\pa{-\frac{m t^2/4}{ 2 s(\Lu_2^2  + \Lu_0\Lu_2 )}}\\
&+ s(d+d^2) \exp\pa{ - \frac{mt^2/4}{s(\Lu_2^2 B_{11} + \Lu_1^2 B_{22}+ \Lu_1\Lu_2 )}}. 
\end{split}
\end{equation}
and \eqref{bound_etafunc_2} follows since given $q\in \CC^d$, $\norm{q}=1$, we have
\begin{align*}
&\sum_{j=1}^p \norm{\diff{2}{\hat f_j-  f_j}(y) q }^2 \leq \sum_{j=1}^s \pa{\norm{\Cov^{(02)}(x_j,y) - \fullCov^{(02)}(x_j,y) }^2 + \norm{\Cov^{(12)}(x_j,y) - \fullCov^{(12)}(x_j,y) }^2  } \leq  2t^2
\end{align*}

\end{proof}

\begin{lemma}[Concentration on kernel]\label{lem:conc_kernel}
Let  $t>0$, $x,x'\in \Xx$. Let $i,j\in \NN_0$ with $i+j\leq 2$. Assume
\[
\PP(E_\om^c)\leq \frac{t}{t+4\norm{\fullCov^{(ij)}(x,x')}} ,\quad \EE[L_i(\om)L_j(\om)\bun_{E_\om^c}] \leq \frac{t}{4}
\]
then 
\[
\PP_{\Eve}\pa{\norm{\Cov^{(ij)}(x,x') - \fullCov^{(ij)}(x,x')}\geq t} \leq 2d \exp\pa{-\frac{mt^2}{ \Lu_p^2 (b_{ij} + 1) + \Lu_i\Lu_j t/3}} 
\]
where $p= \max\pa{i,j}$ and $b_{ij}=1$ if $\min\pa{i,j}=0$ and $b_{ij} \eqdef \norm{\fullCov^{(11)}(x,x')}$ otherwise.
\end{lemma}
\begin{proof}
It is an immediate application of Lemma \ref{lem:adapt_bernstein} with $A_k = \rep{\overline{\diff{i}{\phi_{\om_k}}(x)}\diff{j}{\phi_{\om_k}}(x')^\top}$ for $k=1,\ldots, m$. Note that $A_{k}\in  (\RR^d)^{i+j}$ if $(i,j) \in \ens{(0,0),(0,1),(1,0)}$ and $A_k \in \RR^{d\times d}$ if $\max(i,j) = 2$. noting that under $\Eve$, $\norm{A_k } \leq \Lu_i \Lu_j$. Next, we need to bound $\norm{\EE_\Eve[A_k A_k^\adj]}$ and  $\norm{\EE_\Eve[A_k^\adj A_k]}$. We present only the argument for $(i,j) = (0,2)$, since all the other cases are similar:
\begin{align*}
0\preceq \EE_\Eve A_k A_k^\adj &\preceq  \EE_\Eve[ \norm{{\phi_{\om_k}}(x')}^2 {\diff{2}{\phi_{\om_k}}(x)\diff{2}{\phi_\om}(x)^\adj}]  \\
&\preceq \Lu_2^2 \EE_\Eve \norm{{\phi_{\om_k}}(x')}^2 \Id = \Lu_2^2 \abs{\fullCov_\Eve(x',x')} \Id \preceq (1+t/2) \Lu_2^2 \Id
\end{align*}
so $\norm{\EE_\Eve A_k A_k^\adj} \leq (1+t/2) \Lu_2^2$. Similarly, $\norm{\EE_\Eve A_k^\adj A_k}\leq (1+t/2) \Lu_2^2$ and
$$
\norm{\EE_\Eve A_k^\adj A_k}, \norm{\EE_\Eve A_k A_k^\adj} \leq L_p^2 (B_{qq} + t/2)
$$
where $p= \max\pa{i,j}$ and $q=\min\pa{i,j}$.

\end{proof}

Applying a grid on $\Sn$, we get a uniform version.
\begin{lemma}\label{lem:conc_kernel_unif}

Let $i,j\in \NN_0$ with $i+j\leq 2$, and assume that
\[
\PP(E_\om^c)\leq \frac{t}{t+ 16 B_{ij}} ,\quad \EE[L_i(\om)L_j(\om)\bun_{E_\om^c}] \leq \frac{t}{16}.
\]
Then
\begin{align*}
&\PP_{\Eve}\pa{\exists ~x,x'\in \Sn,~\norm{\Cov^{(ij)}(x,x') - \fullCov^{(ij)}(x,x')}\geq t} \\
&\qquad \qquad \leq 2d s^2  \exp\pa{-\frac{m t^2/16}{ L_p^2 (B_{qq} + 1) + \Lu_i\Lu_j t/12} + 2d \log\pa{ \frac{4  (\Ll_i \Lu_j + \Lu_i \Ll_j)}{t}}}.
\end{align*}
where $p= \max\pa{i,j}$ and $q=\min\pa{i,j}$ and $\Ll_i, \Ll_j$ are as in Lemma \ref{lem:features_lip}

\end{lemma}
\begin{proof}
We define a $\delta$-covering of $\Sn$ for the metric $d_\met$ with $\delta = \min\pa{\rnear,\frac{t}{4 (\Ll_i \Lu_j + \Lu_i \Ll_j)}}$ of size $s\pa{\frac{\rnear}{\delta}}^{d}$.  Let this covering be denoted by $\Xx^{\text{grid}}$.

By the union bound and Lemma \ref{lem:conc_kernel},
 \[
\PP_{\Eve}\pa{\exists x, x' \in \Xx^{\text{grid}}  \text{ s.t. } \norm{\Cov^{(ij)}(x,x') - \fullCov^{(ij)}(x,x')}\geq t/4} \leq 2d s^2 \pa{\frac{\rnear}{\delta
}}^{2d} \exp\pa{-\frac{m t^2/16}{ L_p^2 (B_{qq} + 1) + \Lu_i\Lu_j t/12}} 
\]
where $p= \max\pa{i,j}$ and $q=\min\pa{i,j}$.
 This gives the required upper bound:
 Given any $x,x'\in \Xx$, let $x_\text{grid},x_\text{grid}'\in \Xx^\text{grid}$ be such that $d_\met(x,x_\text{grid}), d_\met(x',x'_\text{grid}) \leq \delta$.
 Then, under event $\Eve$, by Lemma \ref{lem:features_lip},
\begin{align*}
&\norm{\Cov^{(ij)}(x,x') - \Cov^{(ij)}(x_\text{grid},x_\text{grid}')}\leq  (\Ll_i \Lu_j + \Lu_i \Ll_j)\delta\leq t/4.
\end{align*} 
By Jensen's inequality and since $\norm{\fullCov_\Eve^{(ij)}(x,x') -\fullCov^{(ij)}(x,x') }\leq t/4$ for all $x,x'$, we have 
\begin{align*}
\norm{\fullCov^{(ij)}(x,x') - \fullCov^{(ij)}(x_\text{grid},x'_\text{grid})} \leq t/2.
\end{align*}
\end{proof}

We now derive analogous results for the kernel differentiated 3 times.

\begin{lemma}[Concentration on order $3$ kernel]\label{lem:conc_kernel_3}
Let $x,x'\in \Sn$. Assume that
\[
\PP(E_\om^c)\leq \frac{t}{t+4 \max\{B_{12},B_{22}\}} ,\quad \EE[(L_1(\om)L_2(\om)+ L_2^2(\om))\bun_{E_\om^c}] \leq \frac{t}{4}
\]
For $j=1,\ldots, m$, let $a_i = (\diff{1}{\overline{\phi_{\om_j}}}(x))_i \in \CC$, $D\eqdef \diff{2}{\phi_\om}(x')\in \CC^{d\times d}$ and
\begin{equation}\label{eq:A_j}
A_j \eqdef \begin{pmatrix}
a_1 D & a_2 D &\cdots & a_d D
\end{pmatrix}^\top \in \CC^{d^2\times d}
\end{equation}
Let $Z \eqdef  \frac{1}{m}\sum_{j=1}^m(A_j-\EE[A_j])$.
Then, given 
\begin{align*}
g(x') \eqdef (g_i(x'))_{i=1}^d &\eqdef \sum_{k=1}^m\pa{\overline{\diff{1}{\phi_{\om_k}}(x)} \phi_{\om}(x') - \EE[\overline{\diff{1}{\phi_{\om_k}}(x)} \phi_{\om}(x')] }\\
&= \Cov^{(10)}(x,x')  - \fullCov^{(10)}(x,x'), 
\end{align*}
\begin{itemize}
\item[(i)] $\sup_{q\in \CC^d,\norm{q}\leq 1} \sum_{i=1}^d \norm{\diff{2}{g_i}(x') q}^2 = \norm{Z}^2$ ,
\item[(ii)] $\sup_{q\in \CC^d,\norm{q}\leq 1} \norm{\diff{2}{q^\top g}(x')} = \norm{\Cov^{(12)}(x,x') - \fullCov^{(12)}(x,x')} \leq \norm{Z} $.
\end{itemize}
and
\[
\PP_\Eve\pa{\norm{Z}\geq t} \leq (d+d^2) \exp\pa{ - \frac{mt^2/4}{\tilde B + \Lu_1\Lu_2 t/3}}
\]
where $\tilde B\eqdef \max\{\Lu_2^2(B_{11}+t/2),\Lu_1^2(B_{22}+t/2) \}$.
\end{lemma}

\begin{proof}
The claim (i) is simply by definition, since $Zq = \pa{\diff{2}{g_i}(x') q}_{i=1}^d \in \CC^{d^2}$.
For (ii), the first equality is simply be definition, and for the inequality, observe that
\begin{align*}
\sup_{q\in \CC^d,\norm{q}\leq 1}& \norm{\diff{2}{q^\top g} (x')} =  \sup_{q\in \CC^d,\norm{q}\leq 1} \sup_{p\in \CC^d,\norm{p}\leq 1} \norm{ \sum_{i=1}^d q_i \diff{2}{g_i}(x')p}\\
&\leq \sup_{q\in \CC^d,\norm{q}\leq 1} \sup_{p\in \CC^d,\norm{p}\leq 1} \norm{q}  \sqrt{ \sum_{i=1}^d \norm{ \diff{2}{g_i}(x')p}^2} \leq \norm{Z}.
\end{align*}

Finally, the probability bound follows by applying Lemma \ref{lem:adapt_bernstein}:  First note that
under $\Eve$,  $\norm{A_j} \leq \Lu_1 \Lu_2$.
It remains to bound $\norm{\EE_\Eve [A_j^\adj A_j ]}$ and $\norm{\EE_\Eve [A_j A_j^\adj ]}$:
\begin{align*}
\sup_{\norm{q}\leq 1}\EE_\Eve \dotp{ A_j^\adj A_j q}{q} & =\sup_{\norm{q}\leq 1}\EE_E \sum_{i=1}^d \abs{ (\diff{1}{\phi_{\om_j}}(x))_i}^2 \norm{ \diff{2}{\phi_\om}(x') q }^2 \\
& \leq\sup_{\norm{q_k}\leq 1} \Lu_1^2 \EE_\Eve \overline{\diff{2}{\phi_\om}(x')[q_1,q_2]} \diff{2}{\phi_\om}(x')[q_3,q_4] \\
& \leq \Lu_1^2 \norm{\fullCov^{(22)}_\Eve(x,x)} \leq \Lu_1^2 (B_{22} + t/2).
\end{align*}
Given  $p_i\in \CC^d$ for $i=1,\ldots, d$ such that $\sum_i \norm{p_i}^2 \leq 1$, write $P = \begin{pmatrix}
p_1 & p_2 &\cdots p_d
\end{pmatrix}\in \CC^{d\times d}$ and $\bar p = \begin{pmatrix}
p_1^\top & p_2^\top &\cdots p_d^\top
\end{pmatrix}^\top \in \CC^{d^2}$. Then,
\begin{align*}
&\EE_E\dotp{A_j A_j^\adj \bar p}{ \bar p} =\EE_E \norm{\sum_{i=1}^d  (\diff{1}{\phi_{\om_j}}(x))_i \diff{2}{ \phi_{\om_j}}(x') p_i}^2 \\
&= \EE_E\norm{\diff{2}{\phi_{\om_j}}(x')  P \diff{1}{\phi_{\om_j}}(x)}^2\\
&\leq \Lu_2^2 \EE_E \sum_i \abs{\sum_k p_{i,k} (\diff{1}{\phi_{\om_j}}(x))_k }^2\\
&= \Lu_2^2 \sum_i \dotp{\Cov^{(11)}_\Eve(x,x) p_i}{p_i}
\leq \Lu_2^2 \norm{\Cov^{(11)}_\Eve(x,x) }^2 \sum_i \norm{p_i}^2 \leq  \Lu_2^2 (B_{11}+t/2).
\end{align*}

\end{proof}

\begin{lemma}[Uniform concentration on order 3 kernel]\label{lem:conc_kernel_3_unif}

 Assume
\[
\PP(E_\om^c)\leq \frac{t}{t+16\max\{B_{12}, B_{22}\}} ,\quad \EE[L_1(\om)L_2(\om)\bun_{E_\om^c}] \leq \frac{t}{16}
\]
then
\begin{align*}
&\PP_{\Eve}\pa{\exists x,x'\in \Sn,~\norm{\Cov^{(12)}(x,x') - \fullCov^{(12)}(x,x')}\geq t} \\
&\qquad \qquad \leq s^2 (d+d^2) \exp\pa{ - \frac{mt^2/16}{\tilde B+ \Lu_1\Lu_2 t/6} + 2d \log\pa{\frac{8 (\Ll_1 \Lu_2 + \Lu_2 \Ll_2)}{t } }}
\end{align*}
where
$\tilde B\eqdef \max\{\Lu_2^2(B_{11}+t/2),\Lu_1^2(B_{22}+t/2) \}$, $\Ll_1$, $\Ll_2$ are as in Lemma \ref{lem:features_lip}.

\end{lemma}

\begin{proof}
Let $\Xx^{\text{grid}}$ be a $\delta$-covering of $\Sn$ for the metric $\dsep$ with $\delta = \min\pa{\rnear,\frac{t}{8 (\Ll_1 \Lu_2 + \Ll_2 \Lu_2) }}$ of size at most $s \pa{\frac{8 (\Ll_1 \Lu_2 + \Ll_2 \Lu_2) }{t } }^{d}$. By Lemma \ref{lem:conc_kernel_3} and the union bound, 
\begin{align*}
\PP_{\Eve}&\pa{ \exists x, x' \in \Xx^\text{grid}, \; \norm{\Cov^{(ij)}(x,x') - \fullCov^{(ij)}(x,x')}\geq t/2} \\
&\leq s^2(d+d^2)  \pa{\frac{8 (\Lu_1 \Lu_2 + \Lu^2_2) }{t } }^{2d} \exp\pa{ - \frac{mt^2/16}{\Lu_2^2(B_{11}+t/4) + \Lu_1\Lu_2 t/6}} \eqdef \rho.
\end{align*}
Moreover, under event $\Eve$, given any $x,x'\in \Sn$, there exists grid points $x_\text{grid}$, $x_\text{grid}'$ such that $$d_\met(x,x_\text{grid}),d_\met(x',x'_\text{grid}) \leq \delta$$ and
\begin{align*}
\norm{\pa{\Cov^{(12)}(x,x') - \fullCov^{(12)}(x,x')}} &\leq \norm{\pa{\Cov^{(12)}(x_\text{grid},x'_\text{grid}) - \fullCov^{(12)}(x_\text{grid},x'_\text{grid})}} \\
&\qquad + \norm{\pa{\Cov^{(12)}(x,x') - \Cov^{(12)}(x_\text{grid},x'_\text{grid})}} \\
&\qquad + \norm{\pa{\fullCov^{(12)}(x,x') - \fullCov^{(12)}(x_\text{grid},x'_\text{grid})}},
\end{align*}
and by Lemma \ref{lem:features_lip}, under event $\Eve$,
$$
\norm{\pa{\Cov^{(12)}(x,x') - \Cov^{(12)}(x_\text{grid},x'_\text{grid})}} \leq (\Ll_1 \Lu_2 + \Ll_2 \Lu_2) \delta \leq t/8.
$$
and by Jensen's inequality and since $\norm{\fullCov^{(12)}(x,y) - \fullCov_\Eve^{(12)}(x,y)} \leq t/8$,
$$
\norm{\pa{\fullCov^{(12)}(x,y) - \fullCov^{(12)}(x_\text{grid},y)}} \leq 3t/8.
$$
Therefore, conditional on $\Eve$, $\norm{\pa{\Cov^{(12)}(x,y) - \fullCov^{(12)}(x,y)}} <t$ with probability at least $1-\rho$.

\end{proof}

%
%
%
%


\section{Proof of Theorem \ref{thm:main_lip_eta}}\label{sec:nondegen_mnc}

In all the rest of the proofs we fix $X_0\in \Xx^s$ to be $\Delta$-separated points, $a_0\in \CC^s$, and let $\SignVecPad = (\sign(a_0),0_{sd})$. We denote $\Xx^\text{near}_i = \enscond{x\in \Xx}{\dsep(x,x_{0,i}) \leq \rnear}$ and $\Xx^\text{near} = \cup_i \Xx^\text{near}_i$ and $\Xx^\textup{far} = \Xx \backslash\Xx^\textup{near}$.

Since $\fullCov$ is an admissible kernel, from \eqref{eta:sfar} and \eqref{eta:snear} in the proof of  Theorem \ref{thm:NDetanew}
$\fulleta_{X_0}$ satisfies
\begin{itemize}
\item[(i)] for all $y \in   \Xx^\text{far}$,  $\abs{\fulleta_{X_0}(y)} \leq 1- \frac12 \constker_0$,
\item[(ii)] for all $y\in \Xx^\text{near}(i)$, $-\rep{\sign(a_i) \diff{2}{\fulleta_{X_0}}(y)} \succcurlyeq \frac12 \constker_2 \Id$ and $\norm{\imp{\sign(a_i) \diff{2}{\fulleta_{X_0}}(y)}} \leq (\frac{p}{2})\frac12\constker_2$.
\end{itemize}
 $$
p\eqdef \sqrt{(1-\constker_2 \rnear^2/2)/(  \constker_2 \rnear^2/2)} \geq 1,$$
since $\constker_2 \rnear^2 \leq 1$ by assumption of $\fullCov$ being admissible.
We aim to show that, for $X$ close to $X_0$, $\subeta_{X}$ is nondegenerate by showing that $\norm{\diff{r}{\subeta_{X}} - \diff{r}{\fulleta_{X_0}}} \leq c \constker_r$ for some positive constant $c$ sufficiently small.

\subsection{Nondegeneracy of $\subeta_{X_0}$}
We first establish the nondegeneracy of $\subeta_{X_0}$, our proof can be seen as a generalisation of the techniques in \cite{tang2013compressed} to the multidimensional setting with general sampling operators:
\begin{theorem}\label{thm:nondegen_X0}
Let $\rho>0$ and assume that the assumptions in Section \ref{sec:assumption} hold. Assume also that either (a) or (b) holds:
\begin{itemize}
\item[(a)] $\sign(a_0)$ is a Steinhaus sequence and
$$
m\gtrsim C\cdot s \cdot \log\pa{\frac{N^d}{\rho}}\log\pa{\frac{s}{\rho}} 
$$
\item[(b)] $\sign(a_0)$ is an arbitrary sequence from the complex unit circle, and 
$$
m\gtrsim C\cdot s^{3/2} \cdot \log\pa{\frac{N^d}{\rho}} 
$$
\end{itemize}
where $C, N$ are defined in the main paper.
Then with probability at least $1-\rho$, the following hold: For all $y \in   \Xx^\text{far}$,  $\abs{\subeta_{X_0}(y)} \leq 1- \frac{7}{16}\constker_0$, and
for all $y\in \Xx^\text{near}(i)$, $-\rep{\sign(a_i) \diff{2}{\subeta_{X_0}}(y)} \succcurlyeq \frac{7}{16}\constker_2 \Id$ and $\norm{\imp{\sign(a_i) \diff{2}{\subeta_{X_0}}(y)}} \leq (\frac{p}{2}+\frac{p}{8}) \frac12 \constker_2$
and hence,   $\subeta_{X_0}$ is $(\frac{7}{16}\constker_0, \frac{7}{16}\constker_2)$-nondegenerate.

  \end{theorem}

\begin{proof} 
Note that 
\[
\frac{8}{7} \pa{\frac{p}{2} + \frac{p}{8}} = \frac{5}{8}p < \sqrt{\frac{1-7\constker_2 \rnear^2/16}{7\constker_2 \rnear^2/16}}
\]
so $\subeta_{X_0}$ is $(\frac{7}{16}\constker_0, \frac{7}{16}\constker_2)$-nondegenerate by Lemma \ref{lem:lem_nondegen_cond}

Let $c\eqdef 1/32$.
Observe that by assumption and Lemma \ref{lem:cdf_bd}, $\PP(\Eve) \leq \rho/2$. Therefore, it is sufficient to prove that conditional on $\Eve$, with probability at least $1-\delta$ with $\delta \eqdef \rho/2$, $\subeta_{X_0}$ is nondegenerate.

We will repeatedly use the fact that our assumptions (by  Lemma \ref{lem:cdf_bd}) also imply that 
\[
\PP(E_\om^c) \leq \frac{\constker}{m} ,\quad \EE[L_i(\om)L_j(\om)\bun_{E_\om^c}] \leq \frac{\constker}{m}
\]
for all  $(i,j) \in \{(0,0),(1,0), (0,2),(1,2)\}$,

\textbf{Step I: Proving nondegeneracy on a finite grid.}

Let $\Sf_\text{grid}\subset \Sf$ and $\Sf_\text{grid}\subset \Sn$, be finite point sets.
Let  $$
Q_r(y) \eqdef \norm{\diff{r}{\subeta_{X_0}}(y) -\diff{r}{\fulleta_{X_0}}(y)}, \qquad r=0,2.
$$

We first prove that conditional on $\Eve$, with probability at least $1-\delta$ where $\delta\eqdef \rho/2$, that $Q_0(y) \leq c  \constker_0$ for all $y\in \Sf_\text{grid}$  and  $Q_2(y) \leq c \constker_2$ for all $y\in \Sf_\text{grid}$. 


Let us first recall some facts which were proven in the previous section:
Let $a,t\in (0,1)$ and write $\etaFunc = (\bar f_j)_{j=1}^{s(d+1)}$ and $\subetaFunc = (f_j)_{j=1}^{s(d+1)}$. Let $q_0 \eqdef  \etaMat^{-1} \SignVecPad$, so $\norm{q_0} \leq 2\sqrt{s}$. Let $F$ be the event that
\begin{itemize}
\item[(a)] $\norm{\etaMat^{-1} - \subetaMat^{-1}}  \leq t$,
\item[(b)] $\forall y \in \Sf_\text{grid}$, $ \norm{\subetaFunc_{X_0}(y)-\etaFunc_{X_0}(y)} \leq a\constker_0$,
\item[(c)] $\forall y \in \Sn_\text{grid}$, $\sup_{q\in \CC^d,\;\norm{q} = 1} \sqrt{\sum_{j=1}^p \norm{\diff{2}{f_j- \bar f_j}(y) q }^2} \leq  a\constker_2$,
\end{itemize}

Let $G$ be the event that
\begin{itemize}
\item[(d)] $\forall y \in \Sf_\text{grid}$, $ \abs{ (\subetaFunc_{X_0}(y) - \etaFunc_{X_0}(y))^\top q_0} \leq  2 a\constker_0$
\item[(e)]  $\forall y \in \Sn_\text{grid}$, $ \norm{ \diff{2}{(\subetaFunc_{X_0} - \etaFunc_{X_0})^\top q_0}(y) } \leq  2 a\constker_2$
\end{itemize}
then provided that
\begin{equation}\label{eq:event_E_tohold}
\PP(E_\om^c)\leq \frac{u}{u+\max\{4 \sqrt{s} B_{ij},6\}} ,\quad \EE[L_i(\om)L_j(\om)\bun_{E_\om^c}] \leq \frac{u}{4{s}}
\end{equation}
where $u = \min\{a \constker_i, t\}$, we have
\begin{equation}\label{eq:probF}
\begin{split}
\PP_\Eve(F^c) \leq & 4(d+1)s \exp\pa{-\frac{mt^2}{16 s\Lu_{01}^2(3 + 2  t)}} \\
&+ 4s d \abs{\Sf_\text{grid}}\exp\pa{-\frac{m (a\constker_0)^2/8}{ s(\Lu_{01}^2 (B_{11} + 1) + \Lu_{01}^2 )}}\\
&
+s(3d+d^2) \abs{\Sn_\text{grid}} \exp\pa{-\frac{m (a\constker_2)^2/8}{ s(\Lu_2^2  B_{11} + \Lu_1^2 B_{22}) + \Lu_{01}\Lu_2)}}\\
\PP_\Eve(G^c) \leq 
& 2 \abs{\Sf_\text{grid}} \exp\pa{-\frac{m a^2 \constker_0^2}{s(8 \Lu_0^2  + \frac{4}{3} \Lu_0 \Lu_{01} a \constker_0) }} \\
& + 2d \abs{\Sn_\text{grid}} \exp\pa{-\frac{m a^2 \constker_2^2}{s(8 \Lu_2^2  + \frac{4}{3} \Lu_2 \Lu_{01} a \constker_2 )}},
\end{split}
\end{equation}
where for $\PP_\Eve(F^c)$, the first term on the right is due to Proposition \ref{prop:bernstein_etaMat},   the second and third are   due to Proposition \ref{prop:bound_etafunc} while the bound on $\PP_\Eve(G^c)$ is  due to Proposition \ref{prop:etaFunc_vec} (noting that, since this probability bound over the $\om_j$ is valid for all fixed $\SignVecPad$, and the $\om_j$ and the signs are independent, it is valid with the same probability over both $\om_j$ and $\SignVecPad$).

Observe that
\begin{equation}\label{eq:tobound}
\begin{split}
\norm{\diff{j}{\subeta_{X_0}}(y) -\diff{j}{\fulleta_{X_0}}(y)} &=\norm{\diff{j}{(\subetaCoeff_{X_0} - \etaCoeff_{X_0})^\top \subetaFunc_{X_0}}(y) + \diff{j}{ \etaCoeff_{X_0}^\top (\subetaFunc_{X_0} - \etaFunc_{X_0})}(y)} \\
&\leq \norm{\diff{j}{\SignVecPad^\top \pa{ (\subetaMat^{-1} - \etaMat^{-1} )\subetaFunc_{X_0} +  \etaMat^{-1}  (\subetaFunc_{X_0} - \etaFunc_{X_0})}}(y)}
\end{split}
\end{equation}

\textbf{Step I (a): Random signs}

We first bound \eqref{eq:tobound} in the case where  $\SignVecPad$ is a Steinhaus sequence.

Let $\beta_1(y) \eqdef  (\subetaMat^{-1} - \etaMat^{-1} )\subetaFunc_{X_0}(y) $ and $  \beta_2(y) \eqdef \etaMat^{-1}  (\subetaFunc_{X_0}(y) - \etaFunc_{X_0}(y))$.  
Then, event $F$ implies that $\norm{\beta_1(y)} \leq t (B_0+ a \constker_0)$ for all $y\in \Sf_\text{grid}$, and event $G$ implies that $\abs{\SignVecPad^\top \beta_2(y)} \leq 2a \constker_0$. 
So,
\begin{equation}\label{eq:eta0}
\begin{split}
\PP_\Eve&\pa{
\abs{\exists y\in \Sf_\text{grid}, \; \SignVecPad^\top (\beta_1+\beta_2)(y)}> c \constker_0} \\
&
\leq \PP_{F \cap \Eve} \pa{\exists y\in \Sf_\text{grid}, \; 
\abs{\SignVecPad^\top \beta_1(y) }> \frac{c}{2} \constker_0 } \PP_\Eve(F)  + \PP_\Eve \pa{ F^c }  \\
&\qquad + \PP_{G \cap \Eve} \pa{\exists y\in \Sf_\text{grid}, \; 
\abs{\SignVecPad^\top \beta_2(y) }> \frac{c}{2} \constker_0 } \PP_\Eve(G)  + \PP_\Eve \pa{ G^c }\\
&\leq \PP_{F \cap \Eve} \pa{\exists y\in \Sf_\text{grid}, \; 
\abs{\SignVecPad^\top 
\beta_1}> \frac{c}{2} \constker_0}  + \PP_\Eve \pa{ F^c } + \PP_\Eve \pa{ G^c }\\
& \leq 4 \abs{\Sf_\text{grid}}
e^{-\frac{ (c/4)^2 \constker_0^2}{8 t^2 (B_0+ a \constker_0)^2}} + \PP_\Eve(F^c) + \PP_\Eve \pa{ G^c }.
\end{split}
\end{equation}
where we set $a=c/4$ for the second inequality and the last inequality follows from Lemma \ref{lem:hoeffding_sign} and because $\SignVecPad$ consists if random signs.

Now consider $Q_2(y) = \diff{2}{\SignVecPad^\top \beta}(y)$. Under event $G$, $\norm{\diff{2}{\SignVecPad^\top \beta_2}(y)} \leq \frac{c}{2} \constker_2$.  Writing $M = (\subetaMat^{-1} - \etaMat^{-1} )$, we have 
\begin{equation}\label{eq:app-mtx-hoef1}
\diff{2}{\SignVecPad^\top \beta_1}(y) = 
{\diff{2}{\SignVecPad^\top \pa{ M\subetaFunc_{X_0}}}(y)} = \sum_{\ell=1}^p \SignVecPad_\ell \pa{\sum_{j=1}^p M_{\ell j} \diff{2}{f_j}(y)}.
\end{equation}
We aim to bound \eqref{eq:app-mtx-hoef1} by applying the Matrix Hoeffding's inequality (Corollary \ref{cor:mtx_hoeff}): let \[
Y_\ell \eqdef \rep{  \sum_{j=1}^p M_{\ell j} \diff{2}{f_j}(y)} \in \RR^{d\times d}\] which is a symmetric matrix. Note that
\begin{align*}
&\norm{\sum_{\ell=1}^p Y_\ell^2 }=\sup_{q\in \RR^d, \norm{q}=1} \sum_{\ell=1}^p  \dotp{Y_\ell^2 q}{ q}  = \sup_{q\in \RR^d, \norm{q}=1} \sum_{\ell=1}^d \norm{Y_\ell q}^2 \leq \sup_{q\in \RR^d, \norm{q}=1}  \norm{ \sum_{j=1}^p M_{\ell,j}  (\diff{2}{f_j}(y) q)}^2.
\end{align*}
Then, for a vector $q$ of unit norm, let $V_{j,n}\eqdef ( \diff{2}{f_j}(y) q )_n$ for $j=1,\ldots, p$ and $n=1,\ldots, d$, then
\begin{align*}
\sum_{\ell=1}^p &\norm{\sum_{j=1}^p M_{\ell,j}  (\diff{2}{f_j}(y) q)}^2 
=  \sum_{\ell=1}^p  \sum_{n=1}^d{ \abs{\sum_{j=1}^p M_{\ell,j} V_{j,n}}}^2 = \sum_{n=1}^d \norm{M V_{\cdot,n}}^2  \leq \norm{M}^2 \sum_{n=1}^d \norm{V_{\cdot,n}}^2\\
&= \norm{M}^2 \sum_{n=1}^d \sum_{j=1}^p \abs{V_{j,n}}^2 = \norm{M}^2  \sum_{j=1}^p \norm{\diff{2}{f_j}(y) q }^2.
\end{align*}
Under event $F$,  we have 
$\norm{M}^2  \sum_{j=1}^p \norm{\diff{2}{f_j}(y) q }^2 \leq t^2(B_2 + a \constker_2)^2$.
Then,
$$
\PP_{F \cap\Eve}\pa{\norm{{\diff{2}{\SignVecPad^\top \rep{ M\subetaFunc_{X_0}}}(y)}}\geq \frac{c \constker_2}{\sqrt{2}}} \leq 2d \exp\pa{-\frac{(c/2)^2 \constker_2^2}{ 4 t^2(B_2 + a \constker_2)^2 }}.
$$
By repeating this argument for the imaginary part, we obtain
$$
\PP_{F \cap\Eve}\pa{\norm{{\diff{2}{\SignVecPad^\top \imp{ M\subetaFunc_{X_0}}}(y)}}\geq \frac{c \constker_2}{\sqrt{2}}} \leq 2d \exp\pa{-\frac{(c/2)^2 \constker_2^2}{ 4 t^2(B_2 + a \constker_2)^2 }}.
$$
So,
\begin{equation}
\begin{split}
\PP_\Eve&\pa{\exists y\in \Sn_\text{grid}, \;
\norm{ \diff{2}{\SignVecPad^\top \beta(y)}}> c \constker_2} \\
&\leq  \PP_{F \cap\Eve}\pa{\exists y\in \Sn_\text{grid}, \;\norm{{\diff{2}{\SignVecPad^\top \rep{ M\subetaFunc_{X_0}}}(y)}}\geq \frac{c}{2} \constker_2}+ \PP_\Eve(F^c) + \PP_\Eve(G^c)\\
&\leq 4d \abs{\Sn_\text{grid}}  \exp\pa{-\frac{(c/2)^2 \constker_2^2}{ 4 t^2(B_2 + a \constker_2)^2 }} +  \PP_\Eve(F^c) + \PP_\Eve(G^c).
\end{split}
\end{equation}

Therefore,
\begin{align*}
&1- \PP\pa{Q_0(y_0) \leq c \constker_0 \text{ and } Q_2(y_2)\leq c\constker_2, \forall y_0 \in \Sf_\text{grid}, \forall y_2 \in \Sn_\text{grid}}\\
& \leq 4 \abs{\Sf_\text{grid}}
\exp\pa{-\frac{ (c/2)^2 \constker_0^2}{32 t^2 (B_0+ a \constker_0)^2}} + 4 d \abs{\Sn_\text{grid}}  \exp\pa{-\frac{(c/2)^2 \constker_2^2}{ 16 t^2(B_2 + a \constker_2)^2 }} +   2\PP_\Eve(F^c)  + 2\PP_\Eve(G^c). 
\end{align*}
The first 2 terms are each bounded by $\delta/7$ by setting $t$ such that
$$
 \frac{1}{t^2} =  2^{13} \log \pa{\frac{112 \bar N d}{\delta} } \frac{\pa{\bar B+  1} }{c^2 \constker^2} 
$$
where $\bar B \eqdef \max\{B_0, B_2\}$, $\constker \eqdef \min\{\constker_0, \constker_2\}$ and $\bar N = \max\pa{\abs{\Sn_\text{grid}}, \abs{\Sf_\text{grid}}}$.
The first term of \eqref{eq:probF} is bounded by $\delta/7$ if
$$
m\geq \frac{1}{t^2} \log\pa{\frac{28(d+1)s}{\delta}} 64 s \Lu_{01}^2  =  s \Lu_{01}^2  \frac{2^{19} \pa{\bar B+  1} }{c^2 \constker^2} \log \pa{\frac{112 \bar N d}{\delta} } \log\pa{\frac{28(d+1)s}{\delta}}   
$$
and the last 4 terms of \eqref{eq:probF} are each bounded by $\delta/7$ provided that
\[
m\gtrsim \log\pa{ \frac{28(s+d)d \bar N}{\delta} } \frac{16 s (\Lu_2^2 B_{11} + \Lu_1^2 B_{22} + \Lu_{01}\Lu_2) }{c^2 \constker^2}
\]

So, to summarise, recalling that $\delta= \rho/2$, $\subeta_{X_0}$ is nondegenerate on $\Sn_\text{grid}$ and $\Sf_\text{grid}$ with probability at least $1-\delta$ (conditional on $\Eve$) provided that
\[
m\gtrsim \log\pa{ \frac{sd N}{\rho} } \log\pa{\frac{sd}{\rho}} \frac{ s (\Lu_2^2 B_{11} + \Lu_1^2 B_{22} + \bar B \Lu_{01}^2 + \Lu_{01}\Lu_2)  }{ \constker^2}
\]
and
\begin{equation*}
\PP(E_\om^c)\lesssim \frac{\constker}{\bar B^{3/2} \sqrt{s} \sqrt{\log(\bar N d/\rho)}} \qandq  ,\quad \EE[L_i(\om)L_j(\om)\bun_{E_\om^c}] \lesssim \frac{\constker}{4{s} \sqrt{B} \sqrt{\log(\bar N d/\rho)}}
\end{equation*}

\paragraph{Step I (b): Deterministic signs} Assume now that $\SignVecPad$ consists of arbitrary signs. We will show that \eqref{eq:tobound} can be bounded by $c\constker$ when $m$ is chosen as in condition (b) of this theorem.
Let $F'$ be the event that 
\begin{itemize}
\item[(a')] $\norm{\etaMat - \subetaMat} \leq \frac{t}{s^{1/4}}$ and  $\norm{\etaMat^{-1} - \subetaMat^{-1}} \leq \frac{t}{s^{1/4}}$
\item[(b')] $\forall y \in \Sf_\text{grid}$, $ \norm{ (\subetaFunc_{X_0}(y) - \etaFunc_{X_0}(y)) } \leq  \frac{ a\constker_0}{s^{1/4}}$
\item[(c')]  $\forall y \in \Sn_\text{grid}$, $\sup_{\norm{q}=1} \norm{ \diff{2}{(\subetaFunc_{X_0} - \etaFunc_{X_0})^\top q}(y) } \leq  \frac{ a\constker_2}{s^{1/4}}$

%
\item[(f)] $\ns{(\etaMat - \subetaMat )\etaMat^{-1} \SignVecPad} \leq a \constker \ns{\etaMat^{-1} \SignVecPad} \leq 2 a\constker$.
\end{itemize}
Then, provided that 
$$
\PP(E_\om^c) \leq \frac{ u}{u+6 s(B_0+B_2)} \qandq \EE[L_{01}(\om)^2 \bun_{\Eve^c}] \leq \frac{u}{4\bar B s^{3/2}},
$$
with $u = \min\{a \constker_i, t\}$ as before,
 we have
\begin{align*}
\PP_\Eve( (F')^c) \leq 
& 4(d+1)s \exp\pa{-\frac{mt^2}{16 s^{3/2}\Lu_{01}^2(3 + 2  t)}} \\
&+ 4s d \abs{\Sf_\text{grid}}\exp\pa{-\frac{m (a\constker_0)^2/8}{ s^{3/2}(\Lu_{01}^2 (B_{11} + 1) + \Lu_{01}^2 )}}\\
&
+s(3d+d^2) \abs{\Sn_\text{grid}} \exp\pa{-\frac{m (a\constker_2)^2/8}{ s^{3/2}(\Lu_2^2  B_{11} + \Lu_1^2 B_{22} + \Lu_{01}\Lu_2)}}\\
&+  32s \exp \pa{ - \frac{m 4 a^2 \constker^2}{  s\pa{ 32 L_1^2 + 68 a \constker  L_1 \Lu_{01} }}}.
\end{align*}
where the first bound is from Proposition \ref{prop:bernstein_etaMat}, the second and third are from  Proposition \ref{prop:bound_etafunc} and the final bound is due to Proposition \ref{prop:bound_R_vec}.

To bound \eqref{eq:tobound}, we first observe that if event  $G$ holds, then just as observed previously,  $\abs{\diff{r}{\SignVecPad^\top \beta_2}(y)} \leq 2a \constker_r$.
To bound $\abs{\SignVecPad^\top \beta_1(y)}$, observe that
\begin{align*}
\SignVecPad^\top \beta_1(y) &= \SignVecPad^\top(\etaMat^{-1} - \subetaMat^{-1}) (\subetaFunc_{X_0} -\etaFunc_{X_0}) + \SignVecPad^\top(\etaMat^{-1} - \subetaMat^{-1}) \etaFunc_{X_0} \\
&= \SignVecPad^\top(\etaMat^{-1} - \subetaMat^{-1}) (\subetaFunc_{X_0} -\etaFunc_{X_0}) + \SignVecPad^ \top  \etaMat^{-1}( \subetaMat - \etaMat) \subetaMat^{-1} \etaFunc_{X_0} \\
&= \SignVecPad^\top(\etaMat^{-1} - \subetaMat^{-1}) (\subetaFunc_{X_0} -\etaFunc_{X_0}) + \SignVecPad^ \top  \etaMat^{-1}( \subetaMat - \etaMat)( \subetaMat^{-1} - \etaMat^{-1}) \etaFunc_{X_0} 
+ \SignVecPad^ \top  \etaMat^{-1}( \subetaMat - \etaMat) \etaMat^{-1} \etaFunc_{X_0}
\end{align*}
Under event $F'$,
\begin{itemize}
\item $\abs{\SignVecPad^\top(\etaMat^{-1} - \subetaMat^{-1}) (\subetaFunc_{X_0} -\etaFunc_{X_0})} \leq  \sqrt{s}\norm{\etaMat^{-1} - \subetaMat^{-1}} \norm{ \subetaFunc_{X_0} -\etaFunc_{X_0}}  \leq t a \constker$
\item $\abs{\SignVecPad^ \top  \etaMat^{-1}( \subetaMat - \etaMat)( \subetaMat^{-1} - \etaMat^{-1}) \etaFunc_{X_0} }\leq \sqrt{s}\cdot 2 \cdot \norm{ \subetaMat - \etaMat}\norm{ \subetaMat^{-1} - \etaMat^{-1}} B_0 \leq 2 t^2 B_0$
\item $\ns{\etaMat^{-1}  ( \subetaMat - \etaMat)  \etaMat^{-1}\SignVecPad  } \leq \ns{\etaMat^{-1}} \ns{ ( \subetaMat - \etaMat)  \etaMat^{-1}\SignVecPad } \leq 4a \constker$. 
\end{itemize}
Finally, given any vector $q$ such that $\ns{q} \leq 4a \constker$, we have $\abs{q^\top \etaFunc_{X_0}} \leq 4a \constker B_0$. Therefore,
$$
\abs{\SignVecPad^\top \beta_1(y)} \leq ta + 2t^2 + 4a \constker B_0,
$$
and in a similar manner, we can show that the same upper bound holds for $
\norm{\diff{2}{\SignVecPad^\top \beta_1}(y)}
$.

Therefore,
\begin{equation}\label{eq:det_bd}
\norm{\diff{r}{\SignVecPad^\top \beta}(y)} \leq c \constker_r
\end{equation}
 if both $F'$ and $G$ hold, so conditional on $\Eve$, \eqref{eq:det_bd}   holds with probability at least $1-\delta$
provided that
$$
m\gtrsim s^{3/2} \cdot  \frac{(\Lu_2^2 B_{11} + \Lu_1^2 B_{22} + \bar B \Lu_{01}^2 +   \Lu_{01}\Lu_2)}{\constker^2} \cdot \log\pa{\frac{\bar N d s}{\rho} }  
$$
and
\begin{equation*}
\PP(E_\om^c)\lesssim \frac{\constker}{\bar B^{3/2} {s} \sqrt{\log(\bar N d/\rho)}} \qandq  ,\quad \EE[L_i(\om)L_j(\om)\bun_{E_\om^c}] \lesssim \frac{\constker}{{s}^{3/2} \sqrt{B} \sqrt{\log(\bar N d/\rho)}}
\end{equation*}

\textbf{Step II: Extending to the entire space}
To prove that $\subeta_{X_0}$ is nondegenerate on the entire space $\Xx$,
we  first show that $\subeta_{X_0}$ is locally Lipschitz (and hence determine how fine our grids ${\Sn_\text{grid}}$, $ {\Sf_\text{grid}}$ need to be): for $x,x'\in \Xx$ with $d_\met(x,x')\leq \rnear$,
\begin{align}
\norm{\diff{r}{\subeta_{X_0}}(x) -\diff{r}{\subeta_{X_0}}(x')} =&~  \Big\lVert\frac{1}{m}\sum_{k=1}^m \diff{r}{\rep{(\subetaMat_X^{-1}\SignVecPad)^\top \RFVec(\om_k)\phi_{\om_k}}}(\sig) \\
&\qquad - \diff{r}{\rep{(\subetaMat_X^{-1}\SignVecPad)^\top \RFVec(\om_k)\phi_{\om_k}}}(\sigg)\Big\rVert \notag  \\
=&~\norm{\frac{1}{m}\sum_{j=1}^m \rep{\pa{(\subetaMat_X^{-1}\SignVecPad)^\top \RFVec(\om_k)}\cdot \pa{\diff{r}{\phi_{\om_k}}(\sig) - \diff{r}{\phi_{\om_k}}(\sigg)}}} \notag \\
 \leq&~ \norm{\subetaMat_X^{-1}}\norm{\SignVecPad} \sqrt{s}\Lu_{01} \norm{\diff{r}{\phi_{\om_k}}(\sig) - \diff{r}{\phi_{\om_k}}(\sigg)} \\
\leq&~ 4s \Lu_{01} d_\met(x,x') \Ll_r \leq c \constker_r.
 \label{eq:eta_lipschitz}
\end{align}
where we have applied Lemma \ref{lem:features_lip}  to obtain the last line.

Choosing ${\Sf_\text{grid}}$ to be a $\delta_0 \eqdef \frac{c\constker_0}{4\Ll_0 \Lu_{01} s}$-covering of $\Sn$ (of size at most $\Oo(R_\Xx/\delta_0)$), ${\Sf_\text{grid}}$ to be a $\delta_2 \eqdef \frac{c\constker_2}{4\Ll_2 \Lu_{01} s}$-covering of $\Sf$ (of size at most $\Oo(R_\Xx/\delta_2)$). Then 
for any $x\in \Sn$ and $x'\in \Sn_\text{grid}$  such that $d_\met(x,x') \leq \delta_0$,
$$
\abs{\subeta_{X_0}(x)} \leq \abs{\subeta_{X_0}(x')} + \abs{\subeta_{X_0}(x) - \subeta_{X_0}(x')} \leq  1-  \constker_0 + 2c\constker_0.
$$
and given any $x\in \Sf$, let $x'\in \Sf_\text{grid}$ be such that $d_\met(x,x') \leq \delta_2$, so
$$
\rep{\overline{\sign(a_i)}\diff{2}{\subeta_{X_0}}(x)} \preceq \rep{ \overline{\sign(a_i)}\diff{2}{\subeta_{X_0}}(x') } + \norm{\diff{2}{\subeta_X}(x) -\diff{2}{\subeta_X}(x')} \Id \preceq (-\constker_2 + 2c\constker_2)\Id,
$$
and
$$
\norm{\imp{\overline{\sign(a_i) } \diff{2}{\subeta_{X_0}}(x)}} \leq \norm{\imp{ \overline{\sign(a_i)}\diff{2}{\subeta_{X_0}}(x')}} + c \constker_2 \leq (c_2+c)\constker_2.
$$

\end{proof}


\subsection{Nondegeneracy transfer to $\subeta_X$.}
We are now ready to prove Theorem \ref{thm:main_lip_eta}, which we restate below for clarity.
\begin{theorem}\label{thm:lip_eta}
 Under the assumptions of Theorem \ref{thm:nondegen_X0},
the following holds with probability at least $1-\rho$:
for all $X$ such that
\begin{equation}\label{eq:lip_eta_bound}
d_\met(X,X_0) \lesssim \min\pa{\rnear, \constker_r (C_\met B \sqrt{s})^{-1} ,\constker_r (C_\met \Lu_{12} \Lu_{r}\sqrt{s})^{-1}  },
\end{equation}
we have \begin{itemize}
\item[(i)] for all $y \in   \Xx^\text{far}$,  $\abs{\subeta_{X}(y)} \leq 1- \frac{13}{32}  \constker_0$
\item[(ii)] for all $y\in \Xx^\text{near}(i)$, $-\rep{\overline{\sign(a_i)} \diff{2}{\subeta_{X}}(y)} \succcurlyeq \frac{13\constker_2}{32} \Id$ and $\norm{\imp{\overline{\sign(a_i)} \diff{2}{\subeta_{X}}(y)}} \leq (\frac{p}{2}+\frac{3p}{16})\frac{1}{2} \constker_2 $.
\end{itemize}
Hence,   $\subeta_{X}$ is $(\frac{13}{32}\constker_0, \frac{13}{32}\constker_2)$-nondegenerate.

\end{theorem}
The proof essentially exploits the fact that $\subetaMat_X$, $\subetaFunc_X$ are locally Lipschitz  in $X$ with respect to the metric $d_\met$, and consequently nondegeneracy of $\subeta_{X_0}$ implies nondegeneracy of $\subeta_X$ whenever $d_\met(X,X_0)$ is sufficiently small.

\subsubsection{Proof of Theorem \ref{thm:lip_eta}}
We begin with  a lemma which shows that $\subetaMat_X$ is locally Lipschitz in $X$.

\begin{lemma}[Lipschitz bound of $\subetaMat_X$]\label{lem:lip_Gamma}
Let $X_0\in \Xx^s$ be $\Delta$-separated points. Assume that for all $i+j\leq 3$
\[
\PP(E_\om^c)\leq \frac{1}{1+16\sqrt{s} B_{ij}} ,\quad \EE[L_i(\om)L_j(\om)\bun_{E_\om^c}] \leq \frac{1}{16\sqrt{s}}
\]
for all $i,j=0,...,2$.
Let $\rho>0$ and
$$
m\gtrsim s (\Lu_2^2 B_{11} + \Lu_1^2 B_{22} + \Lu_{01}\Lu_2) \pa{
\log\pa{\frac{sd}{\rho}} + d \log\pa{s C_\met \max_{i=0}^3 \Lu_i }  }
$$
Then, conditional on event $\Eve$, with probability at least $1-\rho$,  the following hold:
\begin{itemize}
\item(i) for all $X$ such that $\dsep(x_i,x_{0,i})\leq \rnear$, we have
\[
\norm{\subetaMat_X - \subetaMat_{X_0}} \lesssim  \Cmetrictensor B d_\met(X,X_0) \, .
\]
\item(ii) for all $X$ such that $d_\met(X,X_0) \lesssim  \min\pa{r_\textup{near},\frac{1}{\Cmetrictensor B}}$, we have $\norm{\Id -\subetaMat_X} \leq \frac{3}{4}$ and $\norm{\metg_{X}^{-\frac12}\Gamma_X^*} \lesssim 1$.

\end{itemize}

\end{lemma}

\begin{proof}
By Lemma \ref{lem:conc_kernel_unif} and Lemma \ref{lem:conc_kernel_3_unif}, with probability at least $1-\rho$ conditonal on $\Eve$, for all  $(i,j) \in \{(0,0),(0,1),(1,1),(1,2) \}$ and all $x,y\in \Sn$, 
$$\norm{\Cov^{(ij)}(x,y)}\leq \norm{\fullCov^{(ij)}(x,y)} + \frac{1}{\sqrt{s}},
$$
note that this also holds for $\Cov^{(ji)}(x,y)$ since $\Cov^{(ij)}(x,y) = \overline{\Cov^{(ij)}(y,x)}$.

In particular, for all $x,x'$ such that $d_\met(x,x')\geq \Delta/4$, we have $\norm{\Cov^{(ij)}(x,x')} \leq \frac{2}{\sqrt{s}}$.
Take any $X$ such that $d_\met(x_i, x_{0,i})\leq r_\textup{near}$, we have that both $x_i, x_{0,i}$ are at least $\Delta/4$-separated from $x_j$ and $x_{0,j}$. Therefore, for $k,\ell \in \ens{0,1}$, using Lemma \ref{lem:lip_kernel}:
\begin{equation}\label{eq:c1}
\begin{split}
&\norm{\Cov^{(k\ell)}(x_i,x_j) - \Cov^{(k\ell)}(x_{i,0},x_{j,0})} \lesssim \frac{ \Cmetrictensor}{\sqrt{s}} \sqrt{d_\met(x_i ,x_{0,i})^2 + d_\met(x_j,x_{0,j})^2} \\
&\norm{\Cov^{(k\ell)}(x_i,x_i) - \Cov^{(k\ell)}(x_{i,0},x_{i,0})} \lesssim  \Cmetrictensor\pa{B_{k+1,\ell} + B_{k,\ell+1}}d_\met(x_i, x_{0,i})
\end{split}
\end{equation}
and therefore by Lemma \ref{lem:block_norm}:
\begin{align*}
\norm{\subetaMat_X - \subetaMat_{X_0}}^2 &\leq \sum_{i,j=1}^s \sum_{k,\ell=0}^1 \norm{\Cov^{(k\ell)}(x_i,x_j) - \Cov^{(k\ell)}(x_{0,i},x_{0,j})}^2 \\
&\leq 2 \sum_{i,j=1}^s \sum_{k,\ell=0}^1 \norm{\Cov^{(k\ell)}(x_i,x_j) - \Cov^{(k\ell)}(x_{0,i},x_{j})}^2
+
 \norm{\Cov^{(\ell k)}(x_j,x_{0,i}) - \Cov^{(\ell k)}(x_{0,j},x_{0,i})}^2
\\
&\lesssim \Cmetrictensor^2  \pa{\sum_{\substack{k,l\in \{0,1,2\}\\ k+\ell\leq 3}} B_{k\ell}}^2 \sum_{i} d_\met(x_i ,x_{0,i})^2 + \frac{1}{s}\sum_{j\neq i} d_\met(x_j, x_{0,j})^2
\end{align*}
which yields the desired result.

For the second statement, using Proposition \ref{prop:bernstein_etaMat},
$
\PP_\Eve( \norm{\subetaMat_{X_0} - \etaMat_{X_0}} > \frac{1}{8}) \leq \rho,
$
so conditional on $\Eve$, we have with probability $1-\rho$, $\norm{\subetaMat_X - \subetaMat_{X_0}}\leq \frac{1}{8}$ and the claim follows since $\norm{\Id - \etaMat_{X_0}} \leq \frac12$ (due to Lemma \ref{lem:additional_admissible}) implies that $\norm{\Id - \subetaMat_X} \leq \frac{3}{4}$ and
\[
\norm{\subetaMat_X} \leq 7/4 \qandq \norm{\metg_X^{-\frac12} \Gamma^*_X} = \sqrt{\norm{\subetaMat_X}} \lesssim \sqrt{7}/2.
\]

\end{proof}

\begin{proof}[Proof of Theorem \ref{thm:lip_eta}]

Since $\subeta_{X_0}$ is nondegenerate with probability at least $1-\rho$,
the conclusion follows if we prove that for all $x\in \Xf$ and all $y\in \Xn$,
\begin{equation}\label{eq:toshow_nondegen_trans}
\norm{\diff{2}{\subeta_X - \subeta_{X_0}}(x) } \leq  \constker_0/32 \qandq \norm{\diff{2}{\subeta_X - \subeta_{X_0}}(y) } \leq p \constker_2/32
\end{equation}
with probability at least $1-\rho$.
We first write
\begin{align*}
\subeta_{X}(y) -\subeta_{X_0}(y)
&= \subetaCoeff_X^\top (\subetaFunc_X- \subetaFunc_{X_0})  + (\subetaCoeff_X - \subetaCoeff_{X_0})^\top \subetaFunc_{X_0}(y).
\end{align*}

 Conditional on $\Eve$, with probability at least $1-\rho/2$,  we have by Lemma \ref{lem:lip_Gamma} (note that our assumptions imply the assumptions of Lemma \ref{lem:lip_Gamma}), $\norm{\etaMat_X - \etaMat_{X_0}} \lesssim C_\met B d_\met(X,X_0)$ and $\norm{\etaMat_X^{-1}} \leq 4$. So, $$
\norm{\diff{r}{(\subetaCoeff_X - \subetaCoeff_{X_0})^\top \subetaFunc_{X_0}}(y)} \leq \sqrt{s} \norm{\etaMat_X^{-1} - \etaMat_{X_0}^{-1}} \leq 8 \sqrt{s} \norm{\subetaMat_{X} -\subetaMat_{X_0}} \lesssim \sqrt{s} C_\met B d_\met(X,X_0).
$$
By Lemma \ref{lem:features_lip}, if $\Eve$ occurs, then 
$$
\norm{\diff{r}{\subetaCoeff_{X}^\top (\subetaFunc_X - \subetaFunc_{X_0})}(y)} \leq  C_r \norm{\subetaCoeff_{X}} d_\met(X,X_0) \leq C_r \norm{\subetaMat_X^{-1}}  \sqrt{s} d_\met(X,X_0) \leq 4 C_r \sqrt{s} d_\met(X,X_0),
$$
where $C_r \lesssim  (1+C_\met) \Lu_r \Lu_{12}$.
Finally, since $\PP(\Eve^c)\leq \rho/2$,
we have with probability at least $1-\rho$, for all $y\in \Xx$, \eqref{eq:toshow_nondegen_trans} holds provided that \eqref{eq:lip_eta_bound} holds. Combining with the nondegeneracy of $\subeta_{X_0}$, the conclusion follows with probability $1-2\rho$.

\end{proof}


\section{Supplementary results to the proof Theorem \ref{thm:main}} \label{sec:ift}

%
%
%

Recall that in the proof of Theorem \ref{thm:main},
we defined
 the function $f:\CC^s\times \Xx^s \times \RR_+ \times \CC^m$ by
$$
f(u,v) \eqdef \Gamma_X^* (\Phi_X a - \Phi_{X_0} a_0 - w) + \lambda \binom{\sign(a_0)}{0_{sd}}
$$
where $u = (a,X)$ and $v = (\la,w)$. This function $f$ is differentiable with
\begin{equation}\label{eq:ift_partialv}
\partial_v f(u,v) = \left( \binom{\sign(a_0)}{0_{sd}}, \; -\Gamma_X^* \right) \in \CC^{s(d+1) \times m},
\end{equation}
and $\partial_u f(u,v) $ is
\begin{equation}\label{eq:ift_partialu}
\Gamma^*_X \Gamma_X J_a + \begin{pmatrix}
0_{1 \times s} &A_{11} & 0 & \cdots & 0\\
0_{1 \times s} &0 &A_{12}& \cdots & 0\\
\vdots &\vdots &\vdots& \ddots & \vdots\\
0_{1 \times s} &0 & 0 & \cdots & A_{1s}\\
0_{d \times s} &A_{21} & 0 & \cdots & 0\\
0_{d \times s} &0 &A_{22}& \cdots & 0\\
\vdots &\vdots &\vdots& \ddots & \vdots\\
0_{d \times s} &0 & 0 & \cdots & A_{2s}\\
\end{pmatrix}
\end{equation} 
where 
$
A_{1j} \eqdef \nabla_x\ps{\varphi(x_j)}{z}^\top$, $A_{2j} \eqdef\nabla^2_x\ps{\varphi(x_j)}{z}$, $z \eqdef (\Phi_X a - \Phi_{X_0} a_0 - w)$ and $J_a \in \RR^{s(d+1)\times s(d+1)}$ is a the diagonal matrix:
$$
J_a = \begin{pmatrix}
\Id_{s\times s} &&& 0 \\
& a_1 \Id_{d\times d} &\\
& & \ddots &\\
0 & && a_s \Id_{d\times d}
\end{pmatrix}.
$$

Letting $u_0 = (a_0, X_0)$ and $v_0 = (0,0)$, $\partial_u f(u_0,v_0) = \Gamma_{X_0}^* \Gamma_{X_0} J_a$ is invertible and $f(u_0,v_0) = 0$. 
Hence, by the Implicit Function Theorem, there exists a neighbourhood $V$ of $v_0$ in $\CC \times \CC^m$, a neighbourhood  $U$ of $u_0$ in $\CC^s \times \Xx^s$ and a Fr\'echet differentiable function $g: V\to U$ such that for all $(u,v) \in U\times V$, $f(u,v) = 0$ if and only if $u = g(v)$.
To conclude,  we simply need to bound the size of the region on which $g$ is well defined, and to bound the error between $g(v)$ and $g(0)$. 
Let us first remark that our assumptions imply that $\PP(\Eve^c) \leq \rho/2$ and
\begin{equation}\label{eq:consq:stoc_lip}
\PP(E_\om^c)\leq \frac{1}{1+16\sqrt{s} B_{ij}} ,\quad \EE[L_i(\om)L_j(\om)\bun_{E_\om^c}] \leq \frac{1}{16\sqrt{s}},
\end{equation}
for all $i,j=0,...,2$.
 Therefore, it is sufficient to prove the existence of $g$ conditional on event $\Eve$:

 \begin{theorem}\label{thm:appli_IFT}
 
 Assume that for all $i+j\leq 3$
\[
\PP(E_\om^c)\leq \frac{1}{1+16\sqrt{s} B_{ij}} ,\quad \EE[L_i(\om)L_j(\om)\bun_{E_\om^c}] \leq \frac{1}{16\sqrt{s}}
\]
for all $i,j=0,...,2$.
Let $\rho>0$ and suppose that
$$
m\gtrsim s (\Lu_2^2 B_{11} + \Lu_1^2 B_{22} + \Lu_{01}\Lu_2) \pa{
\log\pa{\frac{sd}{\rho}} + d \log\pa{s C_\met \LL_3 }  }
$$
where  $\LL_r \eqdef \max_{i\leq r} L_r$.
Then,  conditional on event $\Eve$, with probability at least $1-\rho$: there exists a $\Cder{1}$ function $g$  such that, for all $v = (\lambda,w)$ such that $\norm{v} \leq r$ with $r$ satisfying 
\begin{equation}\label{eq:radius}
r= \order{ \frac{1}{\sqrt{s}}\min\pa{\tfrac{\min\{\rnear,(C_\met B)^{-1} \}}{ \min_i \abs{a_{0,i}}},~\tfrac{1}{\Lu_{01} \Lu_{12}(1+\norm{a_0})}, }}
\end{equation}
we have $f(g(v),v) = 0$ and $g(0) = u_0$.
Furthermore, given $(\la,w)$ in this ball,  $(a,X)\eqdef g((\la,w))$ satisfies
\begin{equation}\label{eq:error}
\norm{a-a_0} + d_\met(X,X_0) \leq \frac{\sqrt{s}(\lambda + \norm{w})}{\min_i \abs{a_{0,i}}}.
\end{equation}

\end{theorem}
 We begin with some preliminary results before presenting the proof of this theorem in Section \ref{sec:pf_ift}.

\subsection{Preliminary results}

\begin{theorem}[Quantitative implicit function theorem, adapted from \cite{2017-denoyelle-jafa}]\label{thm:quantitative_IFT}
Let $F:\Hh \times \Yy \to \CC^n$ be a differentiable mapping where $\Hh$ is a Hilbert space, $\Yy \subseteq \CC^s\times \RR^{sd}$, $n=s(d+1)$, $\norm{\cdot}$ be a norm on $\Hh$. For each $y\in \Yy$, suppose that there exists a positive definite matrix $\metg_y$, and let $d_G$ be the associated metric. Assume that $F(x_0,y_0) = 0$, and that for $x \in \Bb_{\norm{\cdot}}(x_0, r_1), y \in \Bb_{d_G}(y_0, r_2)$, $\partial_y F(x,y)$ is invertible and we have
\[
\norm{\metg_y^{-\frac12}\partial_x F(x,y)} \leq D_1 \qandq \norm{\metg_y^{\frac12} \partial_y F(x,y)^{-1} \metg_x^{\frac12}} \leq D_2\, .
\]
Then, defining $R = \min\pa{\frac{r_2}{D_1 D_2},r_1}$, there exists a unique Fr\'echet differentiable mapping $g:\Bb_{\norm{\cdot}}(x_0,R) \to \Bb_{d_G}(y_0,r_2)$ such that $g(x_0)=y_0$ and for all $x \in \Bb_{\norm{\cdot}}(x_0,R)$, $F(x,g(x)) = 0$,
and furthermore
\[
\d g(x) = -(\partial_y F(x,g(x)))^{-1} \partial_x F(x,g(x))
\]
and consequently $\norm{\metg_{g(x)}^{\frac12} \d g(x)}\leq D_1 D_2$.
\end{theorem}
\begin{proof}
%
Let $V^* = \cup_{V \in \Vv} V$, where $\Vv$ is the collection of all open sets $V \in \RR^m$ such that
\begin{enumerate}
\item $x_0 \in V$,
\item $V$ is star-shaped with respect to $x_0$,
\item $V \subset \Bb_{\norm{\cdot}}(x_0, r_1)$,
\item there exists a $\Cc^1$ function $g: V \to \Bb_{d_G}(y_0,r_2)$ such that $g(x_0) = y_0$ and $F(x,g(x)) = 0$ for all $x \in V$.
\end{enumerate}
Observe that $\Vv$ is non-empty by the (classical) Implicit Function Theorem. Moreover, $\Vv$ is stable by union: indeed, all conditions expect the last one are easy to check. Now, let $V,\tilde{V}\in \Vv$ and $g,\tilde{g}$ be corresponding functions. The set $\overline{V} = \ens{x \in V\cap \tilde{V},~g(x) = \tilde g(x)}$ is non-empty (it contains $x_0$), and closed in $V\cap \tilde{V}$. Moreover, it is open: for any $x\in \overline{V}$, by our assumptions $\partial_y F(x,g(x))$ is invertible and the Implicit Function theorem applies at $(x,g(x))$, and by the uniqueness of the mapping resulting from it we obtain an open set around $x$ in which $g$ and $\tilde g$ coincide. Hence $\overline{V}$ is both closed and open in $V\cap \tilde V$, and by the connectedness of it $\overline{V} = V\cap \tilde V$. Therefore, there exists a function $g'$ defined on $V \cup \tilde V$ that satisfies condition 4. above (it is defined as $g$ on $V$ and $\tilde g$ on $\tilde V$, which is well-posed for their intersection), and $\Vv$ is indeed stable by union.

Hence $V^* \in \Vv$, let $g^*$ be its corresponding function. It is unique by the arguments above, satisfies $F(x,g^*(x)) = 0$ and
\begin{align*}
\metg_{g^*(x)}^{\frac12} \d g^*(x) & =  -\metg_{g^*(x)}^{\frac12} (\partial_y F(x,g^*(x)) )^{-1}  \partial_x F(x,g^*(x)) \\
&= -(\metg_{g^*(x)}^{-\frac12} \partial_y F(x,g^*(x)) \metg_{g^*(x)}^{-\frac12})^{-1}   \metg_{g^*(x)}^{-\frac12} \partial_x F(x,g^*(x))
\end{align*}
for all $x \in V^*$. Note that by our assumptions $\norm{\metg_{g^*(x)}^{\frac12} \d g^*(x)} \leq D_1 D_2$.

We finish the proof by showing that $V^*$ contains a ball of radius $r_2/(D_1 D_2)$. Let $x \in \RR^m$ with $\norm{x} = 1$, $R_x = \sup \ens{R,~x_0 + Rx \in V^*}$, and $x^* = x_0 + R_x x \in \partial V^*$. Clearly $0 < R_x \leq r_1$ since $V^*$ is open, assume $R_x < r_1$. Our goal is to show that in that case $R_x \geq \frac{r_1}{D_1 D_2}$. Since $\d g^*$ is bounded, $g^*$ is uniformly continuous on $V^*$ and it can be extended on $\partial V^*$, and by continuity $F(x^*, g^*(x^*)) = 0$. By contradiction, if $g^*(x^*)\in\Bb_{d_G}(y_0, r_2)$, by our assumptions we can apply the Implicit Function Theorem at $(x^*, g^*(x^*))$, and therefore extend $g^*$ on an open set $V$ that is not included in $V^*$ such that $V\cup V^* \in \Vv$, which contradicts the maximality of $V^*$. Hence $d_G(g^*(x^*),y_0) = r_2$. Let $\gamma:[0,1]\to \Yy$ be defined by $\gamma(t) \eqdef g^*(x^* + t(x_0 - x^*))$, so $\gamma'(t) = \d g^*(\gamma(t)) (x_0-x^*)$. Then,
\begin{align*}
r_2 &= d_G(g^*(x^*), g^*(x_0))  \leq \sqrt{\int_0^1 \dotp{\metg_{g^*(\gamma(t))} \gamma'(t)}{\gamma'(t)} \mathrm{d}t}\\
&=  \sqrt{\int_0^1 \norm{\metg_{g^*(\gamma(t))}^{\frac12} \d g^*(\gamma(t)) (x_0-x^*)}^2 \mathrm{d}t} \leq D_1 D_2 R_x.
\end{align*}
\end{proof}

\begin{lemma}\label{lem:bound_Pi}
Asssume that event $\Eve$ occurs. Then, for all $X$ such that $\dsep(x_i,x_{0,i})\leq \rnear$,
$$\norm{\Pi_X \Gamma_{X_0} a }\lesssim \begin{cases}
\Lu_2 \norm{a}_1 \max_i d_\met(x_i , x_{0,i})^2 \\
\Lu_2 \norm{a}_\infty d_\met(X,X_{0})^2
\end{cases}
$$
\end{lemma}

\begin{proof}
Recall that $\Im(\Gamma_X) = \{\phi(x_i),J_\phi(x_i)\}_i$, and $\Pi_X$ is a projector on $\Im(\Gamma_X)^\perp$. Also note that for $\dsep(x_i,x_{0,i}) \leq \rnear$, we have $\norm{\met_{x_{0,i}}^{-\frac12}\met_{x_i}^\frac12}\lesssim 1$, and therefore under $\Eve$:
\[
\norm{\met_{x_{0,i}}^{-\frac12} \nabla^2 \phi_{\om_j}(x_i) \met_{x_{0,i}}^{-\frac12}} \lesssim \norm{\diff{2}{\phi_{\om_j}}(x_i)} \leq \Lu_2
\]
Let $\gamma_i:[0,1]\to \Xx$ be any piecewise smooth curve such that $\gamma_i(1) = x_{0,i}$ and $\gamma_i(0) = x_i$.
Then, by Taylor expanding $\phi(\gamma_i(t))$ about $t=0$, we obtain
$$
\varphi(x_{0,i}) = \varphi(x_i) + \dotp{\nabla\varphi(x_i)}{\gamma_i'(0)} + \int_0^1 \frac{1}{2} \dotp{\nabla^2\varphi(\gamma_i(t)) \gamma_i'(t)}{\gamma_i'(t)} \mathrm{d}t.
$$
Therefore,
\begin{align*}
&\Pi_X \Gamma_{X_0}  a =\Pi_X \pa{\sum_{i=1}^s a_i \varphi(x_{0,i})} = \Pi_X \pa{\sum_{i=1}^s \frac{a_i}{2} \int_0^1  \dotp{\nabla^2\varphi(\gamma_i(t)) \gamma_i'(t)}{\gamma_i'(t)} \mathrm{d}t}
\end{align*}
Taking the norm implies
$$
\norm{\Pi_X \Gamma_{X_0}  a} \leq \sum_{i=1}^s  \frac{\abs{a_i}}{2} \int_{0}^1 \Lu_2 \norm{\met_{\gamma_i(t)} \gamma_i'(t)}^2 \mathrm{d}t
$$
and taking the infimum over all paths $\gamma_i$ yields
$$
\norm{\Pi_X \Gamma_{X_0}  a} \leq \Lu_2 \sum_i \abs{a_i} d_\met(x_i,x_{0,i})^2.
$$

\end{proof}

\subsection{Proof of Theorem \ref{thm:appli_IFT}}\label{sec:pf_ift}

Our goal is to apply Theorem \ref{thm:quantitative_IFT}. Let $u = (a,X)$, $u_0 = (a_0, X_0)$, $v = (\la,w)$ and $v_0 = (0,0)$.
We must control $\norm{\metg_X^{-\frac12} \partial_v f(u,v)}$ and $\norm{\metg_X^{\frac12}\partial_u f(u,v)^{-1} \metg_X^{\frac12}}$ for $(u,v)$ sufficiently close to $(u_0,v_0)$. 
Using Lemma \ref{lem:lip_Gamma}, conditional on event $\Eve$, with probability $1-\rho$ we have
\[
\norm{\metg_X^{-\frac12} \partial_v f(u,v)} \leq \norm{\SignVecPad} + \norm{\metg_X^{-\frac12} \Gamma_X} \lesssim \sqrt{s}
\]

To control $\norm{\metg_X^{\frac12}\partial_u f(u,v)^{-1} \metg_X^{\frac12}}$, first observe  that
\begin{align*}
\metg_{X}^{-1/2} \partial_u f(u,v) \metg_{X}^{-1/2} &= \pa{ \metg_{X}^{-1/2}\Gamma_X^*\Gamma_X \metg_{X}^{-1/2} + M(u,v)}J_a
\end{align*}
where
\begin{equation}\label{eq:M}
M(u,v) \eqdef  \begin{pmatrix}
0_{1 \times s} &\frac{1}{a_1}\pa{\met_{x_{1}}^{-\frac12}\nabla[\ps{\varphi}{z}](x_1)}^\top & \cdots & 0\\
\vdots &\vdots & \ddots & \vdots\\
0_{1 \times s} &0 & \cdots & \frac{1}{a_s}\pa{\met_{x_{s}}^{-\frac12}\nabla[\ps{\varphi}{z}](x_s)}^\top\\
0_{d \times s} &\frac{1}{a_1} \met_{x_{1}}^{-\frac12}\nabla^2[\ps{\varphi}{z}](x_1) \met_{x_{0,1}}^{-\frac12} & \cdots & 0\\
\vdots &\vdots & \ddots & \vdots\\
0_{d \times s} & 0 & \cdots & \frac{1}{a_s} \met_{x_{s}}^{-\frac12}\nabla^2[\ps{\varphi}{z}](x_s)\met_{x_{0,s}}^{-\frac12}\\
\end{pmatrix},
\end{equation}
where $z = (\Phi_X a - \Phi_{X_0} a_0 - w)$.
Now, let us study the invertibility of $\metg_X^{-\frac12} \Gamma_X^*\Gamma_X \metg_X^{-\frac12} + M(u,v)$ and bound the norm of its inverse.

\begin{lemma}[Bound on $M(u,v)$]\label{lem:bound_Muv}
Let $u=(a,X)$, $v = (\la,w)$ and let $M(u,v)$ be as defined in \eqref{eq:M}. Assume that $\Eve$ occurs and given $\epsilon>0$, let $c_\epsilon \eqdef 
\frac{\epsilon\; \min_i \abs{a_{0,i}}}{2 \Lu_{12}}$.
Then, for all $X\in \Xx^s$, $a\in \RR^s$ and $w\in \CC^m$ such that 
\begin{align*}
\norm{a-a_0} \leq \frac{c_\epsilon}{3\Lu_0}, \quad
\norm{w} \leq c_\epsilon/3 \qandq
d_\met(X,X_0) \leq \min\pa{\rnear,    \frac{c_\epsilon}{3\Lu_1\norm{a_0}} },
\end{align*}
 we have
\[
\norm{M(u,v)} \leq \epsilon \qandq  \ns{M(u,v)} \leq \epsilon 
\]
\end{lemma}

\begin{proof}
First note that for $r\in\NN_0$,
\[
\norm{\diff{r}{\phi^\top z}(x_i)} \leq \frac{1}{\sqrt{m}} \sum_{j=1}^m \norm{z_j\diff{r}{\phi_{\om_j}}(x_i)} \leq L_r \norm{z}
\]
Now, for $\bar q=[q_1,\ldots,q_s,Q_1,\ldots,Q_s]\in \CC^{s(d+1)}$, where $q_i\in \CC$ and $Q_i\in \CC^d$, and $\norm{\bar q}=1$, we have
\begin{align*}
\norm{M(u,v)\bar q}^2 &= \sum_{i=1}^s \abs{\frac{1}{a_i}\pa{\met_{x_{i}}^{-\frac12} \nabla [\phi^\top z](x_i)}^\top Q_i}^2 + \norm{\frac{1}{a_i}\met_{x_{i}}^{-\frac12} \nabla^2 [\phi^\top z](x_i)\met_{x_{i}}^{-\frac12} Q_i}^2 \\
&\leq \frac{4}{\min_i \abs{a_{0,i}}^2} \norm{q}^2 \max_i \pa{\norm{\met_{x_{i}}^{-\frac12} \nabla [\phi^\top z](x_i)}^2 + \norm{\met_{x_{i}}^{-\frac12} \nabla^2 [\phi^\top z](x_i)\met_{x_{i}}^{-\frac12}}^2 }\\
&= \frac{4 }{\min_i \abs{a_{0,i}}^2} \max_i\pa{ \norm{\diff{1}{\phi^\top z}(x_i)}^2 + \norm{\diff{2}{\phi^\top z}(x_i)}^2} \\
&\leq \frac{4 }{\min_i \abs{a_{0,i}}^2}  (\Lu_1^2 + \Lu_2^2)\norm{z}^2
\end{align*}
where we have used the fact that $\min_i \abs{a_i} \geq \min_i \abs{a_{0,i}}/2$.
If $\ns{\bar q}=1$, then 
\begin{align*}
\ns{M(u,v)\bar q} &= \max_i \ens{ \abs{\pa{\met_{x_{i}}^{-\frac12} \nabla [\phi^\top z](x_i)}^\top Q_i} , \norm{\met_{x_{i}}^{-\frac12} \nabla [\phi^\top z](x_i)\met_{x_{i}}^{-\frac12} Q_i}^2 } \\
&\leq \max_i\ens{ \norm{\met_{x_{i}}^{-\frac12} \nabla [\phi^\top z](x_i)} , \norm{\met_{x_{i}}^{-\frac12} \nabla [\phi^\top z](x_i)\met_{x_{i}}^{-\frac12}}^2 }
\end{align*}
and the same bound holds.

Now it remains to bound $\norm{z}$. Writing $\phi(x)\eqdef \pa{\phi_{\om_k}(x)}_{k=1}^m$, we have
\begin{align*}
\norm{z} &= \norm{\sum_i (a_i \phi(x_i) - a_{0,i} \phi(x_{0,i})) - w} \\
&\leq \Lu_0 \norm{a-a_0} + \norm{a_0} \max_k \sqrt{\sum_i \abs{\phi_{\om_k}(x_i) - \phi_{\om_k}(x_{0,i})}^2 } + \norm{w}\\
&\leq  \Lu_0 \norm{a-a_0} + \norm{a_0} \Lu_1 d_\met(X, X_0) + \norm{w}
\end{align*}
where the last inequality follows from Lemma \ref{lem:features_lip}.

\end{proof}

The  bound on  $\norm{M(u,v)}$ from Lemma \ref{lem:bound_Muv} allows us to conclude that under event $\Eve$,  taking
\begin{equation}\label{eq:c}
c \eqdef 
\frac{ \min_i \abs{a_{0,i}}}{16   \Lu_{12}}
\end{equation}  for all $X\in \Xx^s$, $a\in \RR^s$ and $w\in \CC^m$ such that 
\begin{align*}
\norm{a-a_0} \leq \frac{c}{3\Lu_0}, \quad
\norm{w} \leq c/3 \qandq
d_\met(X,X_0) \leq \min\pa{\rnear,    \frac{c}{3\Lu_1\norm{a_0}} },
\end{align*}
 we have
$
\norm{M(u,v)} \leq \frac{1}{8}.
$
Combining this with Lemma \ref{lem:lip_Gamma} gives
 $$\norm{\Id - (\metg_{X}^{-\frac12}\Gamma_X^*\Gamma_X \metg_{X}^{-\frac12} + M(u,v))} \leq \norm{\Id - \metg_{X}^{-\frac12}\Gamma_X^*\Gamma_X \metg_{X}^{-\frac12}} + \norm{M(u,v)} < \frac{7}{8}$$ and therefore it is invertible and $$\norm{(\metg_{X}^{-\frac12}\Gamma_X^*\Gamma_X \metg_{X}^{-\frac12} + M(u,v))^{-1}} \leq \frac{1}{1-\norm{\Id-(\metg_{X}^{-\frac12} \Gamma_X^*\Gamma_X \metg_{X}^{-\frac12} + M(u,v))}} = \order{1}.
 $$
In this case, $\partial_u f(u,v)$ is invertible, and we have
\begin{align*}
 \norm{(\metg_{X}^{-\frac12}\partial_u f(u,v) \metg_{X}^{-\frac12})^{-1}} = \norm{J_a^{-1}(\metg_{X}^{-\frac12} \Gamma_X^*\Gamma_X \metg_{X}^{-\frac12} + M(u,v))^{-1}} \lesssim \frac{1}{\min_i\abs{a_{0,i}}}
\end{align*}
since $\norm{a-a_0}\lesssim \min_i\abs{a_{0,i}}$ by assumption.

Therefore we can apply Theorem \ref{thm:quantitative_IFT} with (recalling the definition of $c$ in \eqref{eq:c})
\[
r_1 = c ,~ D_1 = \order{\sqrt{s}}, ~r_2 =  \order{\min\pa{r_\textup{near},~\tfrac{c}{\Lu_1\norm{a_0}}, \tfrac{c}{\Lu_0},  \tfrac{1}{\Cmetrictensor B}}} ,~D_2 = \order{\tfrac{1}{\min_i \abs{a_{0,i}}}} 
\]
with $B= \sum_{i+j\leq 3} B_{ij}$,
we obtain that $g(v)$ is defined for $v \in V \eqdef \Bb_{\norm{\cdot}_2}\pa{0,r}$ with 
\[
r\eqdef\min\pa{\tfrac{r_2}{D_1D_2},r_1}
= \tfrac{r_2}{D_1 D_2}
= \order{ \min\pa{\tfrac{\rnear}{\sqrt{s} \min_i \abs{a_{0,i}}},~\tfrac{1}{\sqrt{s}\Lu_1 \Lu_{12}\norm{a_0}}, \tfrac{1}{\sqrt{s}\Lu_{12}\Lu_0},  \tfrac{1}{\sqrt{s}  \min_i \abs{a_{0,i}} \Cmetrictensor B}}} 
\]
such that $g$ is $\Cc^1$, $f(g(v),v) = 0$, $g(v_0) = u_0$, where we recall that $u_0 = (a_0, X_0)$ and $v_0 = (0,0)$.

Finally, from Theorem \ref{thm:quantitative_IFT} we also have that
\[
\norm{\metg_{X} \d g(v)} \leq D_1 D_2 \lesssim \frac{\sqrt{s}}{\min_i \abs{a_{0,i}}}
\]
and by defining $\gamma(t) = g(v_0 + t(v-v_0))$ for $t\in [0,1]$, we have the following error bound between $u = g(v)$ and $u_0 =  g(v_0)$:
\begin{align*}
d_G(u,u_0) &= \sqrt{\norm{a-a_0}_2^2 + d_\met(X,X_0)^2} \leq  \sqrt{\int_0^1 \dotp{\metg_{\gamma(t)} \gamma'(t)}{\gamma'(t)}\mathrm{d}t} \\
&= \sqrt{\int_0^1\dotp{\metg_{\gamma(t)} \d g(t v) v}{\d g( t v) v } \mathrm{d}t}\\
& \leq \frac{\sqrt{s}}{\min_i \abs{a_{0,i}}} \norm{v}.
\end{align*}

\newcommand{\epst}{\bar \epsilon_2}
\newcommand{\epso}{\bar \epsilon_0}
\newcommand{\km}{\kappa^{\infty}}
\newcommand{\hm}{\bar H^{\infty}}

\newcommand{\fq}{f}

\newcommand{\cc}{\bar C_{\fq }}
\newcommand{\cf}{C_{\fq }}

\section{Examples}\label{sec:app-examples}
\subsection{Fej\'er kernel}\label{sec:fejer}

Let $\fq \in \NN$ and $\Xx \in \TT^d$ the $d$-dimensional torus.
We consider  the Fej\'er kernel
$$
\fullCov(x,x') = \prod_{i=1}^d \kappa(x_i-x_i'),
$$
where $\kappa(x)\eqdef \pa{ \frac{\sin\pa{\pa{\tfrac{\fq }{2}+1} \pi x}}{\pa{\tfrac{\fq }{2}+1} \sin(\pi x)} }^4$,
with constant metric tensor \[
\met_x =  \cf \Id\qandq d_\met(x,x') = \cf^{-\frac12} \norm{x-x'}_2.
\]
where $\cf \eqdef -\kappa''(0) = \frac{\pi^2}{3}\fq(\fq+4) \sim f^2$.
 Note that $\fullCov^{(ij)} = \cf^{-(i+j)/2} \nabla_1^i \nabla_2^j \fullCov$ and since the metric is constant, we can set $C_\met \eqdef 0$.

\subsubsection{Discrete Fourier sampling} 

A random feature expansion associated with the Fej\'er kernel is obtained by choosing   $\Omega = \enscond{\om\in \ZZ^d}{\norm{\om}_\infty\leq \fq }$, $\phi_\om(x) \eqdef  e^{\mathrm{i} 2\pi \om^\top x}$, and  $\Lambda(\om) = \prod_{j=1}^d g(\omega_j)$ where $g(j) = \frac{1}{\fq } \sum_{k=\max(j-\fq ,-\fq )}^{\min(j+\fq ,\fq )}(1-\abs{k/\fq })(1-\abs{(j-k)/\fq })$. Note that this corresponds to sampling \textit{discrete} Fourier frequencies.
In this case, the derivatives of the  random features  are uniformly bounded with
 $\norm{\nabla^j \phi_\om(x)}= \norm{\om}^j = \Oo( \cf^{j/2} d^{j/2})$. So, we can set  $\Lu_i = \Oo(d^{i/2})$.

\subsubsection{Admissibility of the kernel} 

\begin{theorem}
Suppose that $\fq\geq 128$. Then,
$\fullCov$ is an admissible kernel with
 $\rnear = 1/(8\sqrt{2})$, $\constker_2= 0.941$, $\constker_0 = 0.00097$, $h= \Oo(d^{-1/2})$ and $\Delta = \Oo( d^{1/2}{s_{\max}^{1/4}})$, $B_{00} = B_{11}= B_{20} = \Oo(1)$, $B_{01} = \Oo(d^{1/2})$ and $B_{22} = \Oo(d)$.
\end{theorem}
The remainder of this section is dedicated to proving this theorem. The uniform bounds on $B_{ij}$ are due to Lemma \ref{lem:fejer_upbd} (uniform bounds), and the bound on $\Delta$ and $h$ are due to Lemma \ref{lem:far_bounds_fejer}. From Lemma \ref{lem:fejer_hessian}, we see that by setting $\rnear \eqdef \tfrac{1}{8\sqrt{2}}$, for all $\dsep(x,x')\leq \rnear$, $\fullCov^{(20)}(x,x') \prec -\constker_2 \Id$ with $\constker_2 = (1-6\rnear^2)(1-\rnear^2/(2-\rnear^2) -\rnear^2) \geq 0.941$. Finally, from Lemma \ref{lem:near_bound_fejer}, we have that for for all $\dsep(x,x')\geq \rnear$, $\abs{\fullCov} \leq 1- 1/(8^3\cdot 2)$, so we can set $\constker_0 \eqdef 0.00097$.

Before proving these lemmas,  we first summarise  in Section \ref{sec:fejer-props}  some key properties of the univariate Fej\'er kernel $\kappa$ when $\fq \geq 128$ which were derived in \cite{candes-towards2013}.

For notational convenience, write $t_i\eqdef x_i-x_i'$, $\kappa_i \eqdef \kappa(t_i)$, $\kappa_i'\ \eqdef \kappa'(t_i)$, and so on. Let
\[
\fullCov_i \eqdef \prod_{\substack{k=1\\k\neq i}}^d \kappa_k,\quad \fullCov_{ij} \eqdef \prod_{\substack{k=1 \\ k\neq i,j}}^d \kappa_k \qandq \quad \fullCov_{ij\ell} \eqdef \prod_{\substack{ k=1 \\ k\neq i,j,\ell}}^d \kappa_k.
\]
 With this, we have:
\begin{align*}
\partial_{1,i} \fullCov(\sig,\sigg) =&~ \kappa'_i \fullCov_i \\
\partial_{1,i}\partial_{2,i} \fullCov(\sig,\sigg) =&~ -\kappa''_i \fullCov_i, \qandq\forall i\neq j, \; \partial_{1,i} \partial_{2,j} \fullCov(\sig,\sigg) = -\kappa_i' \kappa_j' \fullCov_{ij}.
\end{align*}
Where convenient, we sometimes write $\fullCov(t)  = \fullCov(x-x') \eqdef \fullCov(x,x')$.

\subsubsection{Properties of $\kappa$}\label{sec:fejer-props}

From \cite[Equations (2.20)-(2.24) and (2.29)]{candes-towards2013},
for all $t \in [-1/2,1/2]$  and $\ell=0,1,2,3$:
\begin{equation}\label{eq:upp_bd_fejer}
\begin{split}
 1- \frac{\cf}{2}t^2 &\leq \kappa(t) \leq 1-  \frac{\cf}{2}t^2 + 8 \pa{\frac{1+2/\fq}{1+2/(2+\fq)}}^2  \cf^2 t^4 \leq  1-  \frac{\cf}{2}t^2 + 8  \cf^2 t^4 \\
 \abs{\kappa'(t)} &\leq \cf t,\quad \abs{\kappa''(t)} \leq \cf, \quad \abs{\kappa'''(t)} \leq 3 \pa{\frac{1+2/\fq}{1+2/(2+\fq)}}^2 \cf^2 t \leq 12 \cf^2 t \\
\kappa'' &\leq -\cf + \frac{3}{2} \pa{\frac{1+2/\fq}{1+2/(2+\fq)}}^2 \cf^2 t^2  \leq -\cf + 6 \cf^2 t^2.
\end{split}
\end{equation}

By \cite[Lemma 2.6]{candes-towards2013},
\begin{align*}
\abs{\kappa^{(\ell)}(t)} \leq \begin{cases}
 \frac{\pi^\ell H_\ell(t) }{(\fq+2)^{4-\ell} t^4}, &t\in [\frac{1}{2\fq}, \frac{\sqrt{2}}{\pi}]\\
 \frac{\pi^\ell H_\ell^\infty}{(\fq+2)^{4-\ell} t^4}, &t\in [\frac{\sqrt{2}}{\pi}, \frac{1}{2}),
\end{cases}
\end{align*}
where $H_0^\infty \eqdef 1$, $H_1^\infty \eqdef 4$, $H_2^\infty \eqdef 18$ and $H_3^\infty \eqdef 77$, and $H_\ell(t) \eqdef \alpha^4(t) \beta_\ell(t)$, with
$$
\alpha(t) \eqdef \frac{2}{\pi(1-\frac{\pi^2 t^2}{6})}, \quad \bar \beta(t) \eqdef  \frac{\alpha(t)}{\fq t} = \frac{2}{\fq t \pi(1-\pi^2t^2/6)}
$$
and
$\beta_0(t) \eqdef 1$, $\beta_1(t) \eqdef 2+2\bar \beta(t)$, $\beta_2 \eqdef 4+7\bar \beta(t) + 6\bar \beta(t)^2$ and $\beta_3(t) \eqdef 8+24\bar \beta + 30\bar \beta(t)^2 + 15 \bar \beta(t)^3$.
Let us first remark that $\bar \beta$ is decreasing on $I \eqdef [\frac{1}{2\fq}, \frac{\sqrt{2}}{\pi}]$, so $\abs{\bar \beta(t)} \leq \abs{\bar \beta(1/(2\fq))} \approx 1.2733$, and $a(t) \leq a(\sqrt{2}/\pi) = \frac{3}{\pi}$ on $I$. Therefore, on $I$, $H_0(t) \leq \frac{3}{\pi}$, $H_1(t) \leq 3.79$, $H_2(t) \leq 18.83$ and $H_3(t) \leq 98.26$, and we can conclude that on $[\frac{1}{2\fq}, \frac{1}{2})$, we have
$$
\abs{\kappa^{(\ell)}(t)} \leq \frac{\pi^\ell \hm_\ell}{(\fq+2)^{4-\ell} t^4}
$$
where $\hm_0 = 1$, $\hm_1 \eqdef 4$, $\hm_2 \eqdef 19$, $\hm_3 \eqdef 99$. 
Combining with \eqref{eq:upp_bd_fejer}, we have $\norm{\kappa^{(\ell)}}_\infty \leq \km_\ell$ where $\km_0 \eqdef 1$, $\km_2 \eqdef \cf$,
\begin{align*}
\km_1 \eqdef \sqrt{\cf} \max\pa{\frac{2\pi^4 }{(\frac12+\frac{1}{\fq})^3} \frac{f}{\sqrt{\cf}}, \frac{\sqrt{\cf}}{2f} } = \Oo(\sqrt{\cf}) \\
 \km_3 \eqdef (\cf)^{3/2} \max\pa{\frac{99\pi^3}{(\frac12+\frac{1}{\fq})} \pa{\frac{2\fq}{\sqrt{\cf}}}^4, \frac{6\sqrt{\cf}}{\fq}  } = \Oo((\cf)^{3/2}).
\end{align*}

Finally, given $p\in (0,1)$,
\begin{align*}
(\fq+2)^4 t^4 \geq (1+ p(\fq+2)^2 t^2)^2, \qquad \forall \; t \geq \frac{1}{\sqrt{(1-p)}(\fq+2) }.
\end{align*}
Choosing $p=\frac{1}{2}$ and using $(\fq+2)^2 =( \frac{3}{\pi^2}\cf + 4) \geq \frac{3}{\pi^2} \cf$, we have
\begin{equation}\label{eq:kappa_decay}
\abs{\kappa^{(\ell)}(t)} \leq  \frac{\km_\ell}{(1+ \frac{3}{2\pi^2} \cf t^2)^2}, \qquad \forall\; t^2\geq \frac{2\pi^2}{3 \cf },
\end{equation}

\subsubsection{Bounds in neighbourhood of $x'=x$}\label{sec:fejer_neigh}

\begin{lemma}\label{lem:fejer_hessian}
Suppose that $\cf \norm{t}_2^2 \leq c$ with $c>0$ such that
$$
\epsilon \eqdef   \pa{1-  6 c}\pa{1- \frac{c}{2-{c}}} - c > 0
 $$
 Then, 
 $\Cov^{02}(t) \preceq - \epsilon \Id$.
\end{lemma}
\begin{proof}

We need to show that $\la_{\min}(-\fullCov^{(02)}(t))\geq b$. Let  $q\in \RR^d$,
and note that
\begin{equation}\label{eq:fejer2nd}
\begin{split}
- \dotp{\nabla_2^2 K q}{q}
&= - \sum_i \pa{ q_i \kappa_i'' \fullCov_i - \kappa_i' \sum_{j\neq i} q_j \kappa_j' \fullCov_{ij}} q_i\\
&=  - \pa{\sum_i q_i^2 \kappa_i'' \fullCov_i - \sum_i q_i \kappa_i \sum_{j\neq i} q_j \kappa_j \fullCov_{ij}}
\\
&\geq  \norm{q}^2\pa{ - \max_i \ens{\kappa_i'' \fullCov_i} - \sum_j \abs{\kappa_j'}^2 }.
\end{split}
\end{equation}
We first consider $ \kappa_i'' \fullCov_i$:
\begin{align*}
 \kappa_i'' &\leq -\cf + 6\cf^2 t_i^2,
\\
\fullCov_i &\geq \prod_{j\neq i} \pa{1-\frac{\cf}{2}t_i^2} \geq 1- \frac{\cf}{2} \norm{t}_2^2 - \pa{\frac{\cf}{2} \norm{t}_2^2}^3 - \pa{\frac{\cf}{2} \norm{t}_2^2}^5-\cdots\\
&\qquad \geq 1- \frac{\cf\norm{t}_2^2}{2(1-\frac{\cf}{2}\norm{t}_2^2)}.
\end{align*}
and hence, 
\[
  \kappa_i''\fullCov_i  \leq \pa{-\cf +  6\cf^2 \norm{t}_2^2}\pa{1- \frac{\cf\norm{t}_2^2}{2(1-\frac{\cf}{2}\norm{t}_2^2)}}
 \]
For the second term,
\[
\sum_j \abs{\kappa_j'}^2 \leq \cf^2\norm{t}_2^2 .
\]
Therefore,
$$
\la_{\min}(-\fullCov^{(02)}(t)) \geq  \pa{1-  6\cf \norm{t}_2^2}\pa{1- \frac{\cf\norm{t}_2^2}{2(1-\frac{\cf}{2}\norm{t}_2^2)}} - \cf \norm{t}_2^2
$$

\end{proof}

\begin{lemma}\label{lem:near_bound_fejer}
Assume that $\frac{1}{8 \sqrt{\cf}} \geq   \norm{t}_2 $
Then,
\begin{align*}
K(t) \leq   1 -  \frac{\cf}{4} \norm{t}_2^2 +  16\cf^2\norm{t}_2^4.
\end{align*}
Consequently, for all
\begin{equation*}
0< c  \leq  \frac{1}{8\sqrt{2\cf} },
\end{equation*} and all $t$ such that $\norm{t}_2 \geq c$, 
\begin{align*}
\abs{K(t)} \leq   1 - \frac{\cf}{8} c^2.
\end{align*}
\end{lemma}
\begin{proof}

First note that
$$
\abs{\kappa(u)} \leq  1-  \frac{\cf}{2} u^2 + 32  \cf^2 u^4 = 1- u^2 g(u)
$$
where
$$
g(u) \eqdef  \cf\pa{ \frac{1}{2} - 32  \cf u^2 },
$$
and note that $g(u) \in (0, \tfrac{\cf}{2})$ for $u\in (0,1/(8\sqrt{\cf})$.
So, writing $t = (t_i)_{i=1}^d$ and $g_j \eqdef g(t_j)$, we have
\begin{align*}
&K(t) = \prod_{j=1}^d \kappa(t_i ) \leq \prod_{j=1}^d \pa{1- t_j^2\cdot  g(t_j)} 
\\
&= 1 -  \sum_{j=1}^d t_j^2 g_j + \sum_{j\neq k} t_j^2  t_k^2  g_j g_k - \sum_{j\neq k \neq \ell } t_j^2 t_k^2 t_\ell^2 g_j g_k  g_\ell + \cdots\\
\end{align*}
Note that
\begin{align*}
- &\sum_{j\neq k \neq \ell } t_j^2 t_k^2 t_\ell^2\cdot  g_j g_k  g_\ell + \sum_{j\neq k \neq \ell \neq n } t_j^2 t_k^2 t_\ell^2 t_n^2 \cdot  g_j g_k  g_\ell g_n\\
&\leq  - \sum_{j\neq k \neq \ell } t_j^2 t_k^2 t_\ell^2\cdot  g_j g_k  g_\ell + \pa{\sum_{j\neq k \neq \ell } t_j^2 t_k^2 t_\ell^2  \cdot  g_j g_k  g_\ell }\pa{\sum_{n} t_n^2 g_n}\\
&\leq - \sum_{j\neq k \neq \ell } t_j^2 t_k^2 t_\ell^2\cdot  g_j g_k  g_\ell \pa{ 1-  \frac{\cf}{2}\norm{t}_2^2} <0
\end{align*}
since $\pa{ 1-  \frac{\cf}{2}\norm{t}_2^2} >0$.
Also,
$$
 \sum_{j=1}^d t_j^2 g_j 
\leq  \frac{\cf}{2} \sum_{j=1}^d t_j^2 <1,
$$
by assumption.
So,
\begin{align*}
&K(t) \leq  1 -  \sum_{j=1}^d t_j^2 g_j + \sum_{j\neq k} t_j^2  t_k^2  g_j g_k \\
&\leq  1 -  \sum_{j=1}^d t_j^2 g_j + \frac{1}{2} \pa{\sum_{j} t_j^2    g_j}^2
\leq   1 - \frac{1}{2}  \sum_{j=1}^d t_j^2 g_j\\
&\leq 1 - \frac{\cf}{2}   \pa{ \frac{1}{2} \sum_{j=1}^d t_j^2 - 32  \cf  \sum_{j=1}^dt_j^4 }
\leq  1 -  \frac{\cf}{4} \norm{t}_2^2 +  16\cf^2\norm{t}_2^4.
\end{align*}
Finally, observe that the function
$$
q(z) \eqdef  \frac{\cf}{4} z^2 - 16 \cf^2  z^4
$$
is positive and increasing on the interval $[0, \frac{1}{8\sqrt{2\cf} }]$. So, for $t$ satisfing \begin{equation}
c \leq \norm{t}_2 \leq  \frac{1}{8\sqrt{2\cf} } ,
\end{equation} we have
$
\abs{K(t)} \leq 1- q(c) \leq 1- \frac{\cf}{8} c^2 .
$
Finally, since $\abs{K(t)}$ is decreasing as $t$ increases, we trivially have that $
\abs{K(t)} \leq 1- q(c )
$ for all $t$ with $\norm{t}_2\geq c$.

\end{proof}

\subsubsection{Bounds under separation}\label{sec:fejer_sep}

\begin{lemma}\label{lem:far_bounds_fejer}
Let $i,j\in \{0,1,2\}$ with $i+j\leq 3$.
Let $\bar A \geq \sqrt{\tfrac{4\pi^2}{3 }}$ and   $\norm{t}_2 \geq \bar A \sqrt{d} s_{\max}^{1/4}/\sqrt{\cf}$. Then, we have $\norm{\fullCov^{(ij)}(t) } \leq d^{\frac{i+j-4}{2}} (\bar A^4  s_{\max})^{-1}$.
\end{lemma}
\begin{proof}
Write $t = (t_j)_{j=1}^d$.
To bound $K(t) = \prod_{j=1}^d \kappa(a_j)$, we want to make use of the form \eqref{eq:kappa_decay}. We can do this for each $t_j$ such that $\abs{t_j} \geq \sqrt{\tfrac{2\pi^2}{3 \cf }}$. Note that there exists at least one such $t_j$ since $\norm{t}_\infty \geq \norm{t}_2/\sqrt{d}  \geq \bar A s_{\max}^{1/4}/\sqrt{\cf} \geq \sqrt{\tfrac{2\pi^2}{3 \cf }}$.
 If $\{\abs{t_j}\}_{j=1}^k \subset [0,  \sqrt{\tfrac{2\pi^2}{3 \cf }})$ for $k\leq d-1$, then
 $$
k {\frac{2\pi^2}{3 \cf }} + \sum_{j=k+1}^d t_j^2  \geq  \norm{t}_2^2   \geq \frac{\bar A^2 d s_{\max}^{1/2}}{\cf},
$$ 
which implies that $\sum_{j=k+1}^d t_j^2  \geq \frac{1}{\cf}\pa{ \bar A^2 d s_{\max}^{1/2} - \frac{2\pi^2(d-1)}{3} } \geq \frac{\bar A^2 d s_{\max}^{1/2}}{2\cf} $,
by our assumptions on $\bar A$. Therefore, we may assume that we have some $d\geq p\geq 1$ such that $\{b_j\}_{j=1}^p \subseteq \{t_j\}$ with $\abs{b_j} \geq \sqrt{\tfrac{2\pi^2}{3 \cf }}$ and $\norm{b}_2 \geq \frac{\bar A \sqrt{d} \sqrt[4]{ s_{\max}}}{\sqrt{2\cf} }$. 
Observe that
$$
\prod_{j=1}^p (1+ \frac{3\cf}{2\pi^2} b_j^2) \geq 1+ \frac{3\cf}{2\pi^2} \sum_{j=1}^p b_j^2 = 1+ \frac{3\cf}{2\pi^2} \norm{b}^2_2 \geq 1+  \frac{3}{4\pi^2} {\bar A^2 d \sqrt{ s_{\max}}}.
$$
So, by applying the fact that $\abs{\kappa} \leq 1$, $\km_0 = 1$ and \eqref{eq:kappa_decay}, we have
\begin{align*}
\abs{K(t)} &\leq \prod_{j=1}^p \abs{\kappa(b_j)} \leq  \prod_{j=1}^p \frac{1}{ \pa{1 + \frac{3\cf}{2\pi^2} b_j^2}^2} \leq \frac{1}{\pa{1+  \frac{3}{4\pi^2} {\bar A^2 d \sqrt{ s_{\max}}} }^2}.
\end{align*}

For $\abs{\kappa_i' \fullCov_i}$, if $i\not\in \enscond{j}{ \abs{t_j}> \sqrt{\tfrac{2\pi^2}{3 \cf }} }$, then
$$
\abs{\kappa_i' \fullCov_i} \leq  \norm{\kappa_i'}_\infty \prod_{j=1}^p \abs{\kappa(b_j)} \leq \frac{ \norm{\kappa_i'}_\infty}{\pa{1+  \frac{3}{4\pi^2} \bar A^2 d \sqrt{s_{\max}}}^2},
$$
and otherwise, we have 
$
\abs{\kappa_i' \fullCov_i} \leq  \abs{\kappa'(t_i)} \prod_{j\neq i} \abs{\kappa(b_j)} \leq \frac{\km_1}{\pa{1+  \frac{3}{4\pi^2} \bar A^2 d \sqrt{s_{\max}}}^2},
$
In a similar manner, writing $V\eqdef \pa{1+  \frac{3}{4\pi^2} \bar A^2 d \sqrt{s_{\max}}}^{-2}$, we can deduce that
\begin{align*}
&\abs{\kappa_i' \fullCov_i}  \leq \kappa_1^{\max} V,\qquad \abs{\kappa_i'' \fullCov_i} \leq \kappa_2^{\max} V, \qquad\abs{\kappa_i' \kappa_j' \fullCov_{ij}} ^2 \leq (\kappa_1^{\max})^2 V\\
&\abs{\kappa_i''' \fullCov_i} ^3 \leq \kappa_3^{\max} V, \qquad\abs{\kappa_i'' \kappa_j' \fullCov_{ij}} ^3 \leq \kappa_2^{\max} \kappa^{\max}_1 V, \qquad \abs{\kappa_i'\kappa_j'\kappa_\ell' \fullCov_{ij\ell}}  \leq (\kappa^{\max}_1)^3 V.
\end{align*}

Therefore,
$$
\norm{\fullCov^{(10)}} = \frac{1}{\sqrt{\cf}}
\norm{\nabla_1 K} \leq  \frac{1}{\sqrt{\cf}} \sqrt{\sum_{j=1}^d \abs{\kappa_j' \fullCov_j}^2} \leq \frac{\km_1}{\sqrt{\cf}}  V \sqrt{d} \lesssim \frac{1}{\bar A^4 d^{3/2} s_{\max}}.
$$

Using Gershgorin theorem, we have
\begin{align*}
\norm{\nabla^2_2 \fullCov(\sig,\sigg)} \leq&~ \max_{1\leq i\leq d} \ens{\abs{\kappa''_i \fullCov_i} + \abs{\kappa'_i}\sum_{j\neq i} \abs{\kappa'_j}\abs{\fullCov_{ij}}}
\end{align*}
and hence,
\begin{align*}
\norm{\fullCov^{(02)}} &=  \frac{1}{\cf}\norm{\nabla_2^2 K} \leq \frac{1}{\cf} \max_{i=1}^d \ens{ \abs{\kappa_i'' \fullCov_i} + \abs{\kappa_i'} \sum_{j\neq i} \abs{\kappa_j' \fullCov_{ij}}} \\
&\leq \frac{1}{\cf} V\pa{\kappa_2^{\max}  + (\kappa_1^{\max})^2 (d-1)  } \leq  \frac{\max\{\km_2, (\km_1)^2\} }{\cf} V d \lesssim \frac{1}{\bar A^4 d s_{\max}}.
\end{align*}
Note also that $\norm{\fullCov^{(11)}} = \norm{\fullCov^{(02)}}$.
Finally, since
\begin{align*}
\norm{\partial_{1,i} \nabla^2_2 \fullCov(\sig,\sigg)} \leq&~ \max\Bigg\lbrace \abs{\kappa'''_i \fullCov_i} + \abs{\kappa''_i}\sum_{j\neq i} \abs{\kappa'_j}\abs{\fullCov_{ij}},\\
&\qquad \max_{j\neq i}\ens{\abs{\kappa''_j\kappa'_i \fullCov_{ij}} + \abs{\kappa'_j\kappa''_i\fullCov_{ij}} + \abs{\kappa'_i}\abs{\kappa'_j} \sum_{l\neq i,j} \abs{\kappa'_l}\abs{\fullCov_{ij\ell}}}\Bigg\rbrace,
\end{align*}
we have
\begin{align*}
\norm{\fullCov^{(12)}} &= \frac{1}{\cf^{3/2}}\norm{\nabla_1 \nabla_2^2 K} \\
&\leq \frac{1}{\cf^{3/2}} \sqrt{d} V \max\pa{ \kappa_3^{\max}  + \kappa_2^{\max} \kappa_1^{\max} (d-1), 2\kappa_2^{\max} \km_1  + (d-1) (\km_1)^3  }
\\
&\leq d^{3/2} \max\{\km_3, \km_1\km_2, (\km_1)^3\} \frac{1}{\cf^{3/2}} V \lesssim \frac{1}{\bar A^4 d^{1/2} s_{\max}}
\end{align*}
\end{proof}

\subsubsection{Uniform bounds}

\begin{lemma}\label{lem:fejer_upbd}
If $\rnear \sim 1/\sqrt{\cf}$, then $B_0=\Oo(1)$, $B_{01} = \Oo(\sqrt{d})$, $B_{02} = B_{12} = B_{11}= \Oo(1)$ and $B_{22} = \Oo(d)$.
\end{lemma}
\begin{proof}
We have $\abs{K} \leq 1$, and 
$$
\norm{\nabla \fullCov}^2 \leq \sum_i \abs{\kappa_i}^2\abs{K_i}^2 \leq d (\km_1)^2 \lesssim {\cf} d,
$$
so $B_{01} = \Oo(\sqrt{d})$.

From \eqref{eq:fejer2nd},  for all $\norm{q}=1$, 
$$
\dotp{\nabla_2^2 K(t) q}{q} \leq \max_i \abs{\kappa_i''} \norm{q}_2^2 + \norm{q}_2^2 \sum_i \abs{\kappa_i}^2
\leq \cf + \cf^2 \norm{t}^2 = \Oo(\cf),
$$
for $\norm{t}\lesssim 1/\sqrt{\cf}$.
So, since $\rnear\leq 2/\sqrt{\cf}$, $\norm{\fullCov^{02}(t)}\leq 2\eqdef B_{02}$. The norm bound for $K^{11}$ is the same.

\begin{align*}
\norm{\fullCov^{(12)}} &= \sup_{\norm{q}=\norm{p}=1}\frac{1}{\cf^{3/2}} \Bigg( \sum_{k} \sum_{k\neq i} \partial_{1,i} \pa{\partial_{2,k}^2 K p_i q_k^2 
+ \partial_{1,i} \partial_{2,i}\partial_{2,k} K p_i q_i q_k}\\ 
&+ \sum_i \sum_k \sum_j \partial_{1,i} \partial_{2,j}\partial_{2,k} p_i p_j p_k + \sum_i \sum_{j\neq i} \partial_{1,i}\partial_{2,i} \partial_{2,j} K p_i q_i q_j + \sum_i \partial_{1,i}\partial_{2,j}^2 K p_i q_i^2 \Bigg)\\
=&\sup_{\norm{q}=\norm{p}=1}\frac{1}{\cf^{3/2}} \Bigg( \sum_{k} \sum_{k\neq i} \kappa_i' \kappa_k'' K_{ik} p_i q_k^2 
+ \kappa_i'' \kappa_k'  K_{ik} p_i q_i q_k\\ 
&+ \sum_i \sum_k \sum_j \kappa_i'\kappa_k'\kappa_j'  K_{ijk} p_i p_j p_k + \sum_i \sum_{j\neq i} \kappa_i'' \kappa_j'  K_{ij} p_i q_i q_j + \sum_i \kappa_i' \kappa_j''  K_{ij} p_i q_i^2 \Bigg)\\
&\leq \frac{1}{\cf^{3/2}}\Bigg( 3 \norm{\kappa''}_\infty \sqrt{ \sum_i \abs{\kappa_k'}^2} + \pa{{ \sum_i \abs{\kappa_k'}^2}}^{3/2} + \norm{\kappa'}_\infty \norm{\kappa''}_\infty\Bigg)\\
&\leq \frac{1}{\cf^{3/2}}\pa{ 3\cf^2 \norm{t} + \cf^3 \norm{t}^3 + \Oo( \cf^{3/2}) } = \Oo(1)
\end{align*}
for $\norm{t} \leq 1/\cf^{1/2}$.

We finally consider $K^{(22)}(x,x)$: for $\norm{p}=1$,
\begin{align*}
\sum_i \sum_k \sum_j \partial_{1,k}\partial_{1,i}\partial_{2,j}\partial_{2,i} \fullCov p_j p_k
&= \sum_i \sum_{k\neq i} \kappa_i'' \kappa_k'' p_j^2 K_{ik} + \sum_i \sum_{k\neq i} \kappa'''_i \kappa_k' p_i p_k K_{ik} \\
&+ \sum_{i}\sum_k \sum_j \kappa_i''\kappa_j'\kappa_k' K_{ijk} p_j p_k + \sum_i\sum_j \kappa_i''' \kappa_j' p_j p_iK_{ij} + \sum_i \kappa_i'''' p_i^2 K_i\\
&= \sum_i \sum_{k\neq i} \kappa_i'' \kappa_k'' p_j^2 K_{ik} + \sum_i \kappa_i'''' p_i^2 \\
&= d \Oo(\cf^2)
\end{align*}
since $\kappa'(0) = \kappa'''(0) = 0$ and $\abs{\kappa''(0)} = \Oo(\cf)$, $\abs{\kappa''''(0)} = \Oo(\cf^2)$. So, $B_{22} = \Oo(d)$.
\end{proof}


\subsection{The Gaussian kernel}\label{sec:gaussian}


We consider the Gaussian kernel $\fullCov(\sig,\sigg) = \exp\pa{-\frac12 \normmah{\Sigma^{-1}}{\sig-\sigg}^2}$ in $\RR^d$. Note that $\fullCov$ is translation invariant, so that $\met_x$ will be constant and equal to $-\nabla^2 \fullCov(x,x)$. For simplicity define $t = \sig-\sigg$, $\Cov_\Sigma(t) = \exp\pa{-\frac12 \normmah{\Sigma^{-1}}{t}^2}$ and for $u\in \RR$, $\kappa(u) = \exp\pa{-\frac12 u^2}$. Denote by $\ens{e_i}$ the canonical basis of $\RR^d$, and by $f_i = \Sigma^{-1} e_i$ the $i^{th}$ row of $\Sigma^{-1}$. We have the following:
\begin{align*}
\nabla \Cov_\Sigma(t)=&~-\Sigma^{-1} t \Cov_ \Sigma(t) \\
\nabla^2\Cov_\Sigma(t)=&~\pa{-\Sigma^{-1} + \Sigma^{-1}tt^\top\Sigma^{-1}}\Cov_\Sigma(t) \\
\partial_{1,i}\nabla^2\Cov_\Sigma(t)=&~\pa{\Sigma^{-1} t f_i^\top + f_i t^\top \Sigma^{-1} - (-\Sigma^{-1} + \Sigma^{-1} t t^\top \Sigma^{-1})(t^\top f_i)}\Cov_\Sigma(t)
\end{align*}

Hence we have $\met_x = - \nabla^2 \Cov_\Sigma(0) = \Sigma^{-1}$, and, defining $\dsep(x,x') = \normmah{\Sigma^{-1}}{x-x'} = \norm{\Sigma^{-\frac12}(x-x')}$, we have $\Cdcov = 1, \Cmetrictensor = 0$ (that is, the metric tensor of the kernel is constant, and $\dsep$ is defined as the corresponding normalized norm).

Then, we have
\begin{align*}
\norm{\fullCov^{(10)}(x,x')} &= \norm{\fullCov^{(01)}(x,x')} = \dsep(x,x') \kappa(\dsep(x,x')) \\
\norm{\fullCov^{(02)}(x,x')} &= \norm{\fullCov^{(11)}(x,x')} \leq (\dsep(x,x')^2+1)\kappa(\dsep(x,x')) \\
\fullCov^{(02)}(x,x') &\preccurlyeq (\dsep(x,x')^2 - 1)\kappa(\dsep(x,x')) \Id
\end{align*}
and for $q \in \RR^d$ with $\norm{q} =1$, since 
\[
\sum_i (\Sigma^\frac12 \nabla \phi_\om)_i q_i = \nabla \phi_\om^\top(\Sigma^\frac12 q) = \sum_i \partial_i \phi_\om (q^\top \Sigma^\frac12 e_i)
\]
we can write
\begin{align*}
\fullCov^{(12)}(x,x')q &= \sum_{i=1}^d (q^\top \Sigma^\frac12 e_i)\Sigma^{\frac12}\partial_{1,i}\nabla^2\Cov_\Sigma(t)\Sigma^\frac12
\end{align*}
Thus we examine each term in $\partial_{1,i} \nabla^2 \Cov_\Sigma$. We have
\begin{align*}
\sum_i (q^\top \Sigma^\frac12 e_i) \Sigma^\frac12 \Sigma^{-1} t f_i^\top\Sigma^\frac12 &= \Sigma^{-\frac12} t \pa{\sum_i q^\top \Sigma^\frac12 e_i e_i^\top\Sigma^{-\frac12}} = \Sigma^{-\frac12} t q^\top
\end{align*}
and similarly $\sum_i (q^\top \Sigma^\frac12 e_i)\Sigma^\frac12 f_i t^\top \Sigma^{-1}\Sigma^\frac12 = q t^\top \Sigma^{\frac12}$. Then
\[
\sum_i (q^\top \Sigma^\frac12 e_i)(t^\top \Sigma^{-1} e_i) \Sigma^\frac12 \Sigma^{-1} \Sigma^\frac12 = t^\top \Sigma^{-1} (\sum_i e_i e_i^\top) \Sigma^\frac12 q = (t^\top \Sigma^\frac12 q) \Id
\]
and similarly $\sum_i \sum_i (q^\top \Sigma^\frac12 e_i)(t^\top \Sigma^{-1} e_i) \Sigma^\frac12 \Sigma^{-1}tt^\top \Sigma^{-1} \Sigma^\frac12 = (t^\top \Sigma^\frac12 q) \Sigma^{-\frac12}tt^\top \Sigma^{-\frac12}$.

Hence at the end of the day
\[
\norm{\fullCov^{(12)}(x,x')} \leq (3\dsep(x,x') + \dsep(x,x')^3)\kappa(\dsep(x,x'))
\]
and this bound is automatically valid for $\fullCov^{(21)}$ as well.

Finally, note that
$$
\norm{\fullCov^{(22)}(x,x)} = \sup_{\norm{p}\leq 1} \dotp{\Sigma^{1/2} \nabla_2 \nabla_2\cdot \pa{{ \Sigma^{1/2} \fullCov^{(2,0)}(x,x) p  }}}{p}
$$
where $\nabla_2\cdot$ is the divergence operator on the 2nd variable, and one can show that $\norm{\fullCov^{(22)}(x,x)} = (d+1)$. 

We are then going to use the fact that for any $q\geq 1$ the function $f(r) = r^q e^{-\frac12 {r^2}}$ defined on $\RR_+$ is increasing on $[0,\sqrt{q}]$ and decreasing after, and its maximum value is $f(\sqrt{q}) = \pa{\frac{q}{e}}^{q/2}$. Furthermore, it is easy to see that we have $f(r) = r^q e^{-r^2/2} \leq \pa{\frac{2q}{2}}^{\frac{q}{2}} e^{-r^2/4}$ and therefore $f(r) \leq \varepsilon$ if $r \geq 2\pa{\log\pa{\frac{1}{\varepsilon}} + \frac{q}{2}\log\pa{\frac{2q}{e}}}$.

We define $r_\text{near}= 1/\sqrt{2}$ and $\Delta = C_1\sqrt{\log(s_{\max})} + C_2$ for some $C_1$ and $C_2$.
\begin{enumerate}
\item \emph{Global Bounds.} From what preceeds, we have
\begin{align*}
\norm{\fullCov^{(10)}} \leq \frac{1}{\sqrt{e}},\quad \norm{\fullCov^{(02)}} \leq \frac{2}{e} + 1,\quad \norm{\fullCov^{(12)}} \leq \frac{3}{\sqrt{e}} + \pa{\frac{3}{e}}^{\frac{3}{2}}
\end{align*}
and note that $\norm{\fullCov^{(11)}} = \norm{\fullCov^{(02)}}$,
so for all $i+j\leq 3$ $B_{ij} = \order{1}$.
\item \emph{Near $0$} For $\dsep(x,x') \leq r_\textup{near}$, we have
\[
\fullCov^{(02)} \preccurlyeq -\frac{e^{-\frac{1}{4}}}{2} \Id
\]
and for $\dsep(x,x') \geq \frac{1}{2}$,
\[
\abs{\fullCov} \leq e^{-\frac{1}{4}} = 1- (1-e^{-\frac{1}{4}})
\]
and $\norm{\fullCov^{(22)}(x,x)} = d+1$,
so we have also $\constker_i = \order{1}$, so $B_i = B_{0i} + B_{1i} +1 = \order{1}$ and $B_{22}=d+1$.
\item[3.] \emph{Separation.} Since $\constker_i = \order{1}$ and $B_{ij} = \order{1}$, every condition $\norm{\fullCov^{(ij)}} \lesssim \frac{1}{s_{\max}}$ is satisfied if $\Delta \geq C_1 \sqrt{\log(s_\text{max})} + C_2$ for some constant $C_1$ and $C_2$. 
\end{enumerate}


\subsubsection{Fourier measurements with Gaussian frequencies}

The random feature expansion for $\fullCov$ is $\phi_\om(x) = e^{i \om^\top x}$ and $\Lambda =\mathcal{N}(0,\Sigma^{-1})$. We have immediately $L_0 = 1$. For $j\geq 1$, we have $\diff{j}{\phi_\om}(x)[q_1,\ldots,q_j] = \pa{\prod_i \om^\top(\Sigma^\frac12 q_i)}\phi_\om(x)$ and therefore
\begin{align*}
&\norm{\diff{j}{\phi_\om}} \leq \normmah{\Sigma}{\om}^j 
\end{align*}

Now, we use $\normmah{\Sigma}{\om}^j = (\norm{\Sigma^\frac12 \om}^2)^\frac{j}{2} = W^{\frac{j}{2}}$ where $W$ is a $\chi^2$ variable with $d$ degrees of freedom. 
Then, we use the following Chernoff bound \citep{Dasgupta2003}: for $x\geq d$, we have
\[
\PP(W\geq x) \leq \pa{\frac{ex}{d}e^{-\frac{x}{d}}}^{\frac{d}{2}} \leq \pa{e \pa{\sqrt{\frac{x}{d}}}^2 e^{-\frac12 \cdot \pa{\sqrt{\frac{x}{d}}}^2} e^{-\frac{x}{2d}}}^{\frac{d}{2}}\leq 2^{\frac{d}{2}}e^{-\frac{x}{4}}
\]
by using $x^2e^{-\frac{x^2}{2}} \leq \frac{2}{e}$.

Hence we can define the $F_j$ such that, for all $t \geq d^{j/2}$, $\PP(L_j(\om) \geq t) \leq F_j(t) = 2^{\frac{d}{2}} \exp\pa{-\frac{t^\frac{2}{j}}{4}}$, and $F_j(\bar L_j)$ is smaller than some $\delta$ if $\bar L_j \propto \pa{d + \log\frac{1}{\delta}}^{\frac{j}{2}}$.
Then we must choose the $L_j$ such that $\int_{\bar L_j} t F_j(t) \d t$ is bounded by some $\delta$. Taking $L_j \geq d^{j/2}$ in any case, we have
\begin{align*}
\int_{\bar L_j} tF_j(t) \d t &= 2^{\frac{d}{2}} \int_{\bar L_j} t\exp\pa{-\frac{t^{\frac{2}{j}}}{4}} \d t = 2^{\frac{d}{2}} \int_{\bar L_j^{\frac{2}{j}}} (j/2) t^{j-1} \exp\pa{-\frac{t}{4}} \d t \\
& = 2^{\frac{d}{2}} (j/2) \int_{\bar L_j^{\frac{2}{j}}} \pa{t^{j-1}\exp\pa{-\frac{t}{8}}} \exp\pa{-\frac{t}{8}} \d t \leq 2^{\frac{d}{2}} (j/2) \pa{\frac{8(j-1)}{e}}^{j-1} \int_{\bar L_j^{\frac{2}{j}}} \exp\pa{-\frac{t}{8}} \d t \\
&= 2^{\frac{d}{2}} j \pa{\frac{8(j-1)}{e}}^{j-1} 8 \exp\pa{-\bar L_j^{\frac{2}{j}}/8}
\end{align*}
Hence this quantity is bounded by $\delta$ if
$
\bar L_j \propto \pa{d + \log\pa{\frac{1}{\delta}}}^{\frac{j}{2}}
$. 
Then we have $\bar L_j^2 F_i(\bar L_i) = \bar L_j^2 2^{\frac{d}{2}} \exp\pa{-\frac{\bar L_i^{\frac{2}{i}}}{4}}$ which is also bounded by $\delta$ if
$
\bar L_j \propto \pa{d + \pa{\log\frac{d}{\delta}}^2}^\frac{j}{2}
$. 
At the end of the day, our assumptions are satisfied for
\[
\bar L_j \propto \pa{d + \pa{\log\frac{d m}{\rho}}^{2}}^\frac{j}{2}
\]
\subsubsection{Gaussian mixture model learning}

We apply the mixture model framework with the base distribution:
\[
P_\theta = \mathcal{N}(\theta,\Sigma)
\]
The random features on the data space are $\phi'_\om(x) = C e^{i\om^\top x}$ with Gaussian distribution $\om \sim \Lambda = \mathcal{N}(0,A)$ for some constant $C$ and matrix $A$. Then, the features on the parameter space are $\phi_\om(\theta) = \EE_{x\sim P_\theta} \phi'_\om(x) = Ce^{i\om^\top \theta} e^{-\frac12 \normmah{\Sigma}{\om}^2}$ (that is, the characteristic function of Gaussians). Then, it is possible to show \citep{gribonval2017compressive} that the kernel is
\[
\fullCov(\theta,\theta') = C^2 \frac{\abs{A^{-1}}^\frac12}{\abs{2\Sigma + A^{-1}}^\frac12}e^{-\frac12 \normmah{(2\Sigma + A^{-1})^{-1}}{\theta-\theta'}^2}
\]
Hence we choose $A = c \Sigma^{-1}$, $C = (1+2c)^{\frac{d}{4}}$, and we come back to the previous case $\fullCov(\theta,\theta') = e^{-\frac12 \normmah{\tilde \Sigma^{-1}}{\theta-\theta'}^2}$ with covariance $\tilde \Sigma = (2+1/c)\Sigma$. Hence $\constker_i = \order{1}$, $B_{ij} = \order{1}$, $\dsep(\theta,\theta') = \normmah{\tilde \Sigma^{-1}}{\theta - \theta'} = \frac{1}{\sqrt{2+1/c}}\normmah{\Sigma^{-1}}{\theta - \theta'}$.

\paragraph{Admissible features.} Unlike the previous case, the features are directly bounded and Lipschitz. We have
\begin{align*}
\abs{\phi_\om(\theta)} &\leq C \eqdef L_0,\\
\norm{\diff{j}{\phi_\om(\theta)}} &= C\norm{\tilde\Sigma^\frac12 \om }^j e^{-\frac{\normmah{\Sigma}{\om}^2}{2}} = C\pa{2+1/c}^{\frac{j}{2}}\norm{\Sigma^\frac12 \om }^j e^{-\frac{\normmah{\Sigma}{\om}^2}{2}} \leq C \pa{2+1/c}^{\frac{j}{2}} \pa{\frac{j}{e}}^{\frac{j}{2}} \eqdef L_j 
\end{align*}
Hence all constants $L_j$ are in $\order{C(2+1/c)^{\frac{j}{2}}}$ by choosing $c = \frac{1}{d}$ they are in $\order{d^{\frac{j}{2}}}$.
\newcommand{\dk}{d_\kappa}
\newcommand{\sech}{\mathrm{sech}}
\newcommand{\tk}[1]{\tanh\pa{\frac{\d_{#1}}{2}}}

\subsection{The Laplace transform kernel}\label{sec:laplace}
Let $\alpha\in \RR_+^d$ and let $\Xx\subset \RR_+^d$ be a compact domain. Define for $x\in \Xx$ and $\om \in \RR_+^d$, 
$$
\phi_\om(x) \eqdef \exp(-\dotp{x}{\om}) \prod_{i=1}^d \sqrt{\frac{(x_i+\al_i)}{\al_i}}\qandq \Lambda(\om) \eqdef \exp(-\dotp{2\alpha}{\om})  \prod_{i=1}^d (2\alpha_i), 
$$
The associated kernel is $\fullCov(x,x') = \prod_{i=1}^d \kappa(x_i+\alpha_i,x_i'+\alpha_i)$ where $\kappa$ is the 1D Laplace kernel $$\kappa(u,v) \eqdef 2 \frac{\sqrt{uv}}{(u+v)}.
$$

A direct computation shows that
$\met_x\in \RR^{d\times d}$ is the diagonal matrix  with $(h_{x_i+\alpha_i})_{i=1}^d$ where $h_x \eqdef \partial_x\partial_{x'}\kappa(x,x) = (2x)^{-2}$. Note that  
\begin{equation}\label{lapl-dist}
\begin{split}
\dk(s,t) &= \int_{\min\{s,t\}}^{\max\{s,t\}} (2x+2\al)^{-1} \d x =  \abs{\log\pa{\frac{t+\al}{s+\al}}} \\
\end{split}
\end{equation}
 and so, $d_\met(x,x') = \sqrt{\sum_{i=1}^d \abs{\log\pa{\frac{x_i+\alpha_i}{x_i'+\alpha_i}}}^2}$.

We have the following results concerning the boundedness of $\norm{\diff{j}{\phi_\om}}$ and the admissiblity of $K$:
\begin{theorem}[Stochastic gradient bounds]\label{thm:laplace-stoc-bds}
Assume that the $\al_i$'s are all distinct. Then, $\Lu_0(\om) \leq \Lu_0 \eqdef \pa{1+ \frac{R_\Xx}{\min_i \alpha_i}}^d $ and for $j=1,2,3$,
\begin{align*}
\PP(L_j(\om)\geq t) 
&\leq F_j(t) \eqdef \sum_{i=1}^d \beta_i \exp\pa{-\alpha_i \pa{ \frac{1}{2(R_\Xx + \norm{\al}_\infty)}  \pa{\frac{t}{\Lu_0}}^{1/j} - \sqrt{d}}  } 
\end{align*}
and we have that $\sum_i F_j(\Lu_j)\leq \delta$ and $\Lu_j^2 \sum_i F_i(\Lu_i) + 2 \int_{\Lu_j}^\infty t F_j(t) \mathrm{d}t \leq \delta$ provided that
\[
\Lu_j \propto  \Lu_0 {(R_\Xx + \norm{\al}_\infty)^j \pa{\sqrt{d} +\max_i \frac{1}{\al_i} \log\pa{\frac{d \beta_i  \Lu_0 (R_\Xx + \norm{\al}_\infty)}{\delta \alpha_i}}  }^j}.
\]
where
 $\beta_i = \prod_{j\neq i} \frac{\alpha_j}{\alpha_j - \alpha_i}$. Note that $\alpha_i \sim d$ implies that $\Lu_0 \sim (1+R_\Xx/d)^d \sim e^{R_\Xx}$.

\end{theorem}

\begin{theorem}[Admissiblity of $K$]
The Laplace transform kernel $\fullCov$ is admissible with $\rnear = 0.2$,  $C_\met = 1.25$, $\constker_0 = 0.005$,   $\constker_2 = 1.52 $. For all  $i+j\leq 3$, $B_{ij} = \Oo(1)$, $B_{22} = \Oo(d)$, $\Delta = \Oo( d  + \log\pa{ d^{3/2} s_{\max}})$ and $h = \Oo(1)$.
\end{theorem}
The first result Theorem \ref{thm:laplace-stoc-bds} is proved in Section \ref{sec:laplace-stoc-grad} and the second result, Theorem \ref{thm:laplace_kernel_dec} is a direct consequence of Theorem \ref{thm:laplace_kernel_dec} and Lemma \ref{lem:laplace-metric-var} in Section \ref{sec:laplace-admiss}.

\subsubsection{Stochastic gradient bounds}\label{sec:laplace-stoc-grad}

\begin{proof}[Proof of Theorem \ref{thm:laplace-stoc-bds}]
Let
$
V \eqdef \pa{1-2 (x_i+\alpha_i) \om_i}_{i=1}^d \in \RR^d.
$
Then,
\begin{align*}
\norm{V} &= \sqrt{\sum_i (1-2(x_i+\alpha_i) \om_i)^2} \\
& \leq \sqrt{\sum_i 1+ 4(x_i+\alpha_i)^2 \om_i^2}
\leq \sqrt{d+4 (R_\Xx + \norm{\al}_\infty)^2\norm{w}^2}\\
& \leq \sqrt{d} + 2(R_\Xx + \norm{\al}_\infty) \norm{w}\\
\end{align*}
We have the following bounds:
%
%
\begin{align*}
\abs{\phi_\om(x)} &\leq \prod_{i=1}^d \sqrt{1+ \frac{x_i}{\alpha_i}} \leq \pa{1+ \frac{R_\Xx}{\min_i \alpha_i}}^d \eqdef \Lu_0,\\
\diff{1}{\phi_\om}(x) &= \phi_\om(x) V \implies \norm{\diff{1}{\phi_\om}(x)} \leq  \Lu_0  \norm{V}\\
\diff{2}{\phi_\om}(x) &= \phi_\om(x) (V V^\top - 2\Id) \implies \norm{\diff{2}{\phi_\om}(x)} \leq  \Lu_0  \min\{\norm{V}^2,2\}.
\end{align*}
and given $u,q\in \RR^d$,
\begin{align*}
\diff{3}{\phi_\om}(x)[q,q,u] = \phi_\om(x) \pa{\dotp{u}{V} \dotp{q}{V}^2 - 2\norm{q}^2 - 4 \dotp{u}{q}\dotp{q}{V} + 8 \sum_i q_i^2 u_i},
\end{align*}
so
$$
\norm{\diff{3}{\phi_\om}(x)} \leq \abs{\phi_\om(x)} \pa{\norm{V}^3 + 10 + 4 \norm{V}} \leq \Lu_0 5(\norm{V}^3 + 3),
$$
And therefore, in general,
\[
\norm{\diff{j}{\phi_\om}(x)} \leq L_j(\om) \eqdef \bar R_\Xx^{j+1} \pa{\sqrt{d} + \norm{\om}}^j
\]
\[
\norm{\diff{j}{\phi_\om}(x)} \lesssim L_j(\om) \eqdef \Lu_0 \pa{\sqrt{d} + 2(R_\Xx + \norm{\al}_\infty) \norm{w}}^j
\]

Assuming for simplicity that all $\alpha_j$ are distinct, we have \cite{Akkouchi2008}:
\begin{align*}
\PP(\norm{w} \geq t) &\leq \PP(\norm{\om}_1 \geq t) = \sum_{i=1}^{d} \beta_i e^{-\alpha_i t}
\end{align*}
where $\beta_i = \prod_{j\neq i} \frac{\alpha_j}{\alpha_j - \alpha_i}$, using the fact that $\norm{\om}_1$ is a sum of independent exponential random variable.

Hence, for all $1\leq j\leq 3$ and $t \geq d^{\frac{j}{2}}$ we have
\begin{align*}
\PP(L_j(\om)\geq t) &\leq \PP\pa{\norm{w} \geq \frac{1}{2(R_\Xx + \norm{\al}_\infty)}  \pa{\frac{t}{\Lu_0}}^{1/j} - \sqrt{d} }\\
&\leq F_j(t) \eqdef \sum_{i=1}^d \beta_i \exp\pa{-\alpha_i \pa{ \frac{1}{2(R_\Xx + \norm{\al}_\infty)}  \pa{\frac{t}{\Lu_0}}^{1/j} - \sqrt{d}}  } \leq \delta
\end{align*}
and $F_j(\Lu_j) \leq \delta$ if
\[
\Lu_j \geq \Lu_0 \pa{2^j(R_\Xx + \norm{\al}_\infty)^j \pa{\sqrt{d} +\max_i \frac{1}{\al_i} \log\pa{\frac{d \beta_i}{\delta}}  }^j }
\]
Next, in a similar manner to the Gaussian case, we compute
\begin{align*}
\int_{\Lu_j} t F_j(t) \mathrm{d}t &= \sum_{i=1}^d \beta_i \int_{\Lu_j} t \exp\pa{-\alpha_i \pa{ \frac{1}{2(R_\Xx + \norm{\al}_\infty)}  \pa{\frac{t}{\Lu_0}}^{1/j} - \sqrt{d}}  } \mathrm{d}t\\
&=\Lu_0^2 j \sum_{i=1}^d e^{  \al_i \sqrt{d}  }  \beta_i  \int_{(\Lu_j/\Lu_0)^{1/j}}   \exp\pa{  \frac{-\alpha_i u}{2(R_\Xx + \norm{\al}_\infty)}}  u^{2j-1} \mathrm{d}u\\
&\leq   \pa{\frac{(2j-1)4(R_\Xx + \norm{\al}_\infty)}{e \alpha_i }}^{2j-1}  \Lu_0^2 j \sum_{i=1}^d e^{  \al_i \sqrt{d}  }  \beta_i  \int_{(\Lu_j/\Lu_0)^{1/j}}   \exp\pa{  \frac{-\alpha_i u}{4(R_\Xx + \norm{\al}_\infty)}}   \mathrm{d}u\\
&\leq \pa{\frac{4(R_\Xx + \norm{\al}_\infty)}{\alpha_i }}^{2j} \pa{\frac{2j-1}{e  }}^{2j-1}  \Lu_0^2 j \sum_{i=1}^d e^{  \al_i \sqrt{d}  }  \beta_i    \exp\pa{  \frac{-\alpha_i (\Lu_j/\Lu_0)^{1/j}}{4(R_\Xx + \norm{\al}_\infty)}} \leq \delta   
\end{align*}
if for all $i=1,\ldots, d$,
\[
\frac{4(R_\Xx + \norm{\al}_\infty)}{\alpha_i} \pa{2j \log\pa{\frac{4(2j-1)(R_\Xx + \norm{\al}_\infty)}{e\alpha_i }}  + \log(  \Lu_0^2 j ) +   \al_i \sqrt{d}  + \log\pa{ \frac{d \beta_i}{\delta}}  }\leq     \pa{ \frac{\Lu_j}{\Lu_0}}^{1/j}
\]
that is,
$$
\Lu_j \gtrsim \Lu_0 \pa{2^j(R_\Xx + \norm{\al}_\infty)^j \pa{\sqrt{d} +\max_i \frac{1}{\al_i} \log\pa{\frac{d \beta_i}{\delta}}  }^j }.
$$
%
%
It remains to bound $\Lu_j F_\ell(\Lu_\ell)$ with $\ell,j\in \{0,1,2,3\}$: Let $\Lu_\ell \geq \Lu_0 M^\ell$ for some $M$ to  be determined.
Then,
\begin{align*}
\Lu_j F_\ell(\Lu_\ell) &\leq \Lu_0 M^j \sum_{i=1}^d \beta_i \exp\pa{  \frac{-\alpha_i}{2(R_\Xx + \norm{\al}_\infty)}  M + \alpha_i\sqrt{d}} \\
&=    \Lu_0  \sum_{i=1}^d \beta_i M^j \exp\pa{  \frac{-\alpha_i}{4(R_\Xx + \norm{\al}_\infty)}  M }  \exp\pa{  \frac{-\alpha_i}{4(R_\Xx + \norm{\al}_\infty)}  M } e^{ \alpha_i\sqrt{d}}\\
&\leq \Lu_0 e^{-j} \sum_{i=1}^d  \pa{\frac{{4j(R_\Xx + \norm{\al}_\infty)}}{{\alpha_i}}}^j  \beta_i   \exp\pa{  \frac{-\alpha_i}{4(R_\Xx + \norm{\al}_\infty)}  M } e^{ \alpha_i\sqrt{d}}\\
&\leq \Lu_0 e^{-3} \sum_{i=1}^d  \pa{\frac{{12(R_\Xx + \norm{\al}_\infty)}}{{\alpha_i}}}^3  \beta_i   \exp\pa{  \frac{-\alpha_i}{4(R_\Xx + \norm{\al}_\infty)}  M } e^{ \alpha_i\sqrt{d}} \leq \delta
\end{align*}
if for each $i=1,\ldots, d$
\[
M\geq 4(R_\Xx + \norm{\al}_\infty) \pa{\sqrt{d} + \max_i \frac{1}{\alpha_i}  \log\pa{\frac{\Lu_0 d \beta_i}{\delta e^3}  \pa{\frac{{12(R_\Xx + \norm{\al}_\infty)}}{{\alpha_i}}}^3  }     }.
\]
%
Therefore, similar to the Gaussian case,  the conclusion follows for $\Lu_0 = \pa{1+ \frac{R_\Xx}{\min_i \alpha_i}}^d $, and for $j=1,2,3$,
\[
\Lu_j \propto  \Lu_0 {(R_\Xx + \norm{\al}_\infty)^j \pa{\sqrt{d} +\max_i \frac{1}{\al_i} \log\pa{\frac{d \beta_i  \Lu_0 (R_\Xx + \norm{\al}_\infty)}{\delta \alpha_i}}  }^j}.
\]
\end{proof}


\subsubsection{Admissiblity of  the kernel}\label{sec:laplace-admiss}

\paragraph{Metric variation}

We  have the following lemma on the variation of the Fisher metric:
\begin{lemma}\label{lem:laplace-metric-var}
Suppose that $\dsep(x,x') \leq c$, then  $\norm{\Id - \met_x^{1/2} \met_{x'}} \leq (1+c e^c) \dsep(x,x')$ .
\end{lemma}
\begin{proof}
Note that $\abs{1- \abs{(x_i+\al_i)/(x_i'+\al_i)}} \leq \max\{ e^{\dk(x_i,x_i')} -1, 1- e^{-\dk(x_i,x_i')}\} \leq \dk(x_i,x_i')(1+c e^c) $ for all $\dk(x_i,x_i') \leq c$. Therefore,
\begin{align*}
\norm{\Id - \met_x \met_{x'}}^2= \sum_i \abs{1- \abs{(x_i+\al_i)/(x_i'+\al_i)}}^2 \leq (1+c e^c) \dsep(x,x')
\end{align*}
provided that $\dsep(x,x') \leq c$.

\end{proof}

\paragraph{Admissiblity of  the kernel}

The following theorem provides bounds for $\fullCov$ and its normalised derivatives. 
\begin{theorem}\label{thm:laplace_kernel_dec}

\begin{enumerate}

\item $
\abs{\fullCov(x,x')} \leq \min\ens{ 2^d e^{-\frac12 \dsep(x,x')}, \frac{8}{8+ \dsep(x,x')^2}}.
$
\item $\norm{\fullCov^{(10)}(x,x')} \leq  \min\{2\sqrt{d} \abs{\fullCov}, \sqrt{2}\}$.

\item $\norm{\fullCov^{(11)}} \leq \min\{9d \abs{\fullCov},8\}$

\item 
$\norm{\fullCov^{(20)}}\leq  \min\{10d \abs{\fullCov} , 8\}
$
and $\la_{\min}(-\fullCov^{(20)} ) \geq   \pa{ 2-12\dsep(x,x')^2 } \fullCov$.

\item 
$\norm{\fullCov^{(12)} } \leq \min\{66 \abs{\fullCov} d^{3/2}, 16\sqrt{d} + 49\}$ and $\norm{\fullCov^{(12)} (x,x')} \leq 34$ if $\dsep(x,x')\leq 1$.
\item $\norm{\fullCov^{(22)}} \leq 16 d + 9$.

\end{enumerate}
In particular, for $\dsep(x,x') \geq 2d \log(2) + 2\log\pa{\frac{52 d^{3/2} s_{\max}}{h}}$, we have
$\norm{\fullCov^{(ij)}(x,x')} \leq \frac{h}{s_{\max}}$.
\end{theorem}
To prove this result, we first present some bounds for the univariate Laplace kernel in Section \ref{sec:laplace1d} before applying these bounds in   Section \ref{sec:laplacepf}.

\subsubsection{1D Laplace kernel}\label{sec:laplace1d}

In the following $\kappa^{(ij)}(x,x') \eqdef h_x^{-i/2} h_{x'}^{-j/2} \partial_x^i \partial_{x'}^j \kappa(x,x')$.

\begin{lemma}\label{lem:laplace_1d}
We have
\begin{itemize}
\item[(i)] $\kappa(x,x')= \sech\pa{\frac{\dk(x,x')}{2}} \leq 2e^{-\frac12 \dk(x,x')}$,
\item[(ii)] $
\abs{{\kappa^{(10)}(x,x') }}  = 2\abs{\tanh\pa{ \frac{\dk(x,x')}{2} } {\kappa (x,x')} },
$ and 
 $\abs{\kappa^{(10)}}  \leq  2\abs{\kappa}$.

\item[(iii)] $\abs{\kappa^{(11)}} \leq 4\abs{\kappa}^3 + 4\abs{\kappa}$
\item[(iv)] $\abs{\kappa^{(20)}} \leq 6 \abs{\kappa}$ and $-\kappa^{(20)}  \geq 2\kappa(x,x') \pa{1- 2 \tanh\pa{\frac{\dk(x,x')}{2}}}$.
\item[(v)] $\abs{\kappa^{(12)}} \leq 49 \abs{\kappa}$.
\item[(vi)] $\kappa^{(22)}(x,x) =9$ for all $x$.
\end{itemize}

\end{lemma}


\begin{proof}
We first state the partial derivatives of $\kappa$:
\begin{align*}
& \kappa (x,x') = \frac{2\sqrt{x x'}}{x+x'},\\
&\partial_x \kappa(x,x') = \frac{x'(x'-x)}{\sqrt{x x'}(x+x')^2}\\
&\partial_{x} \partial_{x'} \kappa(x,x') = \frac{-x^2+ 6 x x' - (x')^2}{2\sqrt{x x'} (x+x')^3} \\
&\partial_x^2 \kappa(x,x')  = - \frac{(x')^2 \pa{(x+x')^2 + 4x(x'-x)}  }{2\pa{x x'}^{3/2} (x+x')^3}\\
&\qquad \qquad = - \frac{(x')^2   }{2\pa{x x'}^{3/2} (x+x')} - \frac{ 2x' (x'-x) }{\pa{x x'}^{1/2} (x+x')^3}\\
& \partial_x \partial_{x'}^2 \kappa(x,x') =  \frac{x^3 + 13 x^2 x' - 33 x (x')^2 + 3 (x')^3)}{4 x' (x x')^{1/2} (x + x')^4}\\
& \partial_x^2 \partial_{x'}^2 \kappa(x,x') = -\frac{3 x^4 + 60 x^3 x' - 270 x^2 (x')^2 + 60 x (x')^3 + 3 (x')^4}{8 x x' (x x')^{1/2} (x + x')^5}
\end{align*}

(i)
$$\kappa(x,x') = 2\pa{ \sqrt{\frac{x}{x'}} + \sqrt{\frac{x'}{x}}}^{-1} = \frac{2}{e^{-\frac{\dk(x,x')}{2}} + e^{\frac{\dk(x,x')}{2}}} = \frac{1}{\cosh(\frac{\dk(x,x')}{2})} \leq 2e^{-\frac12 \dk(x,x')},
$$

(ii)
We have, assuming that $x>x'$,
\begin{align*}
\kappa^{(10)}(x,x')&= 2x \partial_x \kappa(x,x') = 2\frac{x'-x}{x+x'}\kappa(x,x')\\
&= 2\pa{\frac{1}{\frac{x}{x'}+1} - \frac{1}{1+\frac{x'}{x}}}\kappa(x,x')\\
&= 2 \pa{\frac{1}{1+\exp(\dk(x,x'))} - \frac{1}{1+\exp(-\dk(x,x'))} }\\
&=  2 \pa{\frac{\exp(-\dk(x,x')) - \exp(\dk(x,x'))}{2+\exp(\dk(x,x')) + \exp(\dk(x,x'))} }\\
&= \frac{-2\sinh(\dk(x,x'))}{1+\cosh(\dk(x,x'))} \kappa(x,x')\\
&= -2\tanh(\dk(x,x')/2) \kappa(x,x'),
\end{align*}

(iii)
\begin{align*}
\kappa^{(11)} &= 4 x x'\partial_{x'} \partial_x \kappa(x,x') = 4x x' \frac{ 4 x x' -
\pa{x - x'}^2 }{2\sqrt{x x'} (x+x')^3}\\
&= 4\kappa(x,x')^3 - \frac{4(x-x')^2}{(x+x')^2} \kappa(x,x') \\
&=  \kappa(x,x') \pa{ 4\kappa(x,x')^2 -4 \tanh^2(\dk(x,x')/2) }
\end{align*}
so $\abs{\kappa^{(11)}} \leq 4\abs{\kappa}^3 + 4\abs{\kappa}$.

(iv)
\begin{align*}
\kappa^{(20)} &= 4x^2 \partial_x^2 \kappa(x,x') =  - \frac{4\pa{x x'}^{1/2} \pa{(x+x')^2 + 4 x(x'-x)}  }{2 (x+x')^3}\\
&=- 2\kappa(x,x') \pa{1 +  \frac{2 x (x'-x)  }{ (x+x')^2} }
\end{align*}
so $\abs{\kappa^{20}} \leq 6 \abs{\kappa}$. 
Also,
\begin{align*}
-\kappa^{(20)}  \geq 2\kappa(x,x')\pa{1 -  2\tanh(\dk(x,x')/2)  } 
\end{align*}

(v)
\begin{align*}
\kappa^{(12)} &= 2x(2x')^2\partial_x \partial_{x'}^2 \kappa(x,x') \\
&=\kappa(x,x')\pa{1+ \frac{2 v (5 u^2 - 18 u v + v^2)}{(u + v)^3}}
\end{align*}
so $\abs{\kappa^{(12)}} \leq 49 \abs{\kappa}$.

(vi) 
\begin{align*}
\kappa^{(22)} &= 16(x x')^2 \partial_x^2 \partial_{x'}^2 \kappa(x,x')\\
&=-3- \frac{48 x x' (x^2 - 6 x x' + (x')^2)}{(x + x')^4}
\end{align*}
and $\kappa^{(22)}(x,x) =9 $ .
\end{proof}


\subsubsection{Proof of Theorem \ref{thm:laplace_kernel_dec}}
\label{sec:laplacepf}

Let $\d_\ell \eqdef \dk(x_\ell+\alpha_\ell, x_\ell'+\alpha_\ell)$ and note that $d_\met(x,x') = \sqrt{\sum_\ell \d_\ell^2}$.
Define $g = \pa{  2\tanh(\frac{\d_\ell}{2}) }_{\ell=1}^d$.
We first prove that

\begin{enumerate}

\item[(i)] $
\abs{\fullCov(x,x')} \leq \prod_{\ell=1}^d \sech(\d_\ell/2) \leq \prod_{\ell=1}^d \frac{1}{1+\d_\ell^2/8}   \leq \frac{1}{1+\frac{1}{8}\dsep(x,x')^2}.
$
\item[(ii)] $\norm{\fullCov^{(10)}(x,x')} \leq  \norm{g}_2 \abs{\fullCov}$.

\item[(iii)] $\norm{\fullCov^{(11)}} \leq \abs{\fullCov} \pa{ \norm{g}^2_2 + 5 }$

\item[(iv)] 
$\norm{\fullCov^{(20)}}\leq  \abs{\fullCov} \pa{\norm{g}_2^2 +  6 }
$
and $\la_{\min}\pa{\fullCov^{(20)}} \geq \fullCov \pa{ 2- 3\norm{g}_2^2 }.
$

\item[(v)] 
$\norm{\fullCov^{(12)} } \leq \abs{\fullCov}\pa{\norm{g}_2^3 + 16 \norm{g}_2 + 49}$
\item[(vi)]$\norm{\fullCov^{(22)}} \leq 16 d + 9$.
\end{enumerate}

The result would then follow because
\begin{itemize}
\item 
$\sech(x) \leq 2 e^{-x}$ and $\sech(x) \leq (1+x^2/2)^{-1}$.
\item $\abs{\tanh(x)} \leq  \min\{x,1\}$, so $ \norm{g} \leq \min\{ \dsep(x,x'), 2\sqrt{d}\}$, 
\end{itemize}
For example,  $\norm{\fullCov^{(12)}} \leq \frac{1}{1+\frac{1}{8}\dsep(x,x')^2} \pa{\dsep(x,x')^3 + 16 \dsep(x,x') + 24} \leq 8 \dsep(x,x') + \frac{\sqrt{8}}{2} + 24\leq 34$
when $\dsep(x,x')\leq 1$.

In the following, we write  $\kappa_\ell^{(ij)} \eqdef \kappa^{(ij)}(x_\ell+\alpha_\ell, x_\ell'+\alpha_\ell)$ and $\kappa_\ell\eqdef \kappa_\ell^{(00)}$ and $\fullCov_i \eqdef \prod_{j\neq i} \kappa_j$. Moreover, we will make use of the inequalities for $\kappa^{(ij)}$ derived in Lemma \ref{lem:laplace_1d}.

(i)
We have
\[
\abs{\fullCov(x,x')} \leq \prod_{\ell=1}^d \sech(\d_\ell) \leq \prod_{\ell=1}^d  \pa{1+\frac{\d_\ell^2}{2}}^{-1} \leq \frac{1}{1+\dsep(x,x')^2}.
\]

(ii)
\begin{align*}
\fullCov^{(10)}(x,x') = \pa{\kappa_\ell^{(10)} \fullCov_{\ell}}_{\ell=1}^d \implies \norm{\fullCov^{(10)}(x,x')} \leq  \norm{g}_2 \abs{\fullCov}.
\end{align*}

(iii)
For $i\neq j$
\begin{align*}
\abs{\fullCov^{(11)}_{ij}} = \abs{\kappa^{(10)}_i \kappa^{(01)}_j \fullCov_{ij} } \leq 4 \tk{i}\tk{j} \abs{\fullCov},
\end{align*}
and $\abs{\fullCov^{(11)}_{ii}} =\abs{ \kappa^{(11)}_i \fullCov_i} \leq 5\abs{\fullCov}$. 
So, given $p\in \RR^d$ of unit norm,
\begin{align*}
\dotp{\fullCov^{(11)} p}{p}&= \sum_{i=1}^d \sum_{j\neq i} \kappa^{(10)}_i \kappa^{(01)}_j \fullCov_{ij} p_i p_j +  \sum_{i=1}^d p_i^2 \kappa^{(11)}_i \fullCov_i\\
&\leq  \abs{\fullCov} \pa{\sum_{i=1}^d \sum_{j\neq i} 4\tanh(\d_i/2) \tanh(\d_j/2) p_i p_j + 5 \sum_{i=1}^d p_i^2 }\\
&\leq \abs{\fullCov} \pa{ \norm{g}^2_2 + 5}
\end{align*}

(iv)
For $i\neq j$, $\fullCov^{(20)}_{ij} = \kappa^{(10)}_i \kappa^{(10)}_j \fullCov_{ij}$, and $\abs{\fullCov^{(20)}_{ii} } = \abs{\kappa^{(20)}_i \fullCov_i} \leq 6 \abs{\fullCov}$ and $-\fullCov^{(20)}_{ii} \geq 2 \fullCov\pa{1-2\tk{i}}$.
\begin{align*}
\dotp{\fullCov^{(20)} p}{p}&= \sum_{i=1}^d \sum_{j\neq i} \kappa^{(10)}_i \kappa^{(10)}_j \fullCov_{ij} p_i p_j +  \sum_{i=1}^d p_i^2 \kappa^{(20)}_i \fullCov_i
\\
&\leq  \abs{\fullCov} \pa{\sum_{i=1}^d \sum_{j\neq i} 4 \tanh(\d_i/2) \tanh(\d_j/2) p_i p_j +  6\sum_{i=1}^d p_i^2  }\\
&\leq \abs{\fullCov} \pa{\norm{g}_2^2 + 6},
\end{align*}
and
\begin{align*}
\dotp{-\fullCov^{(20)} p}{p} \geq \fullCov \pa{ 2-2 \norm{g}_\infty  -\norm{g}_2^2 }
\end{align*}

(v) For $i,j,\ell$ all distinct, 
\[
\fullCov^{(12)}_{ij\ell} = \kappa^{(10)}_i \kappa^{(01)}_j \kappa^{(01)}_{\ell} \fullCov_{ij\ell}
\leq 8 \tk{i}\tk{j}\tk{\ell} \fullCov,
\]
for all $i,\ell$,
\[
\fullCov^{(12)}_{ii\ell} = 8\kappa^{(11)}_i \kappa^{(01)}_{\ell} \fullCov_{i\ell}
\leq 10\tk{\ell} \fullCov
\]

\[
\fullCov^{(12)}_{iji} = \kappa^{(11)}_i \kappa^{(01)}_{j} \fullCov_{ij}
\leq 10 \tk{j} \fullCov,
\]
$\fullCov^{(12)}_{ijj} = \kappa^{(10)}_i \kappa^{(02)}_{\ell} \fullCov_{ij} \leq 12\tk{i}\fullCov$,
and
$\fullCov^{(12)}_{iii} =  \kappa^{(12)}_i  \fullCov_{i}\leq 26\fullCov$.
So, for $p,q\in \RR^d$ of unit norm,
\begin{align*}
&\sum_i  \sum_{j}\sum_\ell \fullCov^{(12)}_{ij\ell} p_j p_\ell q_i
=\sum_i \pa{ \sum_{j\neq i}\sum_\ell   \fullCov^{(12)}_{ij\ell} p_j p_\ell q_i +   \sum_\ell \fullCov^{(12)}_{ii\ell} p_i p_\ell q_i }\\
&= \sum_i \sum_{j\neq i} \pa{\sum_{\ell\not\in \{i,j\}}  \fullCov^{(12)}_{ij\ell} p_j p_\ell q_i +  \fullCov^{(12)}_{ij i} p_j p_i q_i +  \fullCov^{(12)}_{ij j} p_j^2 q_i } \\
& +  \sum_i \sum_{\ell\neq i} \fullCov^{(12)}_{ii\ell} p_i p_\ell q_i  +  \sum_i  \fullCov^{(12)}_{iii} p_i^2 q_i \\
&\leq \abs{\fullCov}\pa{\norm{g}_2^3 + 16 \norm{g}_2 + 49}.
\end{align*}

(vi)

\begin{align*}
\norm{K^{(22)}(x,x)} & = \sup_{\norm{p}=1} \EE[ \dotp{ \met_x^{-1/2} \nabla^2 \phi_\om(x) \met_x^{-1/2} p }{\met_x^{-1/2} \nabla^2 \phi_\om(x) \met_x^{-1/2} p} ]\\
&\leq \sup_{\norm{p}=1}  \sum_i \sum_{k\neq i} \kappa_i^{(11)} \kappa_k^{(11)} p_i^2  + \sum_i \sum_{k\neq i} \kappa_i^{(12)} \kappa_k^{(10)} p_i p_k + \sum_i \sum_{k\neq i}\sum_{j\not\in \{i,k\}} \kappa^{(11)}_i \kappa_k^{(10)} \kappa_j^{(01)}  p_k p_j \\
&+ \sum_i \sum_{j\neq i} \kappa^{(21)}_i \kappa^{(01)}_j p_j p_i + \sum_i \kappa_i^{(22)} p_i^2\\
&= \sup_{\norm{p}=1} \sum_i \sum_{k\neq i} \kappa_i^{(11)} \kappa_k^{(11)} p_i^2  
 + \sum_i \kappa_i^{(22)} p_i^2\\
 &\leq d \norm{\kappa^{(11)}}_\infty  +  \norm{\kappa^{(22)}}_\infty \leq 16 d +  \norm{\kappa^{(22)}}_\infty .
\end{align*}
since $\kappa^{(10)}(x,x) = \kappa^{(01)}(x,x) = 0$, and $\kappa^{(11)}(x,x) = 4$ from the proof of (iii) in Lemma \ref{lem:laplace_1d}.

\section{Tools}
\subsection{Probability tools}

\begin{lemma}[Bernstein's inequality (\cite{Sridharan2002}, Thm. 6)]\label{lem:bernstein}
Let $x_1,\ldots,x_n \in \mathbb{C}$ be $i.i.d.$ bounded random variables such that $\EE x_i=0$, $|x_i| \leq M$ and $Var(x_i) \eqdef \EE[\abs{x_i}^2] \leq \sigma^2$ for all $i$'s.

Then for all $t>0$ we have
\begin{equation}
\Xx\left(\frac{1}{n}\sum_{i=1}^n x_i \geq t\right)\leq 4\exp\left(-\frac{nt^2/4}{\sigma^2+ Mt/(3\sqrt{2})}\right).
\end{equation}
\end{lemma}


\begin{lemma}[Matrix Bernstein (\cite{Tropp2015}, Theorem 6.1.1)]\label{lem:bernstein_matrix}
Let $Y_1,...,Y_m \in \CC^{d_1,d_2}$ be complex random matrices with 
\[
\EE Y_j = 0,\quad \norm{Y_j} \leq L,\quad v(Y_j) := \max(\norm{\EE Y_jY_j^\adj}, \norm{\EE Y_j^\adj Y_j}) \leq M
\]
for each index $1\leq j\leq m$. Introduce the random matrix
\[
Z = \frac{1}{m} \sum_j Y_j.
\]
Then
\begin{equation}
\label{eq:matrix_bernstein}
\mathbb{P}\pa{\norm{Z}\geq t} \leq 2(d_1+d_2) e^{-\frac{mt^2/2}{M + Lt/3}}
\end{equation}
\end{lemma}

\begin{lemma}[Vector Bernstein for complex vectors \cite{minsker2017some}] \label{lem:vec-bernstein} 
Let $Y_1,\ldots, Y_M\in \CC^d$ be a sequence of independent random vectors such that $\EE[Y_i] = 0$, $\norm{Y_i}_2\leq K$ for $i=1,\ldots, M$ and set
$$
\sigma^2 \eqdef \sum_{i=1}^M \EE\norm{Y_i}_2^2.
$$
Then, for all $t\geq (K + \sqrt{K^2 + 36 \sigma^2})/M$,
$$
\PP\pa{\norm{\frac{1}{M} \sum_{i=1}^M Y_i}_2 \geq t} \leq 28 \exp \pa{ - \frac{M t^2/2}{\sigma^2/M + t K /3}}
$$
\end{lemma}


\begin{lemma}[Hoeffding's inequality (\citep{tang2013compressed}, Lemma G.1)]\label{lem:hoeffding_sign}
Let the components of $u\in C^k$ be drawn
 i.i.d. from a
symmetric distribution on the complex unit circle or $0$, consider a vector $w\in\CC^k$. Then, with probability at least $1-\rho$, we have
\begin{equation}
\PP\pa{\abs{\dotp{u}{ w}} \geq t}\leq 4e^{-\frac{t^2}{4\norm{w}^2}}
\end{equation}
\end{lemma}

\begin{lemma}\label{lem:hoeffding_mtx_sign} \cite[Theorem 4.1.1]{Tropp2015}
Let the components of $u\in \RR^k$ be a Rademacher sequence and let $Y_1,\ldots, Y_M\in \CC^{d\times d}$ be self-adjoint matrices. Set $\sigma^2 \eqdef \norm{\sum_{\ell=1}^M Y_\ell^2}$. Then, for $t>0$, 
\begin{equation}\label{eq:mtx_hoeff_sign}
\PP\pa{\norm{\sum_{\ell=1}^M u_\ell Y_\ell}\geq t} \leq 2d \exp\pa{-\frac{t^2}{2\sigma^2}}.
\end{equation}
\end{lemma}

We were only able to find a reference for this result in the case where $u$ is a Rademacher sequence, however, by the contraction principle (see \cite[Theorem 4.4]{ledoux2013probability}), a similar statement is true for Steinhaus sequences (we write only for the case of real symmetric matrices because this is all we require in this paper, but of course, the same argument extends to complex self-adjoint matrices):
\begin{corollary}\label{cor:mtx_hoeff}
Let the components of $u\in \CC^k$ i.i.d. from a
symmetric distribution on the complex unit circle or $0$ and let $B_1,\ldots, B_M\in \RR^{d\times d}$ be symmetric matrices. Set $\sigma^2 \eqdef \norm{\sum_{\ell=1}^M B_\ell^2}$. Then, for $t>0$, 
\begin{equation}\label{eq:mtx_hoeff}
\PP\pa{\norm{\sum_{\ell=1}^M u_\ell B_\ell}\geq t} \leq 4d \exp\pa{-\frac{t^2}{4\sigma^2}}.
\end{equation}
\end{corollary}
\begin{proof}
By the union bound,
\begin{align*}
\PP\pa{\norm{\sum_{\ell=1}^M u_\ell B_\ell}\geq t} 
&\leq 
\PP\pa{\norm{\sum_{\ell=1}^M \rep{u_\ell} B_\ell}\geq \frac{t}{\sqrt{2}}}  + \PP\pa{\norm{\sum_{\ell=1}^M \imp{u_\ell} B_\ell}\geq \frac{t}{\sqrt{2}}}.
\end{align*}
By the contraction principle \cite[Theorem 4.4]{ledoux2013probability}, 
$$
\PP\pa{\norm{\sum_{\ell=1}^M \rep{u_\ell}{ B_\ell}}\geq \frac{t}{\sqrt{2}}} \leq \PP\pa{\norm{\sum_{\ell=1}^M \xi_\ell { B_\ell}}\geq \frac{t}{\sqrt{2}}}
$$
where $\xi$ is a Rademacher sequence, and the same argument applies to the case of $\imp{u_\ell}$.
Therefore by Lemma \ref{lem:hoeffding_mtx_sign}, we have  $\PP\pa{\norm{\sum_{\ell=1}^M u_\ell B_\ell}\geq t} \leq 4d \exp\pa{-\frac{t^2}{4\sigma^2}}$.
\end{proof}

\subsection{Linear algebra tools}
%
%

The following simple lemma will be handy.
\begin{lemma}\label{lem:block_norm}
For $1\leq i,j \leq s$, take any scalars $a_{ij} \in \RR$, vectors $Q_{ij}, R_{ij} \in \RR^d$ and square matrices $A_{ij} \in \RR^{d\times d}$.

\begin{enumerate}
\item Let $M \in \RR^{sd\times sd}$ be a matrix formed by blocks :
\[
M = \left( \begin{matrix}
A_{11} & \ldots & A_{1s} \\
\vdots & \ddots & \vdots \\
A_{s1} & \ldots & A_{ss}
\end{matrix}
\right)
\]
Then we have
\begin{equation}
\label{eq:block_norm}
\normblock{M} = \sup_{\normblock{x} = 1}\normblock{Mx} \leq \max_{1\leq i\leq s} \sum_{j=1}^s \norm{A_{ij}}
\end{equation}
Now, let $P \in \RR^{sd\times s}$ be a rectangular matrix formed by stacking vectors $Q_{ij} \in \RR^{d}$:
\[
M = \left( \begin{matrix}
Q_{11} & \ldots & Q_{1s} \\
\vdots & \ddots & \vdots \\
Q_{s1} & \ldots & Q_{ss}
\end{matrix}
\right)
\]
Then,
\begin{equation}
\label{eq:block_norm_vec}
\norm{M}_{\infty \to \textup{block}}  \leq \max_{1\leq i\leq s} \sum_{j=1}^s \norm{Q_{ij}}_2, \quad \norm{M^\top}_{\textup{block} \to \infty}  \leq \max_{1\leq i\leq s} \sum_{j=1}^s \norm{Q_{ji}}_2
\end{equation}
\item Consider $A \in \RR^{s(d+1)\times s(d+1)}$ decomposed as
\[
M = \left(\begin{matrix}
a_{11} & \ldots & a_{1s} & Q^\top_{11} & \ldots & Q^\top_{1s} \\
\vdots & \ddots & \vdots &\vdots & \ddots & \vdots \\
a_{s1} & \ldots & a_{ss} & Q^\top_{s1} & \ldots & Q^\top_{ss} \\
R_{11} & \ldots & R_{1s} & A_{11} & \ldots & A_{1s} \\
\vdots & \ddots & \vdots &\vdots & \ddots & \vdots \\
R_{s1} & \ldots & R_{ss} & A_{s1} & \ldots & A_{ss}
\end{matrix}\right)
\]
Then,
\begin{align*}
\norm{M} &\leq \sqrt{\sum_{i,j} a_{ij}^2 + \norm{Q_{ij}}^2 + \norm{R_{ij}}^2 + \norm{A_{ij}}^2},\\
 \normblockB{M} &\leq \max_i \{\sum_{j}{ \abs{a_{ij}} + \norm{Q_{ij}}}, \;  \sum_{j}{ \norm{R_{ij}} + \norm{A_{ij}}}\}
\end{align*}

\end{enumerate}
\end{lemma}
\begin{proof}
The proof is simple linear algebra.
\begin{enumerate}
\item Let $x$ be a vector with $\normblock{x} \leq 1$ decomposed into blocks $x = [x_1,\ldots,x_s]$ with $x_i \in \RR^d$, we have
\begin{align*}
\normblock{Mx}^2 &= \max_{1\leq i\leq s} \norm{\sum_{j=1}^s A_{ij} x_j} \leq \max_{i}\sum_j \norm{A_{ij}} \norm{x_j} \leq \max_{i}\sum_j \norm{A_{ij}}
\end{align*}
\item Similarly,
\[
\norm{M^\top x}_\infty= \max_{1\leq i\leq s} \norm{\sum_{j=1}^s Q_{ji}^\top x_j} \leq \max_{i}\sum_j \norm{Q_{ji}} \norm{x_j} \leq \max_{i}\sum_j \norm{Q_{ji}}
\]
Then, taking $x\in \RR^s$ such that $\norm{x}_\infty \leq 1$, we have
\[
\normblock{M x} = \max_{1\leq i\leq s} \norm{\sum_{j=1}^s x_j Q_{ij}} \leq \max_{i}\sum_j \norm{Q_{ij}}
\]
\item Taking $x=[x_1,\ldots,x_s,X_1,\ldots,X_s] \in \RR^{s(d+1)}$ with $\norm{x} = 1$, we have
\begin{align*}
\norm{Mx}^2 &= \sum_{i=1}^s \pa{\sum_{j=1}^s a_{ij}x_j + Q_{ij}^\top X_j}^2 + \norm{\sum_{j=1}^s R_{ij}x_j + A_{ij} X_j}^2 \\
&\leq \sum_{i=1}^s \pa{\norm{x}\sqrt{\sum_{j=1}^s a_{ij}^2 + \norm{Q_{ij}}^2}}^2 + \pa{\norm{x}\sqrt{\sum_{j=1}^s \norm{R_{ij}}^2 + \norm{A_{ij}}^2}}^2 \\
&\leq \sum_{i,j} a_{ij}^2 + \norm{Q_{ij}}^2 + \norm{R_{ij}}^2 + \norm{A_{ij}}^2
\end{align*}
Now, if $\normblockB{x} = 1$, we have
\begin{align*}
\normblockB{Mx} &= \max_i \pa{\abs{\sum_{j=1}^s a_{ij}x_j + Q_{ij}^\top X_j} , \; \norm{\sum_{j=1}^s R_{ij}x_j + A_{ij} X_j} }\\
&\leq \max_i \pa{\sum_{j=1}^s \abs{a_{ij}} +\norm{ Q_{ij}} ,\; \sum_{j=1}^s \norm{R_{ij}x_j + A_{ij} X_j} } 
\end{align*}
\end{enumerate}
\end{proof}

\end{document}